\documentclass[a4paper,11pt]{article}
\usepackage[margin=1in]{geometry}
\usepackage{etoolbox}
\newtoggle{anonymous}
\togglefalse{anonymous}

\usepackage{multirow}%

\usepackage{dsfont}
\usepackage{mathrsfs}
\usepackage{amssymb,amsmath,amsfonts,amsthm}
\usepackage[parfill]{parskip} 
\usepackage{enumitem} 
\setlist[enumerate,1]{label={(\roman*)}}
\usepackage{url} 
\usepackage{xspace}
\usepackage[most]{tcolorbox}

\usepackage[ruled,vlined]{algorithm2e}

\usepackage{changepage}
\usepackage{thmtools}

\usepackage{graphicx}
\usepackage{subcaption}

\usepackage{tikz}
\usetikzlibrary{calc}
\usetikzlibrary{positioning, backgrounds, fit}
\usetikzlibrary{decorations.pathmorphing}
\usetikzlibrary{intersections}

\usepackage{xcolor} 

\definecolor{FigRed}{RGB}{046,037,133}
\definecolor{FigBlue}{RGB}{051,117,056}
\definecolor{FigGreen}{RGB}{148,203,236}

\definecolor{linkGreen}{RGB}{175,0,0}

\definecolor{linkblue}{RGB}{0, 102, 204} 
\usepackage[
colorlinks=true,
linkcolor=linkblue,     
citecolor=linkGreen,     
urlcolor=linkblue       
]{hyperref}

\usepackage{dirtytalk} 
\usepackage{mathtools} 

\usepackage[capitalize,nameinlink]{cleveref}

\Crefname{figure}{Figure}{Figures}
\Crefname{claim}{Claim}{Claims}
\Crefname{algorithm}{Algorithm}{Algorithms}

\crefformat{equation}{\textup{#2(#1)#3}}
\crefrangeformat{equation}{\textup{#3(#1)#4--#5(#2)#6}}
\crefmultiformat{equation}{\textup{#2(#1)#3}}{ and \textup{#2(#1)#3}}
{, \textup{#2(#1)#3}}{, and \textup{#2(#1)#3}}
\crefrangemultiformat{equation}{\textup{#3(#1)#4--#5(#2)#6}}%
{ and \textup{#3(#1)#4--#5(#2)#6}}{, \textup{#3(#1)#4--#5(#2)#6}}{, and \textup{#3(#1)#4--#5(#2)#6}}

\Crefformat{equation}{Equation~\textup{#2(#1)#3}}
\Crefrangeformat{equation}{Equations~\textup{#3(#1)#4--#5(#2)#6}}
\Crefmultiformat{equation}{Equations~\textup{#2(#1)#3}}{ and \textup{#2(#1)#3}}
{, \textup{#2(#1)#3}}{, and \textup{#2(#1)#3}}
\Crefrangemultiformat{equation}{Equations~\textup{#3(#1)#4--#5(#2)#6}}%
{ and \textup{#3(#1)#4--#5(#2)#6}}{, \textup{#3(#1)#4--#5(#2)#6}}{, and \textup{#3(#1)#4--#5(#2)#6}}

\newtheorem{theorem}{Theorem}[section]

\newtheorem{proposition}[theorem]{Proposition}

\newtheorem{corollary}[theorem]{Corollary}
\newtheorem{observation}[theorem]{Observation}
\newtheorem{conjecture}[theorem]{Conjecture}

\newtheorem{lemma}[theorem]{Lemma}

\theoremstyle{definition}
\newtheorem{definition}[theorem]{Definition}

\theoremstyle{remark}
\newtheorem{remark}[theorem]{Remark}
\newtheorem{claim}[theorem]{Claim}

\newenvironment{poc}{\begin{proof}[Proof of {C}laim.]}{\end{proof}}

\newcommand{\whp}{w.h.p.\xspace}
\newcommand{\pr}{\ensuremath{\mathbb{P}}}

\newcommand{\OF}{D} 
\newcommand{\Flood}{R}

\renewcommand{\Pr}[1]{\mathbb{P}\!\left[\,#1\,\right]}

\newcommand{\Ex}[1]{\mathbb{E} \left[\,#1\,\right]}

\newcommand{\geo}[1]{\operatorname{Geo}\!\left(#1\right)}

\newcommand{\ber}[1]{\operatorname{Ber}\!\left(#1\right)}

\newcommand{\cA}{\ensuremath{\mathcal{A}}}
\newcommand{\cB}{\ensuremath{\mathcal{B}}}

\newcommand{\cD}{\ensuremath{\mathcal{D}}}

\newcommand{\cG}{\ensuremath{\mathcal{G}}}
\newcommand{\cH}{\ensuremath{\mathcal{H}}}

\newcommand{\cS}{\ensuremath{\mathcal{S}}}
\newcommand{\cT}{\ensuremath{\mathcal{T}}}

\newcommand{\cX}{\ensuremath{\mathcal{X}}}

\newcommand{\bE}{\ensuremath{\mathbb{E}}}

\newcommand{\bN}{\ensuremath{\mathbb{N}}}

\DeclareMathOperator{\texp}{\text{\textsc{Texp}}}
\DeclareMathOperator{\Dexp}{\text{\textsc{Oexp}}}
\DeclareMathOperator{\unif}{\mathsf{unif}}

\usepackage{soul}
\newcommand{\SuggestChange}[2]{{\color{red} \relax\ifmmode\text{\st{$#1$}}\else \st{#1}\fi}\ {\color{blue} #2}}

\definecolor{samplingboxcolor}{HTML}{FCF5E5} 
\definecolor{titleboxcolor}{HTML}{EAEAEA}   
\definecolor{bordercolor}{rgb}{0.7, 0.7, 0.7} 

\newtcolorbox[auto counter, number within=section]{boxsampling}[1][]{
	colback=samplingboxcolor,  
	colframe=bordercolor,      
	coltitle=black,            
	colbacktitle=titleboxcolor,
	fonttitle=\bfseries\large, 
	title=#1,                  
	rounded corners,           
	boxrule=0.75pt,            
	enhanced                   
}

\newtcolorbox[auto counter, number within=section]{boxproblem}[1][]{
	colback=samplingboxcolor,  
	colframe=bordercolor,      
	coltitle=black,            
	colbacktitle=titleboxcolor,
	fonttitle=\scshape\large, 
	title=#1,                  
	rounded corners,           
	boxrule=0.75pt,            
	enhanced                   
}

\usepackage{thm-restate}

\newcommand{\rst}{RST\xspace}

\usepackage{enumitem,amssymb}
\newlist{todolist}{itemize}{2}
\setlist[todolist]{label=$\square$}
\usepackage{pifont}

\allowdisplaybreaks
\usepackage{orcidlink}

\title{Temporal Exploration of Random Spanning Tree Models}

 \iftoggle{anonymous}{
	\author{Anonymous author(s)}
}{%
\author{Samuel Baguley\thanks{Hasso Plattner Institute, University of Potsdam, Germany, \texttt{samuel.baguley@hpi.de}, \orcidlink{0000-0003-1090-0267}
} \and Andreas G{\"o}bel\thanks{Hasso Plattner Institute, University of Potsdam, Germany, \texttt{andreas.goebel@hpi.de}, 
\orcidlink{0000-0002-5180-7205}
}  \and Nicolas Klodt\thanks{Hasso Plattner Institute, University of Potsdam, Germany, \texttt{nicolas.klodt@hpi.de}, 
\orcidlink{0000-0002-5406-7441}
} \and George Skretas\thanks{Hasso Plattner Institute, University of Potsdam, Germany, \texttt{georgios.skretas@hpi.de}, \orcidlink{0000-0003-2514-8004}
} \and    John Sylvester\thanks{Department of Computer Science, University of Liverpool, UK, \texttt{john.sylvester@liverpool.ac.uk}, \orcidlink{0000-0002-6543-2934}
}
	\and
	Viktor Zamaraev\thanks{Department of Computer Science, University of Liverpool, UK, \texttt{viktor.zamaraev@liverpool.ac.uk}, \orcidlink{0000-0001-5755-4141}
 }}
}
\date{}

\usepackage{url}

\begin{document}

\maketitle
\thispagestyle{empty}
\setcounter{page}{0}

\begin{abstract}
The \emph{Temporal Graph Exploration problem} (TEXP) takes as input a \emph{temporal graph}, i.e., a sequence of graphs $(G_i)_{i\in \bN}$ 
on the same vertex set, and asks for a walk of shortest length visiting all vertices, where the $i$-th step uses an edge from $G_i$.   
If each such $G_i$ is connected, then an exploration of length $n^2$ exists, and this is known to be the best possible up to a constant. More fine-grained lower and upper bounds have been obtained for restricted temporal graph classes, however, for several fundamental classes, a large gap persists between known bounds, and it remains unclear which properties of a temporal graph make it inherently difficult to explore.

Motivated by this limited understanding  and the central role of the Temporal Graph Exploration
problem in temporal graph theory, we study the problem in a \emph{randomised setting}. We introduce the \emph{Random Spanning Tree} (RST) model, which consists of a set of $n$-vertex trees together with an arbitrary probability distribution $\mu$ over this set. 
A random temporal graph generated by the RST model is a sequence of independent samples drawn from $\mu$.

We initiate a systematic study of the Temporal Graph Exploration problem in such
random temporal graphs and establish tight general bounds on exploration time.
Our first main result proves that any RST model can, with high probability\footnote{In this paper ``with high probability'' refers to $1-o(1)$, where $o(1)$ is with respect to $n$.} (\whp), be explored in $O(n^{3/2})$ time, and we show that this bound is tight up to a constant factor.  
This demonstrates a fundamental difference between the adversarial and random settings.
Our second main result shows that if all trees of an RST are subgraphs of a fixed graph with $m$ edges then, \whp, it can be explored in $O(m)$ time.

\end{abstract}

\newpage

\section{Introduction}
Many real-world networks, such as transportation systems, social interactions, and communication infrastructures, are inherently dynamic: their topologies evolve over time as connections appear or disappear \cite{New18}. Temporal graphs (also known as evolving or time-varying graphs) offer a natural and expressive framework for modelling such networks \cite{HS19}. A temporal graph is a finite or infinite sequence of static graphs over the same vertex set, where each static graph, called a snapshot, represents the state of the network at a particular time step. Unlike static graphs, temporal graphs can capture crucial aspects of reachability and interaction patterns that depend on the time ordering of edges. This temporal dimension introduces new algorithmic, structural, and combinatorial challenges \cite{KempeConnectivity2000}, which have, in recent years, driven the rapid development of temporal graph theory towards various directions.

A central theme in this development consists of path-related problems: reachability in temporal graphs is defined via time-respecting paths (also called temporal paths), which must traverse edges in chronological order. This makes the reachability relation in temporal graphs very different from the one used in static graphs (in particular, temporal reachability is neither symmetric, nor transitive), resulting in a rich landscape of natural temporal graph problems, including, connectivity \cite{KempeConnectivity2000,BCCKRRZ23}, reachability \cite{EMMZ21,CRRZ24}, temporal spanners \cite{AF16,CPS21,ABFGHKLSSW24}, and explorability \cite{MS16,ErlebachOnTemporal2021}.
The latter is the main focus of this work. 

\sloppy
\paragraph{Temporal Graph Exploration.}
The goal of the \emph{Temporal Graph Exploration problem} (TEXP) is to determine how quickly one can visit all nodes in a temporal graph starting from a given node using a temporal walk, i.e., using only the edges available at each time step. Note that the temporal explorer is `clairvoyant' as they have knowledge of the snapshot of each time-step. This contrasts work on dynamic networks, where the agent only has information of the past and present set of edges, but not the future. TEXP was introduced by Michail and Spirakis~\cite{MS16} as a temporal analogue of the classic travelling salesperson problem. 

The decision version of \textsc{TEXP}, in which one has to decide if at least one exploration schedule exists,
is an NP-complete problem \cite{MS16}. In fact, this decision problem remains NP-complete even if the underlying graph (i.e., the static graph consisting of all edges that appear in at least one snapshot) has pathwidth 2 and every snapshot is a tree \cite{BodlaenderOnexploring2019}, 
or even if the underlying graph is a star and the exploration has to start and end at the centre of the star \cite{AkridaThetemporal2021}.

Michail and Spirakis \cite{MS16} proved that \textsc{TEXP} admits no $(2- \varepsilon)$-approximation algorithm for any $\varepsilon > 0$, unless P=NP.
In other words, there is no polynomial time algorithm that outputs an exploration schedule whose arrival time is at most $(2- \varepsilon)$ times the arrival time of an optimal exploration schedule.
This was substantially strengthened by Erlebach, Hoffman, and Kammer \cite{ErlebachOnTemporal2021} who established NP-hardness of $n^{1 - \varepsilon}$-approximation for any $\varepsilon > 0$.
In fact, the result was shown for \emph{always-connected} temporal graphs, i.e., temporal graphs in which every snapshot is a connected graph. 
This connectedness assumption makes the inapproximability result tight as any always-connected temporal graph can be explored in at most $n^2$ time steps \cite{MS16}, and clearly any exploration takes at least $n-1$ time steps.

The strong inapproximability result for \textsc{TEXP} on always-connected temporal graphs has motivated further study of bounds on the length of the fastest exploration schedules in such temporal graphs. For convenience, unless stated otherwise, we assume throughout that every temporal graph is always-connected and has at least~$n^2$ snapshots, and therefore can be explored.
Erlebach, Hoffman, and Kammer~\cite{ErlebachOnTemporal2021} showed that, for certain temporal graphs, any exploration requires time $\Omega(n^2)$, thereby proving that the $O(n^2)$ upper bound of~\cite{MS16} is asymptotically tight. 
Subquadratic bounds on the exploration time are available for temporal graphs with the underlying graphs being planar, bounded treewidth, or a cycle with a constant number of chords \cite{ErlebachOnTemporal2021, AdamsonFaster2022,taghian2020exploring}.
Furthermore, temporal graphs with snapshots of bounded degree \cite{ES18,ErlebachTwo2019},
$k$-edge-deficient temporal graphs, where each snapshot is obtained from the underlying graph by removing at most $k$ edges \cite{ErlebachExploration2022}, temporal graphs with regularly and probabilistically present edges \cite{ErlebachOnTemporal2021}, word-representable temporal graphs \cite{Ada25b}, are all known to admit subquadratic exploration time.

The aforementioned bounds on temporal exploration consider a worse-case analysis, in the sense that they assume that the snapshots are drawn adversarially. Despite the significant attention this setting has received, and the extensive progress made, several core cases remain unresolved and there are huge gaps between many of the best known lower and upper bounds. Arguably the most fundamental case is that of temporal graphs with bounded-degree snapshots, where the best known lower bound is $\Omega(n \log n)$ \cite{ErlebachOnTemporal2021}, while the best known upper bound is $O(n^{1.75})$ \cite{ErlebachTwo2019}---even when the underlying graph itself has bounded degree. Even some highly structured examples are not fully understood, for example bounds on exploration of a temporal $2 \times n$ grid differ by an $\Omega(\log^3 n)$ factor \cite{ErlebachOnTemporal2021}.

In discrete mathematics and theoretical computer science, it is common to study stubborn algorithmic, structural, and combinatorial problems on random instances in order to better understand them. Some such examples are $k$-SAT \cite{RandSat,ChvatalR92,Coja-Oghlan10},  Graph Isomorphism \cite{BES80}, graph width parameters \cite{LeeLO12,HendreyNST25}, Ramsey theory \cite{LuczakRV92}, but there are many more, see for example also \cite{CGsparse,Schacht}. 

Motivated by the limited understanding of the Temporal Graph Exploration problem in the adversarial setting, we study the problem on random instances to shed light on its average-case behaviour. Following this line of inquiry, our central question is whether the quadratic bounds for temporal exploration also hold in the average case. To address this, one requires a random model for temporal graphs, and we propose the following.
\begin{boxsampling}[\small Random Spanning Tree (RST) model]
\small
    A \emph{Random Spanning Tree model} is a pair $(\cT,\mu)$, where $\cT$ is a set of trees with vertex set $[n]$, and $\mu$ is a probability distribution over $\cT$.
    
    \medskip
    A random temporal graph $\cG = (G_i)_{i \in \mathbb{N}}$ is said to be generated by the Random Spanning Tree model, denoted $\cG \sim (\cT,\mu)$, if each $G_i$ is a tree independently drawn from $\cT$ according to $\mu$.
\end{boxsampling}

We note that the focus on trees, rather than arbitrary always-connected graphs, is deliberate: trees are minimally connected graphs, and having more edges in snapshots can only facilitate exploration in temporal graphs. Thus, on the one hand, the model is minimal; on the other hand, it captures the core difficulty of temporal exploration.

The main goal of this paper is to comprehensively study the exploration time in random temporal graphs generated by Random Spanning Tree models.

\subsection{Our results}

We begin our investigation of the exploration time of random temporal graphs generated by RSTs by examining existing deterministic lower bound constructions and adapting them into suitable RSTs to see if comparable lower bounds can be achieved.

A natural starting point is the deterministic construction yielding a lower bound of $\Omega(n^2)$ by Erlebach, Hoffmann, and Kammer \cite[Lemma~3.1]{ErlebachOnTemporal2021}. In this construction~\cite[Figure~1]{ErlebachOnTemporal2021}, the temporal graph consists of an even number of vertices partitioned into two equal parts. Each snapshot forms a star, and the centre of the star changes in a round-robin fashion, cycling through the vertices in only one of the two parts.
From this construction, we obtain an RST $(\mathcal{S}, \mu)$ by taking $\mathcal{S}$ to be the set of all possible stars with the centre restricted to the first half of $[n]$, and letting $\mu$ be the uniform distribution over $\mathcal{S}$. We show (\cref{cor:urst-half-stars}) that random temporal graphs generated according to $(\mathcal{S}, \mu)$ require, \whp, $\Omega(n^{3/2})$ steps to be explored. Intuitively, this is because exploring a vertex from the second half requires a star to be picked twice, which happens after roughly $\sqrt{n}$ steps. This is in contrast to the $n$ steps of the deterministic setting.
Furthermore, we show that by varying the number of stars in $\cS$ from 1 to $n$, no larger lower bound can be achieved (\cref{thm:urst-some-stars}).

Is this $\Omega(n^{3/2})$ lower bound the best possible? Could we modify $\mathcal{S}$ and/or $\mu$ to obtain RSTs that are even "harder" to explore?

Our first main result answers these questions by showing that one cannot do worse than this example. Denoting by $\texp(\mathcal{G})$ the minimum number of time steps required by any exploration schedule of $\mathcal{G}$ starting from any vertex (worst case), the result is formally stated as follows:

\begin{restatable}{theorem}{genupper}\label{thm:general-upper}
    There exists a constant $C\in \mathbb{R}_+$ such that the following holds. Let $n \in \bN$ be sufficiently large, $\cT$ be a set of trees on vertex set $[n]$, $\mu$ be a probability distribution on $\cT$, and $\cG \sim (\cT, \mu)$.
    Then
    \[
        \Pr{\texp(\cG) \leq  C\cdot n^{3/2} } \geq 1-  e^{-n}.
    \]
\end{restatable}

\cref{thm:general-upper} is a random analogue to the tight bound of $O(n^2)$ for adversarial always-connected temporal graphs \cite{MS16}. This result is our main technical contribution and although individual snapshots are independent in our model, what makes this result difficult to prove is the non-independence of partially explored/reachable sets of vertices once more than one time step has been revealed. This loss of independence by exposing many time steps together is however necessary, since we show in \cref{prop:online-stars} that for the RST defined by a set of $n/2$ stars and the uniform distribution, with high probability, no online explorer can visit all vertices in less than $\Omega(n^2)$ time steps.

We complement this with our second main result, which shows that any \rst can be explored in time linear in the number of edges of the underlying graph with high probability. 

\begin{restatable}{theorem}{orderm}\label{thm:orderm}
    Let $n$ and $m$ be sufficiently  large.
    Let $G$ be a connected graph on vertex set $[n]$ with $m$ edges and let $\cT$ be the set of spanning trees of $G$. Let $\mu$ be any probability distribution on $\cT$ and $\cG \sim (\cT, \mu)$. Then
    \[
        \mathbb{P}\big[ \texp(\cG) \leq  54m \big]\geq 1-  e^{-\Omega(m^{1/6})}.
    \]
\end{restatable}

\cref{thm:orderm} relates to a result of Erlebach, Hoffmann, and Kammer \cite{ErlebachOnTemporal2021} about exploration of temporal graphs with probabilistically present edges.
In such temporal graphs, every edge $e$ of the underlying graph $G$ appears in each time step independently with its own probability $p_e$, with the extra assumption that the total sum of the probabilities of the edges over each cut of $G$ is lower-bounded by some fixed constant. Due to being always-connected, temporal graphs generated by an RST satisfy  this assumption and therefore any result about temporal graphs with probabilistically present edges applies to our model too. 

In particular, due to \cite[Theorem 6.1]{ErlebachOnTemporal2021}\footnote{The $O(m)$ in expectation bound can be found in the proof of \cite[Theorem 6.1]{ErlebachOnTemporal2021}, the Theorem statement instead gives an $O(m\log n)$ bound \whp}, if $\cT$ is a set of spanning trees of a graph $G$ with $m$ edges, and $\mu$ is an arbitrary distribution over $\cT$, then, \emph{in expectation}, $\cG \sim (\cT,\mu)$ can be explored in $O(m)$ time steps.

This bound is order-optimal for sparse graphs, i.e., graphs with $O(n)$ edges, where $n$ is the number of vertices. Moreover, the corresponding exploration schedule is \emph{online}, meaning that the explorer does not require knowledge of future snapshots to make their decision in the current one.

The main limitation of this result from \cite{ErlebachOnTemporal2021} is that it bounds the exploration time in expectation. It was shown in \cite{ErlebachOnTemporal2021}, that the exploration time can be bounded with high probability at the expense of worsening the bound. Namely, for every $d \geq 1$, there exists an online exploration strategy that can be completed in $O_{d}(m \cdot \log n)$ time steps with probability at least $1-1/n^{d}$ \cite{ErlebachOnTemporal2021}.

\cref{thm:orderm} shows that $\cG$ can be explored in $O(m)$ steps \whp~if the explorer can see the future snapshots. Furthermore, we show in \cref{prop:online-lower-bound} that this linear bound cannot be achieved without knowledge of the future. In particular, we show that any online schedule that succeeds with probability at least $1 - n^{-\Omega(1)}$ requires $\Omega(n\log n)$ time steps, matching the upper bound of \cite{ErlebachOnTemporal2021} in the case of sparse graphs.

\subsection{Techniques \& Proof outlines}
In this section we give a high-level outline of the proof strategies used in our main results. Note that some constants are simplified/neglected for ease of presentation.

\paragraph{Exploration in $O(n^{3/2})$ time (\cref{thm:general-upper}).}

We start by outlining our proof of explorability of $(\cT,\mu)$ in time $O(n^{3/2})$ for arbitrary set of trees $\cT$ over $[n]$ and arbitrary distribution $\mu$ over $\cT$.

The key property (\cref{th:sqrt-close-vertices}) of RST that allows us to establish this upper bound is that
\begin{adjustwidth}{2em}{0em}
    ($\star$) every vertex $v \in [n]$ has a fixed set of at least $\sqrt{n}$ \emph{close} vertices, i.e., vertices that $v$ can reach in $\Theta(\sqrt{n})$ time steps with probability at least 1/9.
\end{adjustwidth}

This property is sufficient to derive the upper bound by constructing a straightforward exploration schedule in which, on average, $\sqrt{n}$ time steps are spent to visit each new vertex.
This is done as follows:
\begin{enumerate}[label=\arabic*.]
	\item First we build an auxiliary \emph{meta-graph} on $[n]$ where two vertices are connected if and only if they are close to each other. 
	
	\item Due to ($\star$), each vertex is in a connected component with at least $\sqrt{n}$ other vertices, and therefore there are at most $\sqrt{n}$ connected components in the meta-graph. 
	
	\item 
    Each connected component $C$ of the meta-graph can be explored in $\Theta(|C|\sqrt{n})$ time steps, where $|C|$ is the number of vertices in $C$, by first fixing an arbitrary spanning tree of $C$ and then traversing it along an Eulerian tour, spending $\Theta(\sqrt{n})$ time steps for each edge of the tour.

	\item After exploring the current component, we go to an unexplored component, which can always be done in at most $n$ steps (see \cref{lem:EHKwalk}).
	
	\item Summing over all connected components, the total time of this schedule is 
	\[
	       \sum_{C}    \Big(O(|C|\sqrt{n}) + n\Big) 
           =
           O(\sqrt{n}) \cdot \sum_{C} |C| + \sum_{C} n =
           O(n^{3/2}).
	\]
\end{enumerate}

The proof of ($\star$) forms the technical core of this upper bound. It proceeds by assuming the existence of a vertex~$w$ with fewer than $\sqrt{n}$ close vertices and arriving at a contradiction by demonstrating the existence of a set of vertices that, in expectation, can reach more than $n$ distinct vertices.

A key tool in the proof is a natural temporal analogue of the Depth-First search (DFS) algorithm, which, to the best of our knowledge, is introduced here for the first time (\cref{sec:temporalDFS}). 
Roughly speaking, the \emph{temporal DFS algorithm} discovers one new vertex in each snapshot~$F$ by revealing the first undiscovered vertex according to the (static) DFS order of~$F$.

A more detailed outline of the proof of ($\star$) is as follows:

\begin{enumerate}[label=\arabic*.]
	\item We fix a set $B$ of vertices that $w$ can discover via the temporal DFS algorithm (and therefore can reach via temporal path) in $\Theta(\sqrt{n})$ time steps with a constant probability. 
	
	\item Each vertex in $B$ is close to $w$, and thus, by our assumption, $|B|<\sqrt{n}$.
	
	\item This upper bound on $|B|$ allows us to show (\cref{cor:many-linking}) that \whp~among $\Theta(\sqrt{n})$ snapshots a constant fraction of them contain an edge that connects a vertex in $B$ with a vertex which is far from (i.e.~not close to) $w$. 
	
	\item We show that each such far vertex $u$, in the subsequent $\Theta(\sqrt{n})$ time steps, in expectation discovers via temporal DFS at least $\Theta(\sqrt{n})$ vertices none of which can be reached from $B$ during the same time interval (\cref{cor:F-S-new}).

	\item We further show that all these sets of $\Theta(\sqrt{n})$ vertices are disjoint, and thus they add up to more than $n$ vertices, leading to the sought contradiction.
\end{enumerate}

\Cref{cor:F-S-new}, which underpins Step 4 above, utilizes the following behaviour of the temporal DFS algorithm on RST models (\cref{lem:stable-diffuse}). Roughly speaking, for a fixed vertex~$v$ and two independent random temporal graphs, \whp, the sets of vertices discovered from~$v$ via temporal DFS in the two graphs either overlap substantially, or each set contains a large subset of private vertices, each of which can be revealed from the opposite set in one step of the temporal DFS algorithm with probability at least $1/(2n)$.

\paragraph{Exploration in $O(m)$ time (\cref{thm:orderm}).}
Recall, in \cref{thm:orderm}, $\cT$ is the set of spanning trees of a connected graph $G$.
For each edge $e$ of graph $G$, the distribution $\mu$ defines the probability of $p_e$ that $e$ appears in any given snapshot. Thus, the number of snapshots one needs to wait for $e$ to appear, is a geometric random variable $\xi_e$ with the success probability $p_e$. We define the weight of $e$ as $w_e = 1/p_e$, i.e., the expected value of $\xi_e$. A natural strategy to build an efficient online exploration schedule of $\cG \sim (\cT, \mu)$ is to find the minimum-weight spanning tree $T$ of $G$ and visit the vertices following an Euler tour of $T$ by crossing each edge at the earliest opportunity. In expectation, the length of such an exploration is upper-bounded by $2\sum_{e \in E(T)} \Ex{\xi_e} = 2\sum_{e \in E(T)} w_e$, i.e., by twice the total weight of $T$.

Using this strategy, Erlebach, Hoffmann, and Kammer \cite{ErlebachOnTemporal2021} showed that, \emph{in expectation}, $\cG$\footnote{Recall, in \cite{ErlebachOnTemporal2021},  $\cG$ is defined as a temporal graph with probabilistically available edges, where each edge $e$ of $G$ appears independently with probability $p_e$} can be explored in $O(m)$ steps by proving that a minimum-weight spanning tree of $G$ has weight $O(m)$, where $m$ is the number of edges if $G$.

The main obstacle to exploring the graph in $O(m)$ steps \emph{\whp} in an online manner by following the minimum-weight spanning tree is the potential presence of high-weight edges, which may take a long time to cross. In essence, the key idea behind the proof of \cref{thm:orderm} is to replace such high-weight edges with fast temporal paths. This substitution requires the explorer to have access to future snapshots.

We outline this strategy in more detail below:
\begin{enumerate}[label=\arabic*.]
    \item \textit{Building the backbone forest}:
    Denote by $T_{\mathsf{min}}$ the minimum-weight spanning tree of $G$.  
    Among the edges of $T_{\mathsf{min}}$, we select those with sufficiently low weight. These low-weight edges form a \emph{backbone} forest~$F$, each connected component of which can be explored efficiently.

    \item \textit{Partitioning backbone components into fast components}: 
    In order to establish fast connections between the connected components of the backbone forest $F$,
    we partition each such component into subtrees, which we call \emph{fast components}. These are constructed to have both small diameter and only ultra low-weight edges, ensuring that traversal within each fast component is fast. 
    We further guarantee that the total number of the fast components is not too large, which allows us to efficiently connect the backbone components together.

    \item \textit{Connecting backbone components via meta-edges}:
    Using the bound on the total number of fast components, we show that in any partition of the fast components into two non-empty sets, there exists \whp~a pair of components—one from each set—such that an edge connects them during any short time interval. This property implies the existence of a fast temporal path between any source and destination vertices in these fast components, starting at any time step. 

    Specifically, by looking ahead into the future, we can identify a connecting edge between the two fast components within a short time interval. We then travel from the source vertex to an endpoint of the connecting edge before the interval begins, cross to the other component as soon as the edge becomes available, and continue to the destination vertex. Overall, the resulting temporal path is fast due to efficient traversal within fast components and the guaranteed presence of a connecting edge in every short time interval.
    
    We mark the existence of such an efficient connection between the two fast components by creating a meta-edge between an arbitrary pair of vertices, with one vertex from each of the components.
    
    \item \textit{Building efficient exploration}:
    Since any bipartition of the connected components of the backbone forest~$F$ induces a bipartition of the fast components, we can use meta-edges to greedily merge the connected components of~$F$ into a tree. We denote this tree by~$H$ and use it to construct an efficient exploration schedule.

    To this end, we fix an Euler tour of~$H$ and define the exploration schedule by following this tour. We show that the total exploration time of this schedule is $O(m)$ \whp, by separately bounding the total traversal time of the edges in the backbone forest and the total traversal time of the meta-edges in~$H$.
\end{enumerate}

\subsection{Future directions and open problems}\label{sec:outlook}

Our study of the Temporal Graph Exploration problem in Random Spanning Tree models suggests several interesting questions and research directions, some of which we outline below.

In deterministic always-connected temporal graphs, temporal exploration can be done much faster than the worst case $n^2$ time if snapshots have bounded vertex degree. 
For example, temporal graphs in which every snapshot has degree at most $d$ can be explored in time $O_d(n^{1.75})$ \cite{ErlebachTwo2019}. On the other hand, for every $2 \leq d \leq n-1$, there exist always-connected temporal graphs with underlying graphs of degree at most $d$ that require $\Omega(d n)$ time steps to be explored \cite{ErlebachOnTemporal2021}. By adapting the construction behind this lower bound, we show (\cref{thm:bounded-degree-lower}) that, for any $2 \leq d \leq n-1$, there exists an RST with the underlying graph of degree at most $d$ that requires $\Omega(\sqrt{d} \cdot n)$ time steps to be explored. This result generalises our $\Omega(n^{3/2})$ lower bound in \cref{thm:general-upper} (where snapshots have maximum degree $d = n-1$), and we believe there exists a matching upper bound, which would generalise our \cref{thm:general-upper}.

\begin{conjecture}\label{con:degree-upper}
    There exists a constant $C\in \mathbb{R}_+$ such that the following holds.
    Let $n \in \bN$, and let $d := d(n)$ be such that $2 \leq d \leq n-1$. Let $\cT$ be a set of trees of maximum degree $d$ on vertex set $[n]$, $\mu$ be a probability distribution on $\cT$, and $\cG \sim (\cT, \mu)$.
	Then 
    \[
        \Pr{\texp(\cG) \leq  C\cdot \sqrt{d} \cdot n } \xrightarrow[n \to \infty]{} 1.
    \]
\end{conjecture}

One approach to proving this conjecture---similarly to the proof of \cref{thm:general-upper}---would be to show that the auxiliary graph, in which two vertices are connected if and only if they can reach each other in \( O(\sqrt{d}) \) time steps with constant probability, has \( O(\sqrt{d}) \) connected components. However, unlike in the proof of \cref{thm:general-upper}, we cannot expect any uniform upper bound on the size of the connected components in the auxiliary graph, as demonstrated by the example in \cref{fig:enter-label}.

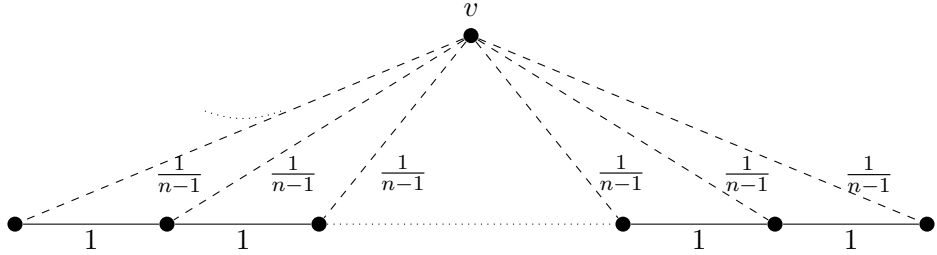
\begin{figure} 
    \centering
 	\begin{tikzpicture}[every node/.style={circle,fill=black,inner sep=2pt}, 
		edge label/.style={draw=none,fill=none,shape=rectangle}]
		
		\node[label=north:$v$] (apex) at (6,2.5) {};
		
		\node (v1) at (0,0) {};
		\node (v2) at (2,0) {};
		\node (v3) at (4,0) {};
		 
		\node (vn2) at (8,0) {};
		\node (vn1) at (10,0) {};
		\node (vn)  at (12,0) {};
		
	 \draw (v1) -- node[midway, below, edge label] {$1$} (v2);
	\draw (v2) -- node[midway, below, edge label] {$1$} (v3);
		\draw[dotted] (v3) -- (vn2);
	  \draw (vn2) -- node[midway, below, edge label] {$1$} (vn1);
	\draw (vn1) -- node[midway, below, edge label] {$1$} (vn);
	
	\draw[dashed] (apex) -- node[pos=0.75, right, edge label] {$\;\;\frac{1}{n-1}$} (v1); 
		\draw[dashed] (apex) -- node[pos=0.75, right, edge label] {$\;\;\frac{1}{n-1}$} (v2); 
			\draw[dashed] (apex) -- node[pos=0.75, right, edge label] {$\,\;\frac{1}{n-1}$} (v3); 
	\draw[dashed] (apex) -- node[pos=0.75, right, edge label] {$\;\;\;\,\frac{1}{n-1}$} (vn); 
\draw[dashed] (apex) -- node[pos=0.75, right, edge label] {$\;\,\;\frac{1}{n-1}$} (vn1); 
\draw[dashed] (apex) -- node[pos=0.75, right, edge label] {$\;\frac{1}{n-1}$} (vn2); 

		\draw[dotted] (2.5,1.5) to[out=-20, in=200] (3.5,1.5);
	\end{tikzpicture}
    \caption{An RST with snapshots of maximum degree $3$, with a vertex $v$ that needs $\Omega(\sqrt{n})$ time steps to reach any other fixed vertex with constant probability.}
    \label{fig:enter-label}
\end{figure}

In this example, all horizontal edges are always present, and at each snapshot exactly one of the other edges appears uniformly at random. In temporal graphs generated by this RST, the maximum degree of each snapshot is at most 3, but the top vertex $v$ in \cref{fig:enter-label} requires $\Theta(\sqrt{n})$ time steps to reach any other fixed vertex with constant probability. Thus, the auxiliary graph has two connected components---one of size 1 and the other of size $n - 1$.

Our next conjecture concerns algorithmic aspects of temporal exploration of RSTs. A natural exploration strategy is greedy exploration, which starts from an arbitrary vertex, looks into the future to find the closest unvisited vertex, visits that vertex at the earliest possible time, and then repeats the process: looking ahead to discover the next closest unvisited vertex and visiting it as early as possible from the current position, and so on. While this is not always an effective strategy in adversarially constructed deterministic temporal graphs, we believe that in the RST models, this strategy performs within a constant factor of the optimal.

\begin{conjecture}
    \label{con:greedy-approx}
    There exists a constant $C\in \mathbb{R}_+$ such that the following holds. Let $n$ be a natural number, $\cT$ be a set of trees on vertex set $[n]$, $\mu$ be a probability distribution on $\cT$, and $\cG \sim (\cT, \mu)$. Let $\tau$ be the exploration time of $\cG$ using the greedy strategy. Then
    \[
        \pr\Big[ {\tau \leq  C \cdot \texp(\cG)} \Big] \xrightarrow[n \to \infty]{} 1.
    \]
\end{conjecture}

Finally, we believe that the RST models is a natural framework that contributes to the actively developing line of research on random models of temporal graphs \cite{CMMD07,CMPS09,CMMPS10,GH10,ZMN17,AMNRSZ20,CRRZ24,BCCKRRZ23,BKL24,ADL24,BLP24,BDKLM25,ADL25}.
Beyond common temporal graph properties and algorithmic problems one can study on temporal graphs generated by Random Spanning Tree models, there are two natural ways to generalise the model.

First, one can consider arbitrary sets of $n$-vertex graphs as snapshots, rather than only spanning trees. 
For example, in the context of temporal graph exploration, if there is at least one connected snapshot in the support of $\mu$, then the exploration time is almost surely finite. It is possible that the assumption of always-connectedness can also be relaxed in a more meaningful way by providing some conditions on the set of snapshots and $\mu$ which guarantee `fast' exploration times even when none of the snapshots are connected. A simple example here is the uniform measure over graphs consisting of a pair of cliques of size $n/2$, which is explorable (even online) in time $O(n)$. 

Second, one can introduce dependencies between snapshots by replacing the distribution $\mu$ with a Markov chain whose states are unique snapshots (note that the RST models would be a special case of this model). Such a generalisation would provide some interpolation between completely independent snapshots and deterministic temporal graphs; for example, it would capture all deterministic temporal graphs with pairwise distinct snapshots as well as periodic temporal graphs. It would be interesting to understand which properties of such Markov chains lead to fast exploration, where an obvious pair would be ergodic and rapidly mixing. Several specific instances of such random temporal graphs have appeared previously, such as edge Markovian \cite{CaiSZ20,EdgeMarkovianClemnti} and mobile geometric \cite{Mobile1,Mobile2} graphs. The flooding problem has  been studied on more general Markovian models \cite{ClementiMPS11}. 

\section{Preliminaries}

For a natural number $n$, we use standard notation $[n]$ to denote the set $\{1,2, \ldots, n\}$.

A {\em temporal graph} $\cG$ is a finite or infinite sequence $G_1=(V,E_1), G_2=(V,E_2), G_3=(V,E_3), \ldots$ of static graphs over the same vertex set $V$, which are called the snapshots of $\cG$. The \emph{lifetime} of $\cG$ is the number of snapshots it has. The static graph over the vertex set $V$ that has an edge between two vertices if and only if these vertices are adjacent in at least one snapshot, is called the \emph{underlying graph} of $\cG$. A \emph{time-edge} of $\cG$ is a pair $(e,t)$, where $t$ is a natural number not exceeding the lifetime of $\cG$ and $e \in E_t$; the number $t$ is called the time step of the time-edge $(e,t)$.
The temporal graph $\cG$ is \emph{always-connected} if each of its snapshots is a connected graph.

A \emph{temporal path} (resp. \emph{temporal walk}) in $\cG$ is a sequence $((e_1,t_1),\dots,(e_{\ell},t_{\ell}))$ of time-edges of $\cG$ such that $(e_1,e_2, \ldots, e_{\ell})$ is a path (resp. walk) in the underlying graph of $\cG$ and $t_i < t_{i+1}$ holds for every $i \in [\ell-1]$.
A vertex $v \in V$ can \emph{reach} a vertex $u$ in $\cG$ by time step $t$, if there exists a temporal path from $v$ to $u$ whose last time-edge has time at most $t$. Given a set of vertices $S \subseteq V$, we denote by $R_t^{\cG}(S)$ the set of vertices each of which can be reached by time $t$ from some vertex in $S$. If $S$ consists of a single vertex $v$, for brevity, we write $R_t^{\cG}(v)$ as a shorthand for $R_t^{\cG}(\{v\})$.

Given a temporal graph $\cG = (G_1,G_2, \ldots)$ and time steps $i,j$ such that $i \leq j$ and $j$ does not exceed the lifetime of \cG, we denote by $\cG_{[i,j]}$ the temporal graph $(G_{i},G_{i+1},\ldots,G_{j})$.

For a temporal graph $\cG$ of finite lifetime, the \emph{reverse} of $\cG$, denoted $\overleftarrow{\cG}$, is the temporal graph whose sequence of snapshots is the reverse of the sequence of snapshots of $\cG$.

\paragraph{Temporal Exploration.} Given a temporal graph $\cG$, we denote by $\texp(\cG)$ the minimum number $t$ such that every vertex of $\cG$ can be visited via a temporal walk starting at some vertex of $\cG$ within the first $t$ snapshots. If no such $t$ exists, then we define $\texp(\cG)$ to be $\infty$. 
It is common to describe exploration schedules in terms of an explorer moving along a temporal walk. In this context, we use the terms \emph{travel} and \emph{traverse} to refer to the act of the explorer moving along the walk.

We also consider exploration of temporal graphs by an `online' explorer, which can only see the current snapshot.
An online exploration algorithm for $\mathcal{G}$ is a function $A$ that, given a time step $t$, the partially revealed temporal graph $\mathcal{G}_{[0,t]}$, the current vertex $v$, and the set of unexplored vertices $U$, outputs the next vertex to visit, which is either the current vertex $v$ or one of its neighbors at time step $t$.
We denote by $\Dexp_A(\cG)$ the smallest time step by which every vertex of $\mathcal{G}$ is visited by algorithm $A$.

We need the following standard yet useful lemma, which says that in any $n$-vertex always-connected temporal graph any vertex can reach any other vertex in at most $n$ time steps. 

\begin{lemma}[{\cite[Lemma 2.1]{ErlebachOnTemporal2021},\cite[Observation 2.1]{Kuhn2010Distributed}}]\label{lem:EHKwalk} 
    Let $\mathcal{G}$ be an always-connected temporal graph on $n$ vertices. Then for any two vertices $v$ and $u$ of $\cG$, and any $t \in \bN$ such that $t+n-1$ does not exceed the lifetime of $\cG$, the temporal graph $\cG_{[t,n+t-1]}$ has a temporal path from $v$~to~$u$.
\end{lemma}

\paragraph{Probability.}
Let $X$ and $Y$ be random variables. We say that $X$ \emph{stochastically dominates} $Y$ if $\Pr{X \geq x} \geq \Pr{Y \geq x}$ for all real numbers $x$, and denote this by $X \succeq Y$.
We will also make use of a number of concentration inequalities, which can be found in \cref{sec:furtherprelim}.

\section{Exploration in \texorpdfstring{$O(n^{3/2})$}{n to the 3/2} time} 
\label{sec:n_1_5_bound}
In this section we prove that for any RST model $(\cT,\mu)$, 
a random temporal graph 
$\cG \sim (\cT, \mu)$ can be explored in time $\Theta(n^{3/2})$ with high probability.
Throughout we assume that a set $\cT$ of trees on vertex set $[n]$ and a probability distribution $\mu$ on $\cT$ are arbitrary but fixed.

This result is based on the fundamental property of $(\cT,\mu)$ that for every vertex $v \in [n]$ in $\cG \sim (\cT, \mu)$ there are $\Theta(\sqrt{n})$ vertices that are ``close'' to $v$.

To state this property formally, we first introduce a notion that measures closeness between vertices in an RST model.
Recall, given a temporal graph $\cG$, a vertex $v$ of $\cG$, and a natural number $t$, the set of vertices that $v$ can reach in the first $t$ snapshots of $\cG$ is denoted by~$R_t^{\cG}(v)$.

\begin{definition}\label{definition:close_vertices}
	Let $u, w \in [n]$, $t \in \mathbb{N}$, $p \in [0,1]$. We say that \emph{$w$ is $(t,p)$-close to $u$ in $(\cT,\mu)$}~if 
	\[
	\underset{\cG \sim (\cT, \mu)}{\pr}\left[w \in \Flood_{t}^{\cG}(u)\right] \geq p.
	\]
\end{definition}

Note, for a finite $T \in \bN$, any $\cG_{[1,T]}$ and its reverse $\overleftarrow{\cG_{[1,T]}}$ have the same probability to be sampled in $(\cT, \mu)$. Therefore, a vertex $w$ is $(t,p)$-close to a vertex $u$ if and only if $u$ is $(t,p)$-close to $w$. In other words, the $(t,p)$-closeness relation is symmetric, and we can talk about $u$ and $w$ being $(t,p)$-close to each other.
If $u$ and $w$ are not $(t,p)$-close to each other, then we say that they are \emph{$(t,p)$-far} from each other in $(\cT, \mu)$. When $(\cT, \mu)$ is clear from the context we may say that $u$ and $w$ are $(t,p)$-close to each other or $(t,p)$-far from each other without specifying $(\cT, \mu)$.

Using the notion of $(t,p)$-closeness, the above-mentioned property of $(\cT,\mu)$ can be stated as follows.

\begin{restatable}{theorem}{sqrtclosevertices}\label{th:sqrt-close-vertices}
	Let $n \in \bN$ be sufficiently large, $\cT$ be a set of trees on vertex set $[n]$, and $\mu$ be a probability distribution on $\cT$.
	Then, for every vertex $v \in [n]$ there are at least $\sqrt{n}$ vertices that are $(700\sqrt{n}, 1/9)$-close to $v$ in $(\cT, \mu)$. 
\end{restatable}

Most of this section is devoted to proving this result. Before proceeding to the proof, we first use \cref{th:sqrt-close-vertices} to establish our main result.

\genupper*
\begin{proof}
    Given $(\cT, \mu)$, we define a \emph{meta-graph}
    on $[n]$ as a graph where two vertices are adjacent if and only if they are $(700\sqrt{n}, 1/9)$-close to each other. By \Cref{th:sqrt-close-vertices}, each connected component of this graph is of size at least $\sqrt{n}$, which means that there are at most $\sqrt{n}$ of those components. Our exploration strategy is then to do an arbitrary tree traversal on each component of the meta-graph, and then traverse between the components in an arbitrary line sequence.  
	
	For traversing the spanning trees of the connected components, we observe that if two vertices $u$ and $v$ are $(700\sqrt{n}, 1/9)$-close to each other, then the time to traverse from $u$ to $v$ (or equivalently cross the meta-edge $e=\{u,v\}$) is stochastically dominated by $700\sqrt{n}\cdot X_{e}$, where  $ X_{e} \sim \geo{1/9}$.
    Let $X=\sum_{e \in Q}X_e$, where $Q$ is the set of at most $2n-2$ meta-edges that must be crossed in the traversals of the connected components. Then, applying \cref{lem:jansontail} with $p_*=1/9$, $\lambda =280n/\mu$, and $9(n-\sqrt{n}) \leq \mu \leq (2n-2)\cdot 9 \leq 20n$ gives  
	\[
		\Pr{X \geq 280n  } \leq \exp\left( - \frac{1}{9} \cdot \mu \cdot \left(\frac{280n}{\mu} -1- \log  \frac{280n}{\mu} \right)\right) \leq \exp(-n), 
	\]
	thus, if $Y$ is the total time spent traversing edges of the connected components, then \[\Pr{Y \geq 700\sqrt{n}\cdot 280n  } \leq e^{-n} . \]Since we can travel between any two components deterministically in time at most $n$ by \cref{lem:EHKwalk}, the time spent travelling between components is at most $\sqrt{n}\cdot n = n^{3/2}$ with probability $1$. Thus, it follows that the total time to explore the graph is at most $700\sqrt{n}\cdot 280n  + n^{3/2} \leq 200000n^{3/2} $ with probability $1-e^{-n}$. 
\end{proof} 

The rest of \cref{sec:n_1_5_bound} is devoted to proving \cref{th:sqrt-close-vertices}. 
To do so, we first define a temporal analogue of the Depth-First Search algorithm in \cref{sec:temporalDFS}. Then, in \cref{sec:tempDFSoverRST}, we establish several properties of temporal DFS when run on random temporal graphs generated by an RST. In \cref{sec:few-close-vertices}, we derive a consequence of a vertex having only a few close vertices. Finally, in \cref{sec:sqrt-close-vertices-proof}, we use these tools to prove \cref{th:sqrt-close-vertices}.

\subsection{Temporal Depth-First Search}
\label{sec:temporalDFS}
The standard Depth-First Search (DFS) algorithm for a \emph{connected} static graph $F$ starts at a given vertex $v$ and visits previously unvisited vertices of $F$ one by one in a specific order (the details of which are not relevant here).
Given a vertex $v$ and a graph $F$, the algorithm induces a linear order on the vertices of $F$ based on the sequence in which they are visited. This order is called the \emph{DFS order} of the vertices of $F$ starting at vertex $v$, and we denote it by $\sigma_{v,F}$. By definition, $v$ is the first vertex in the order $\sigma_{v,F}$.

We introduce a natural temporal generalisation of DFS, which we call the \emph{temporal DFS algorithm}.

Given an always-connected temporal graph $\cG$ and a vertex $v$ in $\cG$, the \emph{temporal DFS algorithm} from vertex $v$ starts at $v$ and in every snapshot $F$ visits a new vertex following the order $\sigma_{v,F}$; namely, we visit the first unvisited vertex under the order $\sigma_{v,F}$. For $t \in \bN$, we denote by $\OF_t^{\cG}(v)$ the set of vertices visited by the temporal DFS algorithm from $v$ by time step $t$.

In order to describe the algorithm formally, we first introduce some notation. Given a snapshot (connected graph) $F$, a starting vertex $v \in [n]$, and a set $D \subseteq [n]$ (which we think of as the set of visited vertices) such that $v \in D$, we denote by $\rho(v,D,F)$ the first vertex $w$ under the order $\sigma_{v,F}$ such that $w$ is outside $D$, but every vertex preceding $w$ in the order $\sigma_{v,F}$ belongs to $D$.

The following pseudocode describes the temporal DFS algorithm formally: given $\cG$, $v\in[n]$, and $t \geq 1$, the algorithm outputs the set $\OF_t^{\cG}(v)$ (see \cref{fig:temporalDFS} for illustration). 

\begin{algorithm}[H]
	\label{alg:ordered-flooding}
	\small
	\caption{\textsc{Temporal DFS}}
	\KwIn{An always-connected temporal graph $\mathcal{G}=(F_i=([n],E_i))_{i\geq 0}$; a \emph{source} vertex $v \in [n]$;  an integer time $t\geq 1$.}
	\KwOut{The set of vertices $\OF_t^{\cG}(v)$ visited by time $t$ in $\cG$ by the temporal DFS algorithm starting from $v$.}
	
	$\OF \gets \{v\}$\;
	
	\For{$i = 1$ \textup{ to } $t$}{
		$u \gets \rho(v,D,F_i)$\;
		$\OF \gets \OF \cup \{u\} $\;
	}	
	\KwRet{$D$}
\end{algorithm}

\begin{remark}\label{rem:temporal-search}
   We note that, in a similar manner, by appropriately specifying the orders $\sigma_{v,F}$, one can define temporal versions of other static graph search algorithms (such as BFS), which may be of independent interest. Moreover, our proofs remain valid without any modification if temporal DFS is replaced by any other type of temporal first-search algorithm.
\end{remark}

\tikzset{
  knoten/.style={
    draw=black,
    circle,
    thick,
    minimum size=0.45cm,
    inner sep=0pt,
    fill=white,
    text centered
  },
  visistedknoten/.style={
    knoten,
    fill=green!30
  },
  edge/.style={
    gray
  },
  visitedge/.style={
    line width=2pt,
    green!70!black
  }
}

\captionsetup[subfigure]{labelformat=empty}
\newcommand{\xshift}{0}
\newcommand{\scale}{1.5}
\newcommand{\snapshotscale}{0.75}

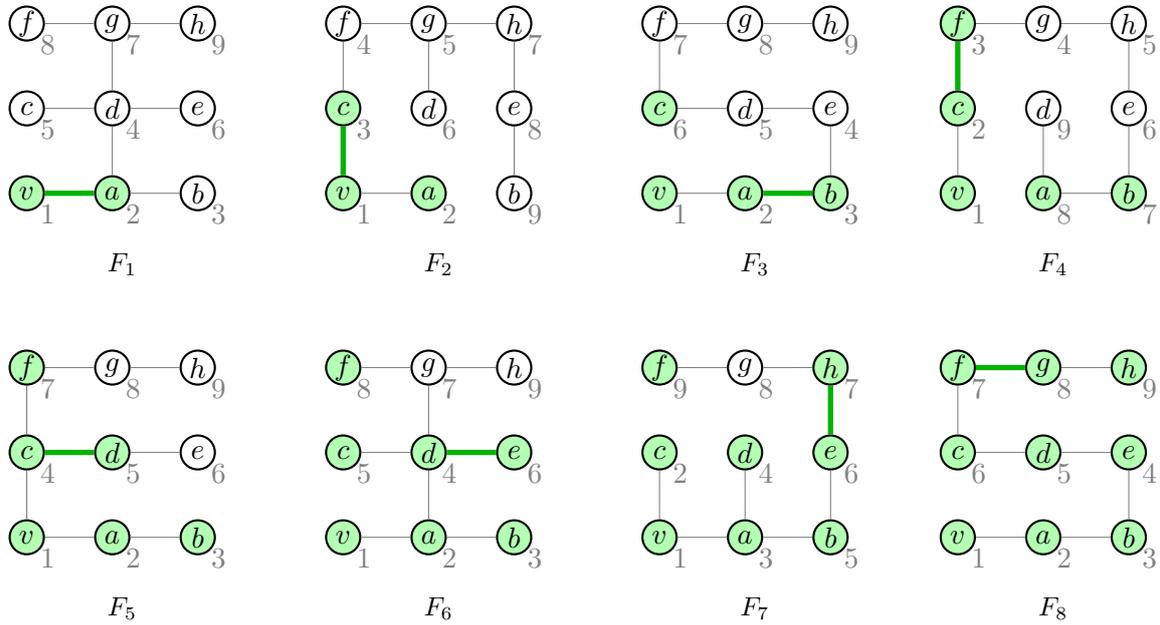
\begin{figure}[htbp]
    \centering
    \begin{tabular}{cccc}
        \begin{subfigure}{0.22\textwidth}
            \centering
            \begin{tikzpicture}[scale=\snapshotscale]
                \foreach \name/\x/\y/\label in 
                {v/0/0/1, a/1/0/2} 
                {
                    \node[visistedknoten] (\name) at ({\scale*\x + \xshift}, {\scale*\y}) {$\name$};
                    \node[gray] at ([xshift=4pt,yshift=-4pt]\name.south east) {\label};
                }
                
                \foreach \name/\x/\y/\label in 
                {b/2/0/3, c/0/1/5, d/1/1/4, 
                e/2/1/6, f/0/2/8, g/1/2/7, h/2/2/9} 
                {
                    \node[knoten] (\name) at ({\scale*\x + \xshift}, {\scale*\y}) {$\name$};
                    \node[gray] at ([xshift=4pt,yshift=-4pt]\name.south east) {\label};
                }

                \foreach \from/\to in 
                {v/a, a/b, a/d, d/g, c/d, d/e, f/g, g/h} 
                {
                    \draw[edge] (\from) -- (\to);
                }

                \draw[visitedge] (v) -- (a);
            \end{tikzpicture}
            \caption{$F_1$}
        \end{subfigure} &
        \begin{subfigure}{0.22\textwidth}
            \centering
            \begin{tikzpicture}[scale=\snapshotscale]
                \foreach \name/\x/\y/\label in 
                {v/0/0/1, a/1/0/2, c/0/1/3} 
                {
                    \node[visistedknoten] (\name) at ({\scale*\x + \xshift}, {\scale*\y}) {$\name$};
                    \node[gray] at ([xshift=4pt,yshift=-4pt]\name.south east) {\label};
                }
                
                \foreach \name/\x/\y/\label in 
                {b/2/0/9, d/1/1/6, 
                e/2/1/8, f/0/2/4, g/1/2/5, h/2/2/7} 
                {
                    \node[knoten] (\name) at ({\scale*\x + \xshift}, {\scale*\y}) {$\name$};
                    \node[gray] at ([xshift=4pt,yshift=-4pt]\name.south east) {\label};
                }

                \foreach \from/\to in 
                {v/a, c/f, f/g, g/h, g/d, h/e, e/b} 
                {
                    \draw[edge] (\from) -- (\to);
                }

                \draw[visitedge] (v) -- (c);
            \end{tikzpicture}
            \caption{$F_2$}
        \end{subfigure}&
        \begin{subfigure}{0.22\textwidth}
            \centering
            \begin{tikzpicture}[scale=\snapshotscale]
                \foreach \name/\x/\y/\label in 
                {v/0/0/1, a/1/0/2, b/2/0/3, c/0/1/6} 
                {
                    \node[visistedknoten] (\name) at ({\scale*\x + \xshift}, {\scale*\y}) {$\name$};
                    \node[gray] at ([xshift=4pt,yshift=-4pt]\name.south east) {\label};
                }
                
                \foreach \name/\x/\y/\label in 
                {e/2/1/4, d/1/1/5, f/0/2/7, g/1/2/8, h/2/2/9} 
                {
                    \node[knoten] (\name) at ({\scale*\x + \xshift}, {\scale*\y}) {$\name$};
                    \node[gray] at ([xshift=4pt,yshift=-4pt]\name.south east) {\label};
                }

                \foreach \from/\to in 
                {v/a, b/e, e/d, d/c, c/f, f/g, g/h} 
                {
                    \draw[edge] (\from) -- (\to);
                }

                \draw[visitedge] (a) -- (b);
            \end{tikzpicture}
            \caption{$F_3$}
        \end{subfigure} & 
        \begin{subfigure}{0.22\textwidth}
            \centering
            \begin{tikzpicture}[scale=\snapshotscale]
                \foreach \name/\x/\y/\label in 
                {v/0/0/1, a/1/0/8, b/2/0/7, c/0/1/2, f/0/2/3} 
                {
                    \node[visistedknoten] (\name) at ({\scale*\x + \xshift}, {\scale*\y}) {$\name$};
                    \node[gray] at ([xshift=4pt,yshift=-4pt]\name.south east) {\label};
                }
                
                \foreach \name/\x/\y/\label in 
                {e/2/1/6, d/1/1/9, g/1/2/4, h/2/2/5} 
                {
                    \node[knoten] (\name) at ({\scale*\x + \xshift}, {\scale*\y}) {$\name$};
                    \node[gray] at ([xshift=4pt,yshift=-4pt]\name.south east) {\label};
                }

                \foreach \from/\to in 
                {v/c, b/e, f/g, g/h, h/e, a/b, a/d} 
                {
                    \draw[edge] (\from) -- (\to);
                }

                \draw[visitedge] (c) -- (f);
            \end{tikzpicture}
            \caption{$F_4$}
        \end{subfigure} \\[2em]
        
        \begin{subfigure}{0.22\textwidth}
            \centering
            \begin{tikzpicture}[scale=\snapshotscale]
                \foreach \name/\x/\y/\label in 
                {v/0/0/1, a/1/0/2, b/2/0/3, c/0/1/4, f/0/2/7, d/1/1/5} 
                {
                    \node[visistedknoten] (\name) at ({\scale*\x + \xshift}, {\scale*\y}) {$\name$};
                    \node[gray] at ([xshift=4pt,yshift=-4pt]\name.south east) {\label};
                }
                
                \foreach \name/\x/\y/\label in 
                {e/2/1/6, g/1/2/8, h/2/2/9} 
                {
                    \node[knoten] (\name) at ({\scale*\x + \xshift}, {\scale*\y}) {$\name$};
                    \node[gray] at ([xshift=4pt,yshift=-4pt]\name.south east) {\label};
                }

                \foreach \from/\to in 
                {v/a, a/b, v/c, c/d, d/e, c/f, f/g, g/h} 
                {
                    \draw[edge] (\from) -- (\to);
                }

                \draw[visitedge] (c) -- (d);
            \end{tikzpicture}
            \caption{$F_5$}
        \end{subfigure} &
        \begin{subfigure}{0.25\textwidth}
            \centering
            \begin{tikzpicture}[scale=\snapshotscale]
                \foreach \name/\x/\y/\label in 
                {v/0/0/1, a/1/0/2, b/2/0/3, c/0/1/5, d/1/1/4, f/0/2/8, e/2/1/6} 
                {
                    \node[visistedknoten] (\name) at ({\scale*\x + \xshift}, {\scale*\y}) {$\name$};
                    \node[gray] at ([xshift=4pt,yshift=-4pt]\name.south east) {\label};
                }
                
                \foreach \name/\x/\y/\label in 
                {g/1/2/7, h/2/2/9} 
                {
                    \node[knoten] (\name) at ({\scale*\x + \xshift}, {\scale*\y}) {$\name$};
                    \node[gray] at ([xshift=4pt,yshift=-4pt]\name.south east) {\label};
                }

                \foreach \from/\to in 
                {v/a, a/b, a/d, d/g, c/d, f/g, g/h} 
                {
                    \draw[edge] (\from) -- (\to);
                }

                \draw[visitedge] (d) -- (e);
            \end{tikzpicture}
            \caption{$F_6$}
        \end{subfigure} &
        \begin{subfigure}{0.22\textwidth}
            \centering
            \begin{tikzpicture}[scale=\snapshotscale]
                \foreach \name/\x/\y/\label in 
                {v/0/0/1, a/1/0/3, b/2/0/5, c/0/1/2, f/0/2/9, e/2/1/6, d/1/1/4, h/2/2/7} 
                {
                    \node[visistedknoten] (\name) at ({\scale*\x + \xshift}, {\scale*\y}) {$\name$};
                    \node[gray] at ([xshift=4pt,yshift=-4pt]\name.south east) {\label};
                }
                
                \foreach \name/\x/\y/\label in 
                {g/1/2/8} 
                {
                    \node[knoten] (\name) at ({\scale*\x + \xshift}, {\scale*\y}) {$\name$};
                    \node[gray] at ([xshift=4pt,yshift=-4pt]\name.south east) {\label};
                }

                \foreach \from/\to in 
                {v/a, v/c, b/e, f/g, g/h, h/e, a/b, a/d} 
                {
                    \draw[edge] (\from) -- (\to);
                }

                \draw[visitedge] (e) -- (h);
            \end{tikzpicture}
            \caption{$F_7$}
        \end{subfigure}& 
        \begin{subfigure}{0.22\textwidth}
            \centering
            \begin{tikzpicture}[scale=\snapshotscale]
                \foreach \name/\x/\y/\label in 
                {v/0/0/1, a/1/0/2, b/2/0/3, c/0/1/6, e/2/1/4, d/1/1/5, f/0/2/7, g/1/2/8, h/2/2/9} 
                {
                    \node[visistedknoten] (\name) at ({\scale*\x + \xshift}, {\scale*\y}) {$\name$};
                    \node[gray] at ([xshift=4pt,yshift=-4pt]\name.south east) {\label};
                }
                
                \foreach \name/\x/\y/\label in 
                {} 
                {
                    \node[knoten] (\name) at ({\scale*\x + \xshift}, {\scale*\y}) {$\name$};
                    \node[gray] at ([xshift=4pt,yshift=-4pt]\name.south east) {\label};
                }

                \foreach \from/\to in 
                {v/a, b/e, e/d, d/c, c/f, f/g, g/h, a/b} 
                {
                    \draw[edge] (\from) -- (\to);
                }

                \draw[visitedge] (f) -- (g);
            \end{tikzpicture}
            \caption{$F_8$}
        \end{subfigure}
    \end{tabular}
    \caption{The execution of the temporal DFS algorithm starting from vertex~$v$. Visited vertices are coloured green, and unvisited vertices are coloured white. The thick edges show from which visited vertex a new vertex gets explored. In each snapshot $F_i$, the gray numbers indicate positions of the vertices in the order $\sigma_{v,F_i}$.}
    \label{fig:temporalDFS}
\end{figure}

\subsection{Temporal DFS over RST}\label{sec:tempDFSoverRST}

In this section we analyse the behaviour of temporal DFS over the Random Spanning Tree model.
We start with some notation and basic facts that will be useful in proving the main statements \cref{lem:stable-diffuse,lem:f-s,cor:F-S-new}. 

\begin{lemma}\label{lem:next-vertex}
	Let $D_1 \subseteq D_2 \subseteq [n]$, $v \in D_1$, and $F$ be a tree on $[n]$. If $\rho(v,D_1,F) \in [n] \setminus D_2$, then $\rho(v,D_1,F) = \rho(v,D_2,F)$.
\end{lemma}
\begin{proof}
	Let $u := \rho(v,D_1,F)$.
	By definition, all vertices smaller than $u$ in $\sigma_{v,F}$ are in $D_1$. Since $D_1 \subseteq D_2$, all those vertices are also in $D_2$. 
	Therefore, $u$ is the minimum vertex in $\sigma_{v,F}$ that is outside $D_2$ such that all vertices smaller than $u$ in order $\sigma_{v,F}$ belong to $D_2$. That is, by definition, $\rho(v,D_2,F) =u$.
\end{proof}

Let $(\cT, \mu)$ be an RST, and $v$ be a vertex.
For a given set $D \subseteq [n]$ such that $v \in D$, the distribution $\mu$ naturally determines a probability distribution over the next vertices reached by the temporal DFS algorithm from vertex $v$. Formally, for every $u \in [n]$,
\[
\nu_{v,D,\mu}(u) := \underset{F \sim \mu}{\pr}[\rho(v,D,F) = u]
\]
We say that a vertex $u$ is \emph{$(v,D,\mu)$-good} if $\nu_{v,D,\mu}(u) \geq \frac{1}{2n}$.

The following observation follows from the fact that for any tree $F$, any sets $D_1 \subseteq D_2 \subseteq [n]$, and any vertices $v \in D_1$ and $u \in [n] \setminus D_2$, we have that $\rho(v,D_1,F) = u$ implies $\rho(v,D_2,F) = u$, which is due to \cref{lem:next-vertex}.

\begin{observation}\label{obs:nu-monotone}
	Let $D_1 \subseteq D_2 \subseteq [n]$, $v \in D_1$, and $u \in [n] \setminus D_2$. Then $\nu_{v,D_1,\mu}(u) \leq \nu_{v,D_2,\mu}(u)$. In particular, if $u$ is $(v,D_1,\mu)$-good, then $u$ is $(v,D_2,\mu)$-good.
\end{observation}

A key property of the $(v,D,\mu)$-good vertices is that they together carry at least half of the total probability mass of $\nu_{v,D,\mu}$.

\begin{lemma}\label{lem:good-heavy}
	Let $D \subseteq [n]$ and $v \in D$. Then 
	\[
	\sum\limits_{u\text{ is } (v,D, \mu)\text{-good}} \nu_{v,D,\mu}(u) \geq 1/2.
	\]
\end{lemma}
\begin{proof} 
    If $u\text{ is  not }(v,D, \mu)\text{-good}$, then $\nu_{v,D,\mu}(u)<\frac{1}{2n}$. Thus, \[ \sum\limits_{u\text{ is } (v,D, \mu)\text{-good}} \nu_{v,D,\mu}(u) = 1 - \sum\limits_{u\text{ is not} (v,D, \mu)\text{-good}} \nu_{v,D,\mu}(u) > 1- n\cdot\frac{1}{2n} \geq 1/2,\] as $ \nu_{v,D,\mu}$ is a probability measure on $[n]$. 
\end{proof}

\begin{definition}
	Let $v \in [n]$, $t \in \mathbb{N}$, $p \in [0,1]$. The \emph{temporal DFS $(t,p)$-ball} of $v$, denoted $B_{t}^{p}(v)$, is the set of vertices that the temporal DFS algorithm reaches in $\cG \sim (\cT, \mu)$ in $t$ steps with probability at least $p$, i.e.,
	\[
	B_{t}^{p}(v) := \{w \in [n] \mid \underset{\cG \sim (\cT, \mu)}{\pr}\left[ w \in \OF_{t}^{\cG}(v) \right] \geq p\}.
	\]
\end{definition}

Since $\OF_{t}^{\cG}(v) \subseteq \Flood_{t}^{\cG}(v)$, we have the following 

\begin{observation}\label{obs:OFball-close}
	Every vertex in $B_{t}^{p}(v)$ is $(t,p)$-close to $v$. 
\end{observation}

We are now ready to prove the main property of the temporal DFS algorithm over RST. Informally, this property states that for a fixed vertex~$v$ and two independent random temporal graphs, \whp, the sets of vertices discovered from~$v$ via temporal DFS in the two temporal graphs either (1) overlap substantially, or (2) each contains a large subset of private vertices, each of which is good with respect to the opposite set, that is, it can be discovered from the opposite set in one step of the temporal DFS algorithm with probability at least $1/(2n)$.

\begin{lemma}\label{lem:stable-diffuse}
    Let $v\in [n]$ and $\cG, \cH \sim (\cT, \mu)$ be independent. Let $T \in \bN$ be such that $T \leq n$. Then, with probability at least $1-4e^{-T/500}$, we have
    \begin{enumerate}
        \item $|\OF_T^{\cG}(v) \cap \OF_T^{\cH}(v)| \geq \frac{T}{50}$, or 
        \label{close_case_1}
		
        \item $\OF_T^{\cG}(v)$ contains at least $\frac{T}{50}$ $(v,\OF_T^{\cH}(v),\mu)$-good vertices, and $\OF_T^{\cH}(v)$ contains at least $\frac{T}{50}$ $(v,\OF_T^{\cG}(v),\mu)$-good vertices.
	\label{close_case_2}
    \end{enumerate}
\end{lemma}
\begin{proof}
	To prove the lemma it is enough to show that with probability at least $1-2e^{-T/500}$ we have 
	\begin{enumerate}
		\item $|\OF_T^{\cG}(v) \cap \OF_T^{\cH}(v)| \geq \frac{T}{50}$, or
		\item $\OF_T^{\cG}(v)$ contains at least $\frac{T}{50}$ $(v,\OF_T^{\cH}(v),\mu)$-good vertices.
	\end{enumerate}
	The result then follows by applying the union bound.
	
	To facilitate the argument, we sample $\cG,\cH \sim (\cT, \mu)$ in two phases:
	\begin{enumerate}
		\item first, we run for $T$ rounds the coupling coin-flipping process defined below. This will generate the first $i$ snapshots of $\cG$ and the first $j$ snapshots of $\cH$, where $i+j=T$;
		
		\item then, we continue sampling snapshots of $\cG$ and $\cH$ independently (i.e.~without using the coupling coin-flipping process).
	\end{enumerate}
	
	\begin{boxsampling}[\small The coupling coin-flipping process]
		\small
		Repeat the following steps in rounds:
		\begin{enumerate}[]
			\item[1.] Sample snapshot $F \sim \mu$.
			\item[2.] Flip a fair coin and depending on the outcome add $F$ to $\cG$ or $\cH$.
		\end{enumerate}
	\end{boxsampling}
	
	This process results in temporal graphs $\cG$, $\cH$ independently sampled from $(\cT, \mu)$. The first phase is a convenient coupling that will allow us to prove the target result.
	
	To do this, we will reveal snapshots of $\cG$ and $\cH$ iteratively and, abusing notation, we will denote by $\OF_{\cG}$ and $\OF_{\cH}$ the sets of vertices that are reached from $v$ by the temporal DFS algorithm in $\cG$ and $\cH$, respectively, using the revealed snapshots.
	
	We will show that by the end of phase (i), with probability at least $1-2e^{-T/500}$ we have 
	\begin{enumerate}
		\item[(1)] $|\OF_{\cG} \cap \OF_{\cH}| \geq \frac{T}{25}$, or 
		\item[(2)] $\OF_{\cG}$ contains at least $\frac{T}{25}$ many $(v,\OF_{\cH},\mu)$-good vertices.
	\end{enumerate}
	We claim that this is enough to deduce the lemma. Indeed, if (1) is true at the end of the first phase, then it is also true for the second phase. Otherwise, if $u \in \OF_{\cG}$ is $(v,\OF_{\cH}, \mu)$-good by the end of the first phase, then, by the time $\cG$ and $\cH$ each acquire $T$ snapshots,
	either $u$ will be reached by temporal DFS in $\cH$, in which case it will belong to $\OF_{\cG} \cap \OF_{\cH}$, or, otherwise, by \cref{obs:nu-monotone}, it will stay $(v,\OF_{\cH}, \mu)$-good.
	Consequently, by the time both $\cG$ and $\cH$ have $T$ snapshots, we will have, as desired, that $|\OF_{\cG} \cap \OF_{\cH}| \geq \frac{T}{50}$, or $\OF_{\cG}$ has at least $\frac{T}{50}$ $(v,\OF_{\cH}, \mu)$-good vertices.
	
	Thus, in the rest of the proof, we prove the desired outcome of the first phase of the sampling process. 
	
	Consider a specific step of the coin-flipping process, and let $F$ be a snapshot drawn from $\mu$ at this step. Let $u_{\cG} = \rho(v,\OF_{\cG},F)$ and $u_{\cH} = \rho(v,\OF_{\cH},F)$, i.e., $u_{\cG}$ is the vertex that the temporal DFS algorithm adds to $\OF_{\cG}$, if $F$ is added to $\cG$, and, respectively, $u_{\cH}$ is the vertex that temporal DFS adds to $\OF_{\cH}$, if $F$ is added to $\cH$.
    
	\begin{claim}\label{cl:D_G-or-D_H}
		If $u_{\cG} \neq u_{\cH}$, then either $u_{\cG} \in \OF_{\cH}$ or $u_{\cH} \in \OF_{\cG}$.
	\end{claim}
	\begin{poc}
		If $\sigma_{v,F}(u_{\cG}) < \sigma_{v,F}(u_{\cH})$, then, by definition, $u_{\cH} = \rho(v,\OF_{\cH},F)$ implies $u_{\cG} \in \OF_{\cH}$.
		Similarly, if $\sigma_{v,F}(u_{\cH}) < \sigma_{v,F}(u_{\cG})$, then $u_{\cG} = \rho(v,\OF_{\cG},F)$ implies $u_{\cH} \in \OF_{\cG}$.
	\end{poc}

	Let $\tau$ be the number of rounds of the coin-flipping process in which the vertex $u_{\cH}$ was $(v,\OF_{\cH},\mu)$-good  at the moment of drawing the  corresponding snapshot $F$.
	\begin{claim}\label{cl:many-good}
		After $t$ rounds of the coin-flipping process, 
		$
			\pr[\tau \geq 2t/5] \geq 1 - e^{-t/100}.
		$
	\end{claim}
	\begin{poc}
		By \cref{lem:good-heavy}, at every round, the probability of $u_{\cH}$ to be $(v,\OF_{\cH},\mu)$-good is at least 1/2. Thus, $\tau$ stochastically dominates $X \sim \operatorname{Bin}(t,1/2)$.
		Therefore, by \cref{lem:simplechb} applied to $X$ with $\delta=0.2$, we have
		$
			\mathbb{P}[X \leq 2t/5 ]\leq e^{- t/100},
		$
		which implies 
		$
			\pr[\tau \geq 2t/5] \geq \mathbb{P}[X \geq 2t/5 ] \geq 1-e^{- t/100}.
		$
	\end{poc}

	By \cref{cl:many-good}, at the end of the coin-flipping process, with probability at least $1- e^{- T/100}$ we have $\tau \geq 2T/5$. Among the rounds in which $u_{\cH}$ was $(v,\OF_{\cH},\mu)$-good, denote by $\tau_1$ the number of rounds when we had $u_{\cG} = u_{\cH}$, and by $\tau_2$ the number of rounds when we had $u_{\cG} \neq u_{\cH}$. Note that $\tau = \tau_1+\tau_2$.
	We split the analysis in two cases:
	\begin{enumerate}
		\item $\tau_1 \geq \tau/2 \geq T/5$. In each of the rounds where $u_{\cG} = u_{\cH}$, the corresponding snapshot $F$ is added to $\cG$ with probability $1/2$, in which case $u_{\cG}$ is added to $\OF_{\cG}$ and at that moment $u_{\cG}$ is $(v,\OF_{\cH},\mu)$-good (because $u_{\cH}$ is $(v,\OF_{\cH},\mu)$-good and $u_{\cG} = u_{\cH}$). Thus, as before, we can conclude from \cref{lem:simplechb}, that with probability at least $1-e^{-T/500}$ the number of rounds when $u_{\cG}$ is added to $\OF_{\cG}$ and it is $(v,\OF_{\cH},\mu)$-good at that moment, is at least $2T/25$. Note that any such vertex is either added to $\OF_{\cH}$ later in the coin-flipping process, or, otherwise, by \cref{obs:nu-monotone}, stays $(v,\OF_{\cH},\mu)$-good. This implies that by the end of the coin-flipping process we have that $\OF_{\cG} \cap \OF_{\cH}$ contains at least $T/25$ vertices or $\OF_{\cG}$ contains at least $T/25$ $(v,\OF_{\cH},\mu)$-good vertices.
		
		\item $\tau_2 \geq \tau/2 \geq T/5$. In each of the rounds, in which we have $u_{\cG} \neq u_{\cH}$, by \cref{cl:D_G-or-D_H}, either $u_{\cG} \in \OF_{\cH}$ or $u_{\cH} \in \OF_{\cG}$.
		In the former case, upon the coin-flip, if $F$ is added to $\cG$, then $u_{\cG} \in \OF_{\cG} \cap \OF_{\cH}$ by the end of the round. In the latter case, if $F$ is added to $\cH$, then $u_{\cH} \in \OF_{\cG} \cap \OF_{\cH}$. Thus, the probability that $|\OF_{\cG} \cap \OF_{\cH}|$ increases by one in this round, is at least 
		\begin{align*}
			&\pr[(u_{\cG} \in \OF_{\cH})\text{ and }(F \text{ is added to } \cG)] + \pr[(u_{\cH} \in \OF_{\cG})\text{ and }(F \text{ is added to } \cH)] =\\  
			&\pr[u_{\cG} \in \OF_{\cH}] \cdot \pr[F \text{ is added to } \cG] + \pr[u_{\cH} \in \OF_{\cG}] \cdot \pr[F \text{ is added to } \cH] =\\
			&\frac{1}{2}\left( \pr[u_{\cG} \in \OF_{\cH}] + \pr[u_{\cH} \in \OF_{\cG}]\right) = \frac{1}{2}.
		\end{align*}
		Therefore, applying \cref{lem:simplechb} again, we deduce that by the end of the coin-flipping process we have that with probability at least $1-e^{-T/500}$ the number of vertices in $\OF_{\cG} \cap \OF_{\cH}$ is at least $2T/25 \geq T/25$. 
	\end{enumerate}
	
	Applying the union bound, we conclude that with probability at least $1 - e^{-T/100} - e^{-T/500} \geq 1 - 2e^{-T/500}$, by the end of the coin-flipping process, we have $|\OF_{\cG} \cap \OF_{\cH}| \geq \frac{T}{25}$, or $\OF_{\cG}$ contains at least $\frac{T}{25}$ $(v,\OF_{\cH},\mu)$-good vertices.
\end{proof}

We now use \cref{lem:stable-diffuse} to establish our next property of temporal DFS on RST. It says that for a fixed vertex~$v$, two independent random temporal graphs $\cG,\cH \sim (\cT, \mu)$, and a set $S$ of vertices that may depend on the first $T$, but is independent of the next $T$ snapshots of both graphs, either (1) with probability at least 1/3, in the first $T$ time steps, temporal DFS on $\cG$ discovers many vertices that are outside $S$, or (2) with probability at least 1/3, in the first $2T$ time steps temporal DFS on $\cH$ discovers at least one vertex from $S$.

\begin{lemma}\label{lem:f-s}
	Let $n \in \bN$ be large enough, $v\in [n]$, and $\cG, \cH \sim (\cT, \mu)$ be independent. Let $T \in \bN$ satisfy $20 \sqrt{n} \leq  T \leq n/2$.
	Let $S \subseteq [n]$ be any set that is independent of the snapshots of $\cG$ and $\cH$ from $T+1$ to $2T$, but may depend on the first $T$ snapshots of $\cG$ and $\cH$. Then
	\begin{enumerate}
		\item $\pr[\OF_{2T}^{\cH}(v) \cap S  \neq \emptyset] \geq 1/3$; or
		\label{f-s_case_1}
		\item $\pr[|\OF_{T}^{\cG}(v) \setminus S| \geq T/100] \geq 1/3$.
		\label{f-s_case_2}
	\end{enumerate}
\end{lemma}
\begin{proof}
	We start by introducing notation for the following events:
	\begin{align*}
		Q&:=\{ \OF_{2T}^{\cH}(v) \cap S\neq \emptyset\}, \\
		W &:=\{|\OF_{T}^{\cG}(v) \setminus S| \geq T/100\},\\
		Q_1&:=\{|\OF_{T}^{\cG}(v) \cap \OF_{T}^{\cH}(v)| \geq \frac{T}{50}\},\\
		Q_2&:=
		\left\{   
		\begin{matrix}
			\OF_T^{\cG}(v)\text{ contains at least }\frac{T}{50}\text{ $(v,\OF_T^{\cH}(v),\mu)$-good vertices}
		\end{matrix}
		\right\}.
	\end{align*}
	
	In this notation, our goal is to show that \begin{equation}\label{eq:f-smaingoal}
		\pr[Q] \geq 1/3\qquad \text{ or }\qquad \pr[W] \geq 1/3.
	\end{equation} 
	We will instead prove that 
	\begin{equation}\label{eq:f-sfirstgoal}
		\pr[Q \mid Q_1\cup Q_2 ]  \geq 2/5\qquad \text{ or }\qquad \pr[W\mid Q_1\cup Q_2 ] \geq 2/5.
	\end{equation} 
	Indeed, \cref{eq:f-sfirstgoal} together with the law of total probability, and the bound $ \pr[Q_1 \cup Q_2] \geq 1-4e^{-T/500}\geq 5/6$ (by \cref{lem:stable-diffuse}), imply \cref{eq:f-smaingoal}.
	
	To prove \eqref{eq:f-sfirstgoal}, we suppose that $\pr[ \lnot W \mid Q_1\cup Q_2 ] \geq  3/5$, and show that this implies $\pr[Q \mid Q_1 \cup Q_2] \geq 2/5$. Note that $\pr[\lnot W \cap (Q_1 \cup Q_2)]>0$, or else we contradict our assumption.
    Observe that, 
	\begin{equation}\label{eq:qcondit}
		\pr[Q \mid Q_1 \cup Q_2] 
		\geq  
		\pr[Q \mid \lnot W \cap (Q_1 \cup Q_2)]\cdot\pr[ \lnot W \mid   Q_1 \cup Q_2] 
		\geq \frac{3}{5} \cdot \pr[Q \mid \lnot W \cap (Q_1 \cup Q_2)]. 
	\end{equation}
	Notice that we cannot have both $\pr[\lnot W \cap Q_1] = 0$ and $\pr[\lnot W \cap (Q_2 \setminus Q_1)] =0 $, since $\pr[\lnot W \cap (Q_1 \cup Q_2)]>0 $.  We now assume that \begin{equation}\label{eq:nonzeroassumpt}\pr[\lnot W \cap Q_1] > 0 \qquad \text{and} \qquad \pr[\lnot W \cap (Q_2 \setminus Q_1)] >0 .\end{equation} Then, by the law of total probability, we have
	\begin{align}
		\pr[Q \mid \lnot W \cap (Q_1 \cup Q_2)] 
		&= 	\pr[Q \mid \lnot W \cap Q_1]  \cdot \pr[Q_1 \mid \lnot W\cap(Q_1 \cup Q_2)]\notag \\ 
		&\qquad +	\pr[Q \mid \lnot W \cap (Q_2 \setminus Q_1)]  \cdot \pr[Q_2 \setminus Q_1 \mid \lnot W\cap(Q_1 \cup Q_2)] \notag \\
		&\geq \min\{\pr[Q \mid \lnot W \cap Q_1],\,\pr[Q \mid \lnot W \cap (Q_2 \setminus Q_1)] \} \label{eq:minoftwo} .
	\end{align}
	
	We now estimate each of the terms in \eqref{eq:minoftwo} individually. However, first observe that we can assume w.l.o.g.\ that \eqref{eq:nonzeroassumpt} holds, as otherwise we just end up with a simpler expression for \eqref{eq:minoftwo} only involving one of the two conditional probabilities.  
	
	\textbf{Estimation of }$\pr[Q \mid \lnot W \cap Q_1]$. 	
	We claim that 
	\begin{equation}
		\label{eq:the1prob}\pr[Q~|~\lnot W \cap Q_1] = 1. 
	\end{equation} 
	Indeed, assuming $|\OF_T^{\cG}(v) \cap \OF_T^{\cH}(v)| \geq \frac{T}{50}$ (event $Q_1$) and $|\OF_{T}^{\cG}(v) \setminus S| < T/100$ (event $\lnot W$), it is easy to see that  $\OF^{\cH}_T(v) \cap S \neq \emptyset$, and hence, by the monotonicity of $\OF^{\cH}_t(v)$, we have $\OF^{\cH}_{2T}(v) \cap S \neq \emptyset$ (event $Q$).

	\textbf{Estimation of }$\pr[Q \mid \lnot W \cap (Q_2 \setminus Q_1)]$. 	Assuming $\OF^{\cG}_T(v)$ contains at least $\frac{T}{50}$ many $(v,\OF^{\cH}_T(v),\mu)$-good vertices (event $Q_2 \setminus Q_1$) and $|\OF^{\cG}_T(v) \setminus S| < \frac{T}{100}$ (event $\lnot W$), it follows that 
	$\OF^{\cG}_T(v) \cap S$ contains at least $\frac{T}{100}$ many $(v,\OF^{\cH}_T(v),\mu)$-good vertices.
	Then, at least one of these $(v,\OF^{\cH}_T(v),\mu)$-good vertices will be reached by the temporal DFS algorithm in $\cH$ in the next $T$ rounds, implying $\OF_{2T}^{\cH}(v) \cap S \neq \emptyset$, with probability at least 
	\[
	1 - \left(1 - \frac{T}{200n}\right)^T.\]  
	Recall that for $n\geq 1$ and $|x|\leq n$ we have $(1-\frac{x}{n} )^n \leq e^x $. Thus for $T\leq 200 n $ we have 
    \[ 
        \left(1-\frac{T}{200n} \right)^T = \left[\left(1-\frac{T}{200n} \right)^{200n} \right]^{T/200n} \leq \left[e^{-T} \right]^{T/200n} = e^{-T^2/200n}.
    \] 
	Note that $T\leq 200 n $ is satisfied as we assume $T\leq n/2$ in the statement. Thus, we have 
	\begin{equation}\label{eq:Q-Q_2}
		\pr[Q \mid \lnot W \cap (Q_2 \setminus Q_1)] \geq 1 - \left(1 - \frac{T}{200n}\right)^T \geq 1- e^{-T^2/200n} .
	\end{equation}
	
	Returning to the main thread of the proof of \eqref{eq:f-sfirstgoal}, by \eqref{eq:qcondit}, \eqref{eq:minoftwo}, \eqref{eq:the1prob}, \eqref{eq:Q-Q_2}, we have 
	\begin{align*}		
		\pr[Q \mid Q_1 \cup Q_2] 
		&\geq   
		\frac{3}{5} \cdot \pr[Q \mid \lnot W \cap (Q_2 \setminus Q_1)] \\
		&\geq 
		\frac{3}{5} \cdot (1- e^{-T^2/200n})\\
		&\geq 2/5,  
	\end{align*}
	establishing \eqref{eq:f-sfirstgoal}, where the last inequality follows since $T\geq 20\sqrt{n} $.
\end{proof}

We now apply \cref{lem:f-s} to show that if two vertices $u$ and $w$ are far from each other in $(\cT, \mu)$, then, in expectation, a constant fraction of vertices revealed by temporal DFS on $\cG \sim (\cT, \mu)$ from vertex $u$ by time step $T$ is not reachable by the same time in $\cG$ from a fixed temporal DFS ball of $w$.

\begin{lemma}\label{cor:F-S-new}
	Let $p \in (0,1/9]$, $t \in \bN$, and $T \in \bN$ satisfies $20 \sqrt{n} \leq  T \leq n/2$.
	If $u,w \in [n]$ are $(3T+4t, p)$-far from each other, then
	\[
		\underset{\cG \sim (\cT, \mu)}{\mathbb{E}}
		\left[|\OF_{T}^{\cG}(u) \setminus \Flood_{T}^{\cG}(B_t^p(w))|\right] \geq \frac{T}{300}.
	\]
\end{lemma}
\begin{proof}
	Let $\cG \sim (\cT, \mu)$, $\cH_1 := \cG_{[1,2T]}$, $\cH_2 := \cG_{[2T+1,3T]}$, $\cH_3 := \cG_{[3T+1,3T+4t]}$.
	Note that $\cH_1, \cH_2, \cH_3$ are independent of each other and can be seen as initial parts of independent temporal graphs sampled from $(\cT, \mu)$. We assume that snapshots of these temporal graphs are numbered starting from $1$.
	
	Denote $A := \OF_{2T}^{\cH_1}(u) \cap \Flood_{T}^{\overleftarrow{\cH_2}}(B_t^p(w))$. We claim that $\pr[A \neq \emptyset] < 1/3$. Suppose that $\pr[A \neq \emptyset] \neq 0$, as otherwise we are done. See \cref{fig:cones} for the construction of the graph and an illustration of the reachability of set $A$.
	
	Observe that if the set $A$ is non-empty then there is a path of length at most $2T$ from $u$ to $A$ in $\cH_1$.
	Note also that if the set $\Flood_{T}^{\cH_2}(A) \cap \Flood_{4t}^{\overleftarrow{\cH_3}}(w)$ is non-empty then it contains some vertex $x$ such that there is a path of length at most $T$ from $A$ to $x$ in $\mathcal{H}_2$, and a path of length at most $4t$ from $x$ to $w$ in $\cH_3$. 
	Consequently, if both $A$ and $\Flood_{T}^{\cH_2}(A) \cap \Flood_{4t}^{\overleftarrow{\cH_3}}(w)$ are non-empty, vertex $u$ reaches $w$ in $\cG$ by time step $3T+4t$. Recalling the assumption that $u$ and $w$ are $(3T+4t,p)$-far from each other, it follows that
    \begin{equation}
	  	p 
	  	> 
	  	\pr\left[ u \text{ reaches } w \text{ in } \cG \text{ by time } 3T+4t \right] \\ 
	   \geq 
	  \pr\left[ (A \neq \emptyset) \text{ and } (\Flood_{T}^{\cH_2}(A) \cap \Flood_{4t}^{\overleftarrow{\cH_3}}(w) \neq \emptyset) \right].
    \end{equation}
  Conditional on $A\neq \emptyset$, from the definition of $A$, we have that there is path of length at most $T$ from $A$ to $B_t^p(w)$ in $\cH_2$, i.e.,
  $\Flood_{T}^{\cH_2}(A)$ contains a vertex from $B_t^p(w)$. Thus, in any time window of length $t$, with probability at least $p$, the vertex $w$ can reach a vertex in $\Flood_{T}^{\cH_2}(A)$. By considering four of such independent time windows of length $t$ we have   
  \[\pr\left[\Flood_{T}^{\cH_2}(A) \cap \Flood_{4t}^{\overleftarrow{\cH_3}}(w) \neq \emptyset ~|~ A \neq \emptyset \right] \geq 1-(1-p)^4 > 3p, \] where the last inequality holds since $p\leq 1/9$. Putting these bounds together gives
	\begin{align*}
		p &> \pr\left[ u \text{ reaches } w \text{ in } \cG \text{ by time } 3T+4t \right] \\ 
		&\geq 
		\pr\left[ (A \neq \emptyset) \text{ and } (\Flood_{T}^{\cH_2}(A) \cap \Flood_{4t}^{\overleftarrow{\cH_3}}(w) \neq \emptyset) \right] \\
		&= \pr\left[ A \neq \emptyset] \cdot \pr[\Flood_{T}^{\cH_2}(A) \cap \Flood_{4t}^{\overleftarrow{\cH_3}}(w) \neq \emptyset ~|~ A \neq \emptyset \right]  \\ 
		&\geq \pr[A \neq \emptyset] \cdot 3p, 
	\end{align*}
	 and hence we have $\pr[A \neq \emptyset] = \pr\left[ \OF_{2T}^{\cH_1}(u) \cap \Flood_{T}^{\overleftarrow{\cH_2}}(B_t^p(w))\neq \emptyset \right] < 1/3$.
	 Therefore, from \cref{lem:f-s} applied to $\cG=\overleftarrow{\cH_2}$ and $\cH=\cH_1$ with $S = \Flood_{T}^{\overleftarrow{\cH_2}}(B_t^p(w))$  we conclude that 
	 \[
            \pr\left[ |\OF_{T}^{\overleftarrow{\cH_2}}(u) \setminus \Flood_{T}^{\overleftarrow{\cH_2}}(B_t^p(w))| \geq T/100 \right] \geq 1/3.
	 \] 
	 implying 	
	 \[
	 	\mathbb{E}\left[|\OF_{T}^{\overleftarrow{\cH_2}}(u) \setminus \Flood_{T}^{\overleftarrow{\cH_2}}(B_t^p(w))|\right] \geq \frac{T}{300}.
	 \] 
	 	
	 Since, $\overleftarrow{\cH_2}$ is distributed as $\cG_{[1,T]}$, we conclude the desired bound.
\end{proof}

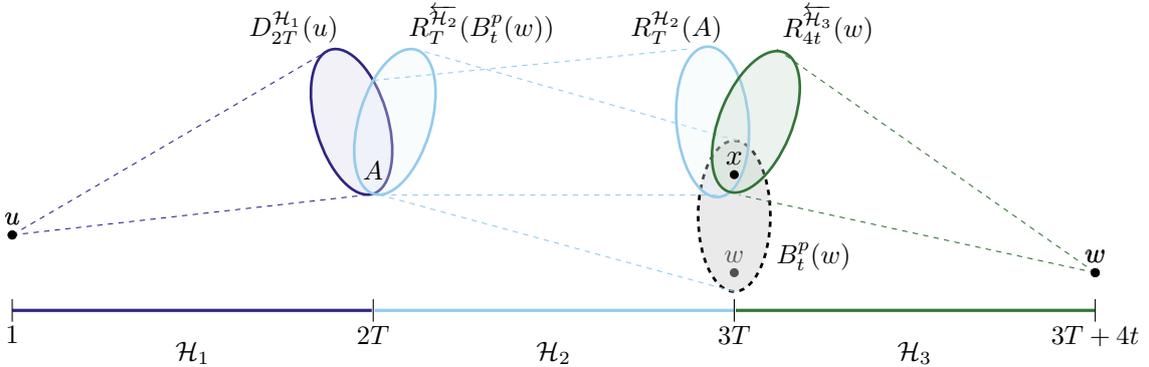
\begin{figure}[ht]
	
	\centering
	
	\begin{tikzpicture}[node/.style={circle,fill=black,inner sep=1.2pt},xscale=.95 ]
		\newcommand{\ybump}{.5}
		
		\draw[very thick, FigRed] (0,{1-\ybump})  -- (4.98,{1-\ybump});
		\draw[very thick, FigGreen] (5.02,{1-\ybump})  -- (10,{1-\ybump});
		\draw[very thick, FigBlue] (10.02,{1-\ybump})  -- (15,{1-\ybump});
		\draw (0,{0.85-\ybump}) -- (0,{1.15-\ybump}) node[pos=0.25, below] {\small 1}; 
		
		\draw (5,{0.85-\ybump}) -- (5,{1.15-\ybump}) node[pos=0.25, below] {\small $2T$}; 
		
		\draw (10,{0.85-\ybump}) -- (10,{1.15-\ybump}) node[pos=0.25, below] {\small $3T$}; 
		
		\draw (15,{0.85-\ybump}) -- (15,{1.15-\ybump}) node[pos=0.25, below] {\small $3T + 4t$}; 
		
		
		\node[above ] at (2.5,{0.15-\ybump}) {\small $\mathcal{H}_1$};
		\node[above ] at (7.5,{0.15-\ybump}) {\small $\mathcal{H}_2$};
		\node[above ] at (12.5,{0.15-\ybump}) {\small $\mathcal{H}_3$};
		
		\node[node] at (0,1.5) {};
		\node[above] at (0,1.5) {\small $u$};
		\draw [dash pattern=on 2pt off 2pt,color=FigRed] (0,1.5)-- (4.42,3.97);
		\draw [dash pattern=on 2pt off 2pt,color=FigRed] (0,1.5)-- (4.9,2.03);	
			
		\node[node] at (10,{1.5-\ybump}) { };
		\node[above] at (10,{1.5-\ybump}) {\small $w$};
		
		\draw [rotate around={-72.5:(4.7,3)},line width=1pt, color=FigRed,fill=FigRed!20, fill opacity=.30] (4.7,3) ellipse (1cm and 0.5cm);
		\draw [rotate around={72.5:(5.3,3)},line width=1pt, color=FigGreen,fill=FigGreen!20, fill opacity=.30] (5.3,3) ellipse (1cm and 0.5cm);	
		
		\node[above] at (3.9,3.9) {\small $\OF_{2T}^{\cH_1}(u)$};
		\node[above] at (6.5,3.9) {\small $\Flood_{T}^{\overleftarrow{\cH_2}}(B_t^p(w))$};
		\node[above] at (5,2.1) {\small $A$};
		
		
		\draw [line width=1.02pt,dash pattern=on 2pt off 2pt,color=black,fill=black!20, fill opacity=.4] (10,{2.25-\ybump})  ellipse (.5cm and 1cm);	
		
		\draw [dash pattern=on 2pt off 2pt,color=FigGreen] (10.05,{3.25-\ybump})-- (5.6, 3.96 );
		
		\draw [dash pattern=on 2pt off 2pt,color=FigGreen] (10,{1.25-\ybump})-- (5.1, 2.03 );

		
		\draw [dash pattern=on 2pt off 2pt,color=FigGreen] (5,3.55)-- (9.42,3.97);
		\draw [dash pattern=on 2pt off 2pt,color=FigGreen] (5.1,2.03)-- (9.9,2.03);

		\draw [rotate around={-85:(9.7,3)},line width=1pt, color=FigGreen,fill=FigGreen!20, fill opacity=.30] (9.7,3) ellipse (1cm and 0.5cm);
		\draw [rotate around={66.4:(10.3,3)},line width=1pt, color=FigBlue,fill=FigBlue!30, fill opacity=.30] (10.3,3) ellipse (1cm and 0.5cm);	
		
		\node[above] at (9.2,3.9) {\small $\Flood_{T}^{\cH_2}(A)$};
		\node[above] at (11.3,3.9) {\small $\Flood_{4t}^{\overleftarrow{\cH_3}}(w)$};
				\node[above] at (11.1,.9) {\small $B_t^p(w)$};

		\node[node] at (15,{1.5-\ybump}) { };
		\node[above] at (15,{1.5-\ybump}) {$w$};
		\draw [dash pattern=on 2pt off 2pt,color=FigBlue] (15,{1.5-\ybump})-- (10.58,3.97);
		\draw [dash pattern=on 2pt off 2pt,color=FigBlue] (15,{1.5-\ybump})-- (10.07,2.03);	
		
		\node[node] at (10,2.3) { };
		\node[above] at (10,2.3) {$x$};

            \node[node] at (0,1.5) {};
		\node[above] at (0,1.5) {\small $u$};

            \node[node] at (15,{1.5-\ybump}) { };
		\node[above] at (15,{1.5-\ybump}) {\small $w$};
	\end{tikzpicture}
	
	\caption{Illustration of setup in the proof of \cref{cor:F-S-new}.}\label{fig:cones}
\end{figure}

\subsection{Vertices with few close vertices}\label{sec:few-close-vertices}
In this section, we prove that if a vertex has only a few close vertices, then a temporal DFS ball around $v$ will have connecting edges to vertices that are far from $v$ in a constant fraction of snapshots.

Let $S_1, S_2 \subseteq [n]$ be two sets of vertices. A graph $F$ on vertex set $[n]$ is called \emph{$(S_1,S_2)$-linking} if $F$ contains an edge between a vertex in $S_1$ and a vertex in $S_2$. Note that for fixed sets $S_1, S_2$,  $(\cT, \mu)$ defines a Bernoulli random variable indicating whether a random snapshot $F \sim \mu$ is $(S_1,S_2)$-linking or not.

Now, we can use the above definition and have set $S_1$ be the vertices in the temporal DFS ball of vertex $v$, and set $S_2$ be the set of vertices that are far from $v$. The next lemma shows how often there is a direct edge between those two sets. 

\begin{lemma}\label{lem:close-ball}
	Let $p \in (0,1/9]$, $k \in \bN$, $t := 22k$, and let $\tau \geq t$.
	Let $v \in [n]$ be a vertex that has at most $k$ many $(\tau,p)$-close vertices.
	Let $B := B_t^p(v)$, and $C$ be the set of vertices that are $(\tau,p)$-far from $v$.
	Then the probability that a snapshot $F \sim \mu$ is $(B,C)$-linking is at least $1/22$.
\end{lemma}
\begin{proof}
	Let $F \sim \mu$ and $X:=X(F)$ be the Bernoulli random variable that is equal to 1 if and only if $F$ is $(B,C)$-linking.
	To prove the lemma, it is enough to show that in a sequence of $t$ snapshots sampled from $\mu$, the expected number of $(B,C)$-linking snapshots is at least $k = t/22$. Indeed, the number of $(B,C)$-linking snapshots is the sum of $t$ i.i.d. random variables with the same distribution as $X$, and thus the assumption $k = t/22 \leq t \cdot \bE[X] = t \cdot \pr[X = 1]$ would imply $\pr[X = 1] \geq 1/22$. 
	
	By \cref{obs:OFball-close}, every vertex in $B$ is $(t,p)$-close to $v$. Since $\tau \geq t$, every vertex in $B$ is also $(\tau,p)$-close to $v$, and thus by our assumption on the number of $(\tau,p)$-close vertices, we have $|B| \leq k$.

	Denote by $F_i$ the $i$-th snapshot of $\cG$. We define multisets of vertices $U$, $W$ below, followed by an informal explanation to assist the reader.  
    Let $U^0=\emptyset, W^0 = \{v\}$, and for $i=1,\ldots,t$ let 
	\begin{align*}
		x_U^i &= \rho(v,B, F_i) \\
		x_W^i &= \rho(v,W^{i-1} \cap B, F_i)\\
		U^i &= U^{i-1} \cup \{ x_{U}^i \}\\
		W^i &= W^{i-1} \cup \{ x_{W}^i \}.
	\end{align*}
	We define $U := U^t$ and $W := W^t$, and for $x \in [n]$ we denote by $U_x$ the multiplicity of $x$ in $U$, i.e., the number of steps $i \in [t]$ for which $x_U^i = x$. Similarly, $W_x$ is the multiplicity of $x$ in $W$. 
    Informally (see \cref{fig:U-W-sets} for illustration), set $U$ consists of all vertices that temporal DFS would discover by processing $F_i$, $i \in [t]$, given the current set of discovered vertices is $B$. Set $W$ in turn consists of vertices that temporal DFS would discover by processing $F_i$, $i \in [t]$, given the set of discovered vertices consists of all vertices in $B$ that were discovered by temporal DFS by processing snapshots $F_1, F_2, \ldots, F_{i-1}$.

    \begin{figure}
    \centering
    \begin{tikzpicture}
    
    \def\vradius{2pt}
    \def\cradius{3pt}
    
    \draw[FigBlue] (0,0) circle (3cm);
    \fill[green!20, fill opacity=0.1] (0,0) circle (3cm);
    \node at (-1.5,3) {\textcolor{FigBlue}{$B$}};
    
    \path[name path=circlearc] (30:3cm) arc (30:-30:3cm);
    
    \path[name path=wobble] plot[smooth cycle, tension=1] coordinates 
    {(30:3cm) 
    (1.8, 1.0) 
    (1.6, 1.5)
    (-0.8, 1.3)
    (-1.9, 0.5)
    (-1.7, -0.5)
    (-1.3, -1.5)
    (1.0, -2.0)
    (2, -1.4)
    (330:3cm)}; 
    
    \fill[red!40, fill opacity=0.1]
    plot[smooth cycle, tension=1] coordinates  
      {(30:3cm) 
       (1, 1.0) 
       (0.5, 1.5)
       (-0.8, 1.3)
       (-1.5, 0.5)
       (-1.7, -0.5)
       (-1.0, -1.5)
       (-0.3, -2.0)
       (0.5, -1.4)
       (330:3cm)};
       
    \draw[FigGreen] 
      plot[smooth cycle, tension=1] coordinates 
      {(30:3cm) 
       (1, 1.0) 
       (0.5, 1.5)
       (-0.8, 1.3)
       (-1.5, 0.5)
       (-1.7, -0.5)
       (-1.0, -1.5)
       (-0.3, -2.0)
       (0.5, -1.4)
       (330:3cm)};

    \node at (-0.5,-2.3) {\textcolor{FigGreen}{$W_{i-1} \cap B$}};
    
    \node[circle, fill=black, inner sep=2.3pt, label=below:$v$] at (0,0) {};
    
    \node[circle, fill=FigGreen, inner sep=\vradius] at (0.5, 1.7) {};
    \node[circle, fill=FigGreen, inner sep=\vradius] at (1, 1.3) {};
    \node[circle, fill=FigGreen, inner sep=\vradius] at (1, 1.3) {};
    \node[circle, fill=FigGreen, inner sep=\vradius] at (-0.5, 1.7) {};
    \node[circle, fill=FigGreen, inner sep=\vradius] at (-1.3, 1.4) {};
    \node[circle, fill=FigGreen, inner sep=\vradius] at (-1.8, 0.6) {};
    \node[circle, fill=FigGreen, inner sep=\vradius] at (-2, -0.5) {};
    \node[circle, fill=FigGreen, inner sep=\vradius] at (-1.4, -1.6) {};
    \node[circle, fill=FigGreen, inner sep=\vradius] at (0.5, -1.7) {};
    \node[circle, fill=FigGreen, inner sep=\vradius] at (1.8, -1.9) {};
    
    \node[circle, fill=FigBlue, inner sep=\vradius] at (2, 2.5) {};
    \node[circle, fill=FigBlue, inner sep=\vradius] at (1.0, 3.1) {};
    \node[circle, fill=FigBlue, inner sep=\vradius] at (-0.5, 3.2) {};
    \node[circle, fill=FigBlue, inner sep=\vradius] at (-2.5, 2.0) {};
    \node[circle, fill=FigBlue, inner sep=\vradius] at (-3.3, 0.0) {};
    \node[circle, fill=FigBlue, inner sep=\vradius] at (-2.5, -2.0) {};
    \node[circle, fill=FigBlue, inner sep=\vradius] at (-1.0, -3.1) {};
    \node[circle, fill=FigBlue, inner sep=\vradius] at (1.8, -2.6) {};
    
    \begin{scope}
      \clip (3.2, 0.8) circle (\cradius);
      \fill[FigGreen] (3.2, 0.8) circle (\cradius);
    \end{scope}
    \begin{scope}
      \clip (3.2, 0.8) -- ++(-90:\cradius) arc (-90:90:\cradius) -- cycle;
      \fill[FigBlue] (3.2, 0.8) circle (\cradius);
    \end{scope}

    \begin{scope}
      \clip (3.3, 0.0) circle (\cradius);
      \fill[FigGreen] (3.3, 0.0) circle (\cradius);
    \end{scope}
    \begin{scope}
      \clip (3.3, 0.0) -- ++(-90:\cradius) arc (-90:90:\cradius) -- cycle;
      \fill[FigBlue] (3.3, 0.0) circle (\cradius);
    \end{scope}
    
    \begin{scope}
      \clip (3.1, -1.3) circle (\cradius);
      \fill[FigGreen] (3.1, -1.3) circle (\cradius);
    \end{scope}
    \begin{scope}
      \clip (3.1, -1.3) -- ++(-90:\cradius) arc (-90:90:\cradius) -- cycle;
      \fill[FigBlue] (3.1, -1.3) circle (\cradius);
    \end{scope}

    \end{tikzpicture}
    \caption{Construction of sets $U$ and $W$. 
    The \textcolor{FigBlue}{green} and \textcolor{FigGreen}{blue}-\textcolor{FigBlue}{green} vertices outside the set $B$ are candidates (depending on the snapshot $F_i$) for being $x^i_U$. Similarly, the \textcolor{FigGreen}{blue} and \textcolor{FigGreen}{blue}-\textcolor{FigBlue}{green} vertices outside the set $W_{i-1} \cap B$ are candidates for being $x^i_W$. 
    If $x^i_W$ is \textcolor{FigGreen}{blue}, i.e., belongs to $B$, then it will be added to the set of discovered vertices used by the temporal DFS to determine $x^{i+1}_W$.
   }
    \label{fig:U-W-sets}
\end{figure}
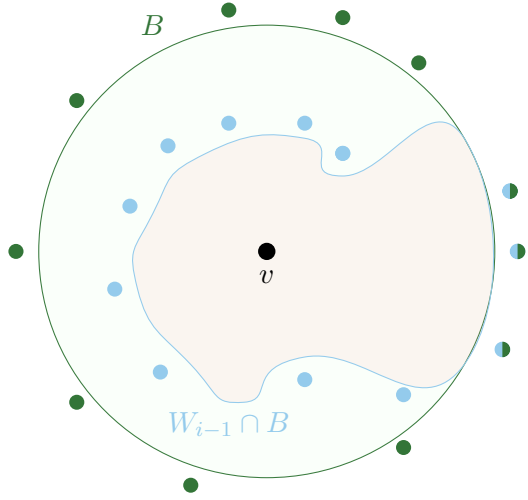
	
	Note that $U_x=0$ holds for every $x \in B$, and thus
	\begin{equation}\label{eq:sumUx}
		t = \sum_{x \in [n]} U_x = 
		\sum_{x \text{ is } (\tau,p)\text{-far from } v} U_x +
		\sum_{\substack{x \in [n] \setminus B \\ x \text{ is } (\tau,p)\text{-close} \text{ to } v}} U_x.
	\end{equation}
	Note that the first sum of the right-hand side of \cref{eq:sumUx} lower bounds the number of $(B,C)$-linking snapshots among the snapshots $F_1, F_2, \ldots, F_t$. Thus, our goal is to prove that 
	\begin{equation}\label{eq:far-goal}
		\sum_{x \text{ is } (\tau,p)\text{-far from } v} \bE[U_x] \geq k.
	\end{equation}

	Given \cref{eq:sumUx}, to achieve this goal, it is enough to show that
	\begin{equation}\label{eq:sum-Ux-close}
		\sum_{\substack{x \in [n] \setminus B \\ x \text{ is } (\tau,p)\text{-close} \text{ to } v}} \bE[U_x] \leq t-k = 21k.
	\end{equation}
	
	In the rest of the proof we establish \cref{eq:sum-Ux-close}.
	
	Since $x_U^i \not\in B$ and $W^{i-1} \cap B \subseteq B$ hold for every $i \in [t]$, it follows from \cref{lem:next-vertex} that $x_U^i \neq x_W^i$ if and only if $x_W^i \in B$. This implies $U_x \geq W_x$ for every $x \in [n] \setminus B$.
	Furthermore, note that $x_W^i \in B$ can happen in at most $|B|$ snapshots, and therefore  
	\begin{equation}\label{eq:Ux-Wx}
		\sum_{x \in [n] \setminus B} (U_x - W_x) \leq |B| \leq k.
	\end{equation}
	
	\begin{claim}\label{clm:boundnotinball}
		For every $x \in [n] \setminus B$, we have 
		$\bE[U_x] \leq 19 \cdot \bE[U_x -W_x]+2$.
	\end{claim}
	\begin{poc}
	 	Suppose, towards a contradiction, that
	 	$\bE[U_x] > 19 \cdot \bE[U_x -W_x]+2$ for some $x \in [n] \setminus B$.
	 	This implies $\bE[U_x] > 2$ as  $U_x -W_x \geq 0$. 
	 	
	 	Consider $W_x$ and write it as $W_x = U_x - (U_x-W_x)$. 
	 	Since  $U_x$ is binomially distributed, by \cref{lem:simplechb} we have 
	 	\begin{equation}\label{eq:Ux}
	 		\pr[U_x \leq \bE[U_x]/2] < e^{-(1/2)^2\bE[U_x]/2} < e^{-1/4} < 0.78.
	 	\end{equation}
	 	
	 	For $U_x-W_x$, we have 
	 	\begin{equation}\label{eq:p-Ux-Wx}
	 		\Pr{U_x-W_x \geq 9.5 \cdot \mathbb{E}[U_x-W_x]}\leq \frac{1}{9.5} < 0.106,
	 	\end{equation}
	 	by Markov's inequality.
	 	So with probability at least $1-0.78-0.106 = 0.114 > 1/9 = p$ we have 
	 	\[
	 		W_x = U_x - (U_x-W_x) \geq \frac{\bE[U_x]}{2}  - 9.5 \cdot \bE[U_x-W_x]  >  \frac{19 \cdot \bE[U_x - W_x] +2}{2} - 9.5 \cdot \bE[U_x-W_x] =  1,
	 	\] 
	 	where in the last inequality we used the assumption $\bE[U_x] > 19 \cdot \bE[U_x -W_x]+2$.
	 	
	 	Since $W$ (as a set) is a subset of $\OF_{t}^{\cG}(v)$, the above implies that $\pr[x \in \OF_{t}^{\cG}(v)] \geq \pr[x \in W] > p$, and so, by definition, $x$ belongs to $B=B_{t}^p(v)$. This is a contradiction since we assumed $x \in [n] \setminus B$.
 	\end{poc}

 	Now, using \cref{clm:boundnotinball}, \cref{eq:Ux-Wx}, and the bound on the number of $(\tau,p)$-close vertices, we obtain
 	\begin{align*}\label{eq:close}
 		\sum_{\substack{x \in [n] \setminus B \\ x \text{ is } (\tau,p)\text{-close} \text{ to } v}} \bE[U_x] &\leq
 		\sum_{\substack{x \in [n] \setminus B \\ x \text{ is } (\tau,p)\text{-close} \text{ to } v}}  (19 \cdot \bE[U_x -W_x]+2) \\
 		& \leq 19  \sum_{x \in [n] \setminus B} \bE[U_x -W_x] + 2k
 		\leq 21k,
 	\end{align*}
 	as desired.
\end{proof}

\begin{corollary}\label{cor:many-linking}
	Let $p \in (0,1/9]$, $k \in \bN$, $t := 22k$, and let $\tau \geq t$.
	Let $v \in [n]$ be a vertex that has at most $k$ many $(\tau,p)$-close vertices.
	Let $B := B_t^p(v)$ and $C$ be the set of vertices that are $(\tau,p)$-far from $v$.
	Then in a sequence of $r$ snapshots with probability at least $1-e^{-r/4400}$ there are at least $\frac{9r}{220}$ many $(B,C)$-linking snapshots.
\end{corollary}
\begin{proof}
	From \cref{lem:close-ball}, the number of $(B,C)$-linking snapshots stochastically dominates a random variable $Y \sim \operatorname{Bin}(r,1/22)$. Therefore, the result follows by applying \cref{lem:simplechb} to $Y$ with $\delta=1/10$.
\end{proof}

\subsection{Proof of \texorpdfstring{\cref{th:sqrt-close-vertices}}{Theorem \ref{th:sqrt-close-vertices}}}\label{sec:sqrt-close-vertices-proof}

With these tools at hand, we are now ready to prove our main \cref{th:sqrt-close-vertices}, which we restate for convenience.

\sqrtclosevertices*
\begin{proof}
	Let \begin{equation}\label{eq:consts}  
		t := 22\sqrt{n} ,\quad    p:=1/9 ,\quad  \tau_1,\tau_2:=100 \sqrt{n} \quad \text{and}\quad T:=3(\tau_1 +\tau_2) +4t. 
	\end{equation} 
	Note, for any $n$ large enough $t$ satisfies $20 \sqrt{n} \leq  t\leq n/2$.
	
	Observe that $T \leq 700 \sqrt{n}$ by \eqref{eq:consts} and therefore any vertex that is $(T,p)$-close to a vertex $v$ is also $(700\sqrt{n},1/9)$-close to $v$. Thus, to prove the theorem we will show that for every vertex $v \in [n]$ there are at least $\sqrt{n}$ vertices that are $(T,p)$-close to $v$.
	Suppose, towards a contradiction, there exists a vertex $w \in [n]$ with less than $\sqrt{n}$ vertices that are $(T, p)$-close to $w$.
	
	Let $\cG \sim (\cT, \mu)$, $B := B_t^p(w)$, and $C$ be the set of vertices that are $(T,p)$-far from $w$. Let $\zeta =\zeta(\cG)$ be the number of $(B,C)$-linking snapshots among the first $\tau_1$ snapshots of $\cG$,
	and let $1 \leq t_1 < t_2 < \ldots < t_{\zeta} \leq \tau_1$ be the time steps of these snapshots.
    Denote by $u_1, u_2, \ldots, u_{\zeta}$ the corresponding $(T,p)$-far vertices, i.e., for $i \in[\zeta]$, vertex $u_i$ is $(T,p)$-far from $w$ and the snapshot of $\cG$ at time step $t_i$ contains an edge that connects $u_i$ with a vertex in $B^{p}_t(w)$. 
    Note that for some time steps $t_i$, multiple such $(T,p)$-far vertices may exist, in which case we simply choose one to represent $u_i$.
	
	For $i \in [\zeta]$ define $\cG_i = \cG_{[t_i+1, \tau_1 + \tau_2]}$ and denote by $T_i := \tau_1 +\tau_2-t_i \geq \tau_2$ the lifetime of $\cG_i$. Let $X_i := \OF_{T_i}^{\cG_i}(u_i) \setminus \Flood_{T_i}^{\cG_i}(B_t^p(w))$. Since $u_i$ is $(T,p)$-far from $w$ and $T \geq 3T_i + 4t$, vertex $u_i$ is also $(3T_i +4t, p)$-far from $w$, and so by \cref{cor:F-S-new} we have 
	\begin{equation} \label{eq:lowerinX_i} 
		\bE[|X_i|] \geq \frac{T_i}{300} \geq \frac{\tau_2}{300}.
	\end{equation}
	
	Observe that for any $1 \leq i < j \leq \zeta$, since $t_i < t_j$ and the snapshot at time step $t_j$ contains an edge between a vertex in $B_t^p(w)$ and $u_j$, we have that $\OF_{T_j}^{\cG_j}(u_j) \subseteq \Flood_{T_i}^{\cG_i}(B_t^p(w))$. In particular, this implies
	\begin{equation}\label{eq:R-inc-D}
		\Flood_{\tau_1 +\tau_2}^{\cG}(B_t^p(w)) 
		\supseteq 
		\OF_{T_1}^{\cG_1}(u_1) \cup \Flood_{T_1}^{\cG_1}(B_t^p(w)) 
		\supseteq 
		\bigcup_{i=1}^{\zeta} \OF_{T_i}^{\cG_i}(u_i) 
		\supseteq 
		\bigcup_{i=1}^{\zeta} X_i.
	\end{equation}
	
	Furthermore, we claim that for any $1 \leq i < j \leq \zeta$ the sets $X_i$ and $X_j$ are disjoint. Indeed, suppose this is not the case, i.e., $X_i \cap X_j \neq \emptyset$. This assumption together with $X_j \subseteq \OF_{T_j}^{\cG_j}(u_j) \subseteq \Flood_{T_i}^{\cG_i}(B_t^p(w))$ imply that $X_i \cap \Flood_{T_i}^{\cG_i}(B_t^p(w)) \neq \emptyset$, which contradicts the definition of $X_i$. 
	
	Thus, from \cref{eq:R-inc-D} and the disjointness of $X_i$, $i \in [\zeta]$, we have 
	\[
		\bE\left[ \Flood_{\tau_1 +\tau_2}^{\cG}(B_t^p(w)) \right] \geq \bE\left[\sum_{i=1}^{\zeta}  |X_i|\right]. 
	\]
	Set $\ell =\lfloor \frac{9\tau_1}{220}\rfloor$ and then observe that,  since $|X_i|\leq n$, 
	\[
		\bE\left[\sum_{i=1}^{\zeta}  |X_i|\right] \geq  \bE\left[\left(\sum_{i=1}^{\zeta}  |X_i|\right)\mathbf{1}_{\{\zeta\geq \ell \}} \right] \geq  \sum_{i=1}^{\ell }\bE[  |X_i|\cdot \mathbf{1}_{\{\zeta\geq \ell\}} ] \geq \sum_{i=1}^{\ell}\bE[  |X_i| ] - \ell \cdot n\pr[\zeta<\ell].
	\] 
	Note that $\pr[\zeta<  \ell] \leq  e^{-\tau_1/4400}$ by  \cref{cor:many-linking}, and thus by \eqref{eq:lowerinX_i},   
	\begin{align*}	
		\bE\left[ \Flood_{2T}^{\cG}(B_t^p(w)) \right] 
		&\geq  
		\ell\cdot  \frac{\tau_2}{300} - \ell \cdot n e^{-100\sqrt{n}/4400} \\
		&\geq 
		\Big(  \frac{900\sqrt{n}}{220}-1 \Big) \cdot  \frac{\sqrt{n}}{3} - n^2 e^{-\sqrt{n}/44}\\
		&= 
		\frac{300}{220}n - O(\sqrt{n})    > n, 
	\end{align*} 
	since $n$ is large, a contradiction. 
\end{proof}

\section{Exploration in \texorpdfstring{$O(m)$}{linear in m} time}
\label{sec:linear-in-m}

A result of Erlebach, Hoffman \& Kammer \cite[Theorem 6.1]{ErlebachOnTemporal2021} implies that for any connected graph $G$ and distribution $\mu$ on the spanning trees of $G$ with $m$ edges, and  $\cG \sim (\cT, \mu)$,  there exists an online algorithm $A$ such that, for any constant $c>0$, there exists a constant $C$ such that  
\[
    \Pr{\Dexp_A(\cG) \leq C\cdot m\log n} \geq 1 - n^{-c}.
\] 
The aim of this section is to prove \cref{thm:orderm}, which shows that knowing future snapshots allows for improvements in both exploration time and success probability

\orderm*
\begin{proof}
    For each edge $e$ of $G$ define the weight of $e$ as $w_e=1/p_e$, where $p_e$ is the probability that $e$ appears in any given time step, that is \[ p_e := \sum_{T \in \cT} \mu(T)\cdot \mathds{1}(e\in E(T)).\] 
    Let $T_{\mathsf{min}}$ be any minimum-weight spanning tree of $G$ (see \cref{fig:subfig-a}). By \cite[Theorems 5.1 \& 6.1]{ErlebachOnTemporal2021} (and the proofs thereof), we have	\begin{equation}\label{eq:assumpt} w(T_{\mathsf{min}}) := \sum_{e \in E(T_{\mathsf{min}})} w_e\leq 4m. \end{equation}  We now  set two `thresholds' used throughout the proof:   
    
    \[
        \alpha:= m^{1/6} \qquad  
        \text{and} \qquad\beta:=  m^{1/4}.
    \]
    Using the first one of these we define the following edge subset of $T_{\mathsf{min}}$
    \[
        L :=\left\{e\in E(T_{\mathsf{min}}) : w_e \leq  
        \tfrac{4m}{\alpha}   \right\},
    \]
    and denote by $F$ be the spanning forest of $T_{\mathsf{min}}$ induced by the edges in $L$  (see \cref{fig:subfig-b}). 
    We call the edges in $L$ the \emph{backbone edges}, and forest $F$ the \emph{backbone forest}.
    The backbone edges have sufficiently low weights, and therefore each connected component of $F$ can be traversed efficiently. Our main goal is to establish efficient connections between  the components of $F$.
    
    From \eqref{eq:assumpt} and the definition of $L$, we have 
    \begin{equation}\label{eq:edgesinL}
        |E(T_{\mathsf{min}})\setminus L|\leq \alpha, 
    \end{equation} 
    and thus the backbone forest $F$ has at most $\alpha+1$ connected components.

For a graph with non-negative edge weights, we define the \emph{distance} between two vertices to be the weight of the lightest path between them, and the \emph{diameter} to be the maximum distance between any pair of vertices.  

First, we prove that we can remove a small number of backbone edges in order to break the connected components of $F$ up into components of diameter at most $m/\beta$ and edge weights at most $4m/\alpha^2$. We call these \emph{fast components} and will use them to establish fast connections between the components of $F$.

\begin{claim}[Partitioning into fast components] \label{clm:notmancfastcomps} 
    There exists a spanning forest of $F$ consisting of at most $\alpha^2 +1 + 12\beta$ trees (fast components), each of which has diameter at most $m/\beta$ and consists of edges with weights at most $4m/\alpha^2$. 
\end{claim}
\begin{poc}
Let $L':=\left\{e\in E(T_{\mathsf{min}}) : w_e \leq  \tfrac{4m}{\alpha^2} \right\}$ and let $F'$ be the spanning forest of $T_{\mathsf{min}}$ induced by the edges in $L'$. Observe that $L' \subseteq L$ and thus $F'$ is a spanning subgraph of $F$.
Furthermore, analogously to \cref{eq:edgesinL}, we have $|E(T_{\mathsf{min}})\setminus L'|\leq \alpha^2$, and thus the forest $F'$ has at most $\alpha^2+1$ components.

 We can further partition $F'$ so that every connected component has diameter at most $m/\beta$ by iteratively performing the following edge removal operation. Pick a connected component of $F'$ with diameter strictly larger than $m/\beta$, find a path with weight strictly larger than $m/\beta$, and remove an edge $e$ of this path whose removal breaks the path into two subpaths with weights as equal as possible. 

Given any connected component of $F'$ with diameter strictly larger than $m/\beta$ we can find a path of weight equal to its diameter. We then find an edge $e$ of this path whose removal breaks the path into two subpaths with weights as equal as possible. 

Observe that each subpath has weight at least $m/(2\beta) - 4m/\alpha^2 \geq m/(3\beta)$, since each edge in $L'$ has weight at most $4m/\alpha^2$ and $m$ is large. We continue doing this until all components have diameter at most $m/\beta$. Now, we started with at most $\alpha^2+1$ components of $F'$ and each time we removed an edge, we broke a single component into two components with diameter, and thus also total weight, at least $m/(3\beta)$. Since the total weight of all edges is at most $4m$ by \eqref{eq:assumpt}, the total number of edges removed is at most $4m/(m/(3\beta)) = 12 \beta.$ Since each edge removed adds one component and every edge in $F'$ has weight at most $4m/\alpha^2$, the claim follows. 
\end{poc}

Denote by $F_{\mathsf{fast}}$ a spanning forest of $F$ with at most  $\alpha^2 +1 + 12\beta$ trees, each of which has diameter at most $m/\beta$ and edge weights at most $4m/\alpha^2$ (see \cref{fig:subfig-c}). This exists by \cref{clm:notmancfastcomps}. 

Next, we show that one can travel swiftly within a fast component. 
For this, for vertices $x,y \in V$, and $i\geq 1$, denote by $\tau_{i}(x,y)$ the minimum number of time steps required to reach vertex $y$ starting from vertex $x$ at time step $i$, i.e.,
\[
    \tau_{i}(x,y)=
    \min\big\{t \geq i \,\mid \, y\text{ is reachable from }x\text{ in }\mathcal{G}[i,t]  
    \big\}.
\]
Observe that, if we fix $x,y \in V$, then for any time steps $i,j \in \mathbb{N}$ the random variables $\tau_{i}(x,y)$ and $\tau_{j}(x,y)$ have the same distribution. Thus, unless the specific time step is important, we will drop the dependence on the starting time step and denote $\tau_i(x,y)$ by $\tau(x,y)$.    
 
\begin{figure}[htbp]
    \begin{tabular}{ccc}
        \begin{subfigure}{0.32\textwidth}
            \begin{tikzpicture}[scale=0.45, every node/.style={circle, draw, fill=white, inner sep=1pt}, thick]
                \foreach \i/\x/\y in {
                  1/0/0, 2/-1/0, 3/-1/-1, 4/0/-1, 5/0/-2, 6/1/-2, 7/2/-3, 8/3/-1, 9/4/0, 10/5/0, 11/6/0, 12/6/1, 13/6/-1, 14/7/-1, 15/6/-2, 16/5/-3, 17/5/-5, 18/6/-4,
                  19/7/-3, 20/7/-4, 21/8/-4, 22/9/-4,
                  23/3/-9, 24/2/-8, 25/1/-7, 26/2/-4, 27/1/-5, 28/2/-5, 29/3/-4}
                  \node (v\i) at (\x,\y) {};
            
                \foreach \a/\b in {
                    1/2, 2/3, 3/4, 4/5, 5/6, 6/7, 7/8, 8/9, 9/10, 10/11, 11/12, 11/13, 13/14, 13/15, 15/16, 16/17, 17/18, 18/19, 18/20, 20/21, 21/22, 7/26, 26/27, 26/28,
                    26/29, 16/24, 23/24, 24/25,
                    1/6, 3/27, 4/9, 8/16, 29/11,
                    12/14, 16/22, 9/22, 7/25, 19/22,
                    23/16, 29/17, 1/12}
                \draw[thin, color=gray] (v\a) -- (v\b);
            
                \foreach \a/\b in {
                    1/2, 2/3, 3/4, 4/5, 5/6, 6/7, 7/8, 8/9, 9/10, 10/11, 11/12, 11/13, 13/14, 13/15, 15/16, 16/17, 17/18, 18/19, 18/20, 20/21, 21/22, 7/26, 26/27, 26/28, 26/29, 16/24, 23/24, 24/25}
                  \draw[ultra thick, color=black!40] (v\a) -- (v\b);
            \end{tikzpicture}
            \caption{(a)}
            \label{fig:subfig-a}
        \end{subfigure} &
        \begin{subfigure}{0.32\textwidth}
            \begin{tikzpicture}[scale=0.45, every node/.style={circle, draw, fill=white, inner sep=1pt}, thick]
                \foreach \i/\x/\y in {
                  1/0/0, 2/-1/0, 3/-1/-1, 4/0/-1, 5/0/-2, 6/1/-2, 7/2/-3, 8/3/-1, 9/4/0, 10/5/0, 11/6/0, 12/6/1, 13/6/-1, 14/7/-1, 15/6/-2, 16/5/-3, 17/5/-5, 18/6/-4,
                  19/7/-3, 20/7/-4, 21/8/-4, 22/9/-4,
                  23/3/-9, 24/2/-8, 25/1/-7, 26/2/-4, 27/1/-5, 28/2/-5, 29/3/-4}
                  \node (v\i) at (\x,\y) {};
            
                \foreach \a/\b in {
                    1/2, 2/3, 3/4, 4/5, 5/6, 6/7, 7/8, 8/9, 9/10, 10/11, 11/12, 11/13, 13/14, 13/15, 15/16, 16/17, 17/18, 18/19, 18/20, 20/21, 21/22, 7/26, 26/27, 26/28, 26/29, 16/24, 23/24, 24/25}
                  \draw[ultra thick, color=black!10] (v\a) -- (v\b);

                \foreach \a/\b in {
                    1/2, 2/3, 3/4, 4/5, 6/7, 7/8, 8/9, 10/11, 11/12, 11/13, 13/14, 13/15, 15/16, 17/18, 18/19, 18/20, 20/21, 21/22, 7/26, 26/27, 26/28, 26/29, 23/24, 24/25}
                  \draw[ultra thick, color=blue!40] (v\a) -- (v\b);
            \end{tikzpicture}
            \caption{(b)}
            \label{fig:subfig-b}
        \end{subfigure} &
        \begin{subfigure}{0.32\textwidth}
            \begin{tikzpicture}[scale=0.45, every node/.style={circle, draw, fill=white, inner sep=1pt}, thick]
                \foreach \i/\x/\y in {
                  1/0/0, 2/-1/0, 3/-1/-1, 4/0/-1, 5/0/-2, 6/1/-2, 7/2/-3, 8/3/-1, 9/4/0, 10/5/0, 11/6/0, 12/6/1, 13/6/-1, 14/7/-1, 15/6/-2, 16/5/-3, 17/5/-5, 18/6/-4,
                  19/7/-3, 20/7/-4, 21/8/-4, 22/9/-4,
                  23/3/-9, 24/2/-8, 25/1/-7, 26/2/-4, 27/1/-5, 28/2/-5, 29/3/-4}
                  \node (v\i) at (\x,\y) {};

                \foreach \a/\b in {
                    1/2, 2/3, 3/4, 4/5, 6/7, 7/8, 8/9, 10/11, 11/12, 11/13, 13/14, 13/15, 15/16, 17/18, 18/19, 18/20, 20/21, 21/22, 7/26, 26/27, 26/28, 26/29, 23/24, 24/25}
                  \draw[ultra thick, color=blue!40] (v\a) -- (v\b);

                \foreach \a/\b in {
                    1/2, 2/3, 4/5, 6/7, 7/8, 8/9, 10/11, 11/12, 13/14, 13/15, 15/16, 17/18, 18/19, 20/21, 21/22, 26/27, 26/28, 26/29, 23/24, 24/25}
                  \draw[ultra thick, color=orange!80] (v\a) -- (v\b);

                \foreach \a/\b in {
                    1/9, 25/28, 8/17, 15/19}
                  \draw[double, line width=0.3pt, double distance=1pt, color=green!80] (v\a) -- (v\b);
            \end{tikzpicture}
            \caption{(c)}
            \label{fig:subfig-c}
        \end{subfigure}
    \end{tabular}
    \caption{(a) The \textcolor{gray}{\textbf{bold}} edges form the spanning tree $T_{\mathsf{min}}$; (b) The \textcolor{blue!40}{\textbf{purple}} edges form the backbone forest $F$; (c) The \textcolor{orange!80}{\textbf{orange}} edges form the subgraph $F_{\mathsf{fast}}$ consisting of the fast components; and the \textcolor{green!80}{green} double edges show the meta-edges.}
    \label{fig:Om-exploration}
\end{figure}
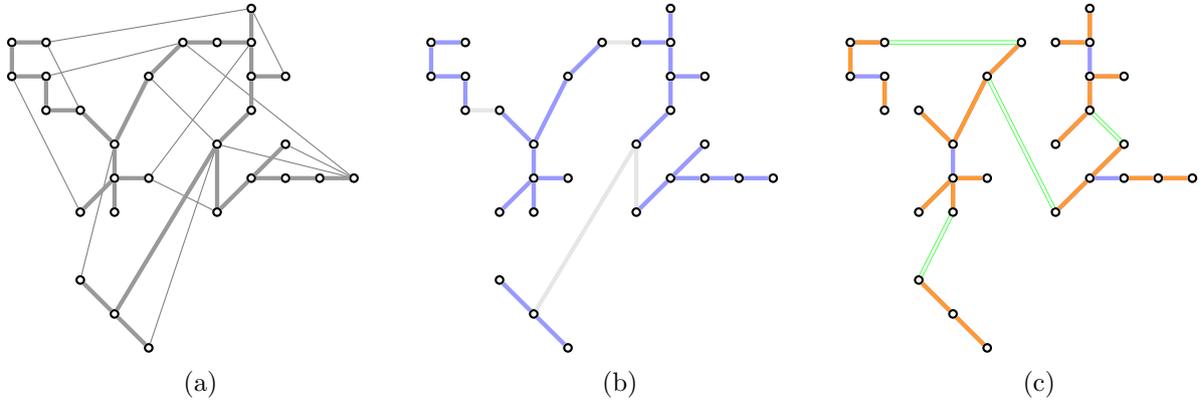

\begin{claim}[Within fast components]\label{clm:fastmovement}
    For any fast component $Q$, and any $q_s,q_t \in V(Q)$, we have
    $
        \pr\Big[\tau(q_s,q_t) >  \frac{5m}{\alpha}   \Big] \leq  e^{-\alpha}.
    $
\end{claim}
\begin{poc}
    By the definition of the fast components there exists a path $P$ in $Q$ between $q_s$ and $q_t$, such that $P$ has weight at most $m/\beta$ and each edge in this path has weight at most $4m/\alpha^2$. Thus $\tau(q_s,q_t)$ is stochastically dominated by $X= \sum_{e\in P } X_e$, where $X_e \sim \geo{1/w_e}$ are independent, with $\mu:=\Ex{X}\leq m/\beta$. It also follows that  $p_*:= \alpha^2/4m$ is a lower bound on $1/w_e$. Hence, as $\tfrac{ 5m }{\alpha \cdot \mu} \geq  \frac{5\beta}{\alpha}=\omega(1)$, for any $q_s,q_t \in V(Q)$, \cref{lem:jansontail} yields
    \[
        \mathbb{P}\Big[\tau(q_s,q_t) >  \frac{5m}{\alpha}   \Big]
    	\leq \exp( - p_* \cdot \mu \cdot (\tfrac{ 5m }{\alpha \cdot \mu}  -1- \log \tfrac{ 5m }{\alpha \cdot \mu} )) \leq\exp( - \tfrac{\alpha^2}{4m} \cdot \mu \cdot  \tfrac{ 4m }{\alpha\cdot \mu}) = e^{-\alpha},
    \]
    as claimed.
\end{poc}

Now we show that one can move fast from one part to the other in any bipartition of the fast components.
 
 \begin{claim}[Between fast components]\label{clm:fastcrossing} 
    Let $(\cA,\cB)$ be a bipartition of the connected components of $F_{\mathsf{fast}}$ such that both $\cA$ and $\cB$ are non-empty.
    Then there exists a pair $Q \in \cA$, $Q' \in \cB$ such that 
    \[ 
        \mathbb{P}\bigg[\bigcap_{i\in [m^2]} \bigcap_{\substack{q_s\in V(Q) \\ q_t' \in V(Q')}}\Big\{\tau_{i}(q_s,q_t')\leq \frac{11m}{\alpha}\Big\}\bigg] \geq 1- 3n^6\cdot e^{-\alpha}.  
    \]   
 \end{claim} \begin{poc}

 Let $Q_1,Q_2$ be two connected components of $F_{\mathsf{fast}}$. Given a snapshot of $\cG$, we denote by $\Lambda_{Q_1,Q_2}$ the event that the there exists an edge between $Q_1$ and $Q_2$ in the snapshot. 
 
 We start by showing that there exist $Q \in \cA$, $Q' \in \cB$ such that $\pr[\Lambda_{Q,Q'}] \geq 1/\alpha^4$.
 
 For this, denote by $\Lambda$ the event that a given snapshot contains an edge between a connected component in $\cA$ and a connected component in $\cB$.
 Since every snapshot of $\cG$ is connected, we have $\pr[\Lambda] = 1$.
 Furthermore, notice that
 \[
    \sum_{Q_1 \in \cA, Q_2 \in \cB} \pr[\Lambda_{Q_1,Q_2}] \geq \pr[\Lambda].
 \]
 Therefore, since, by \cref{clm:notmancfastcomps}, there are at most $\binom{\alpha^2 + 1+ 12\beta}{2} \leq \alpha^4$ pairs of distinct fast components, we conclude that there exists a pair $Q \in \cA$, $Q' \in \cB$ with 
 \begin{equation}\label{eq:QQ}
     \pr[\Lambda_{Q,Q'}] \geq \pr[\Lambda]/\alpha^4 = 1/\alpha^4,    
 \end{equation}
  as required.

    Now, from \cref{eq:QQ}, it follows that  
    \begin{equation}\label{eq:probcross}
        \mathbb{P}\bigg[\bigcap_{\substack{q\in V(Q) \\ q' \in V(Q')}}\{\tau(q,q')> m/\alpha\}\bigg]\leq (1-1/\alpha^4)^{m/\alpha}\leq e^{-m/\alpha^5} = e^{-\alpha}.
    \end{equation}

Suppose we are at any vertex $q_s\in V(Q)$ at some time step $t_0$ and we want to go to a vertex $q_t'\in V(Q')$. Our strategy is to look\footnote{Our strategy requires us to look into the future, and as we discuss later in this section, this cannot be avoided if we want to achieve an $O(m)$ time exploration.} for a time-edge crossing from some vertex $q\in V(Q)$ to some vertex $q'\in V(Q')$ between the time steps $t_0+\frac{5m}{\alpha}$ and $t_0+\frac{6m}{\alpha}$. Then we get from $q_s$ to $q$ in at most $\frac{5m}{\alpha}$ steps using edges of $Q$, wait (if needed) until we cross the time-edge from $q$ to $q'$, and finally we get from $q'$ to $q_t'$ in at most $\frac{5m}{\alpha}$ steps using edges of $Q'$. Therefore, by the union bound over all $m^2$ time steps, all pairs of vertices in $Q$, and all pairs of vertices in $Q'$ we have 
\begin{align*}
\mathbb{P}\bigg[
    \bigcup_{i\in [m^2]} 
    \bigcup_{\substack{q_s\in V(Q) \\ q_t' \in V(Q')}} 
    \Big\{ \tau_{i}(q_s,q_t') > \frac{11m}{\alpha} \Big\}
\bigg]
&\leq 
m^2 \cdot \Bigg(
    \mathbb{P}\bigg[
        \bigcap_{\substack{q\in V(Q) \\ q' \in V(Q')}} 
        \Big\{ \tau(q,q') > \frac{m}{\alpha} \Big\}
    \bigg] \\
&\qquad + 
n^2 \cdot \bigg(
    \mathbb{P}\Big[ \tau(q_s,q) > \frac{5m}{\alpha} \Big] 
    + \mathbb{P}\Big[ \tau(q',q_t') > \frac{5m}{\alpha} \Big]
\bigg)
\Bigg).
\end{align*}

Consequently, using \eqref{eq:probcross}, \cref{clm:fastmovement}, and as $m^2\leq n^4$, we have 
\[ 
    \mathbb{P}\bigg[\bigcup_{i\in [m^2]} \bigcup_{\substack{q_s\in V(Q) \\ q_t' \in V(Q')}}\Big\{\tau_{i}(q_s,q_t')> \frac{11m}{\alpha}\Big\}\bigg] 
    \leq 
    n^4\cdot \Big(e^{-\alpha} + n^2 \cdot (e^{-\alpha} + e^{-\alpha}) \Big)
    = 
    3n^6\cdot e^{-\alpha}, 
\] 
which implies the claim.  
\end{poc}

\textbf{Building exploration schedule.} 
We now build an exploration schedule using an auxiliary tree $H$, which we construct using \cref{clm:fastcrossing}.  Tree $H$ consists of the backbone edges together with some new `meta-edges' of weight $\frac{11m}{\alpha}$ which connect the components of $F$ (see \cref{fig:subfig-c}). These meta-edges represent the paths between fast components found in \cref{clm:fastcrossing}, and are constructed as follows.

Let $C_0$ be an arbitrary connected component of $F$. Define $\cD_1 := \{ C_0 \}$, and let $\cD_2$ be the set of all connected components of $F$ except $C_0$.
The set $\cD_1$ contains all components of $F$ that we connected so far in a tree using meta-edges, and $\cD_2$ contains all the remaining components of $F$. Thus, when $\cD_2$ becomes empty, we obtained the sought tree $H$. 
Until $\cD_2 \neq \emptyset$ we proceed as follows. 
Since every fast component is a subgraph of some component of $F$, the partition $(\cD_1, \cD_2)$ of the components of $F$ induces a partition $(\cA, \cB)$ of the components of $F_{\mathsf{fast}}$. Applying \cref{clm:fastcrossing} to the partition $(\cA, \cB)$ we find a pair of fast components $Q,Q'$ such that 
\whp there exists a path from any vertex in $Q$ to any vertex in $Q'$ which can be traversed in time $\frac{11m}{\alpha}$ starting at any point $i \in [m^2]$.
Let $C' \in \cD_2$ be the component of $F$ such that $Q'$ is a subgraph of $C'$.
We remove $C'$ from $\cD_2$, add it to $\cD_1$, and add to $H$ a meta-edge of weight $ \frac{11m}{\alpha}$ going between an arbitrary vertex in $Q$ and an arbitrary vertex in $Q'$.

Now, to build an exploration schedule, we follow an Euler tour of $H$ to visit all vertices of the graph. In order to estimate the travel time of this schedule, we notice that each edge of $H$ is traversed at most twice and estimate the time needed to traverse the meta-edges and the backbone edges separately. 

\textbf{Traverse time of the meta-edges.}
By \eqref{eq:edgesinL} there are at most $\alpha$ meta-edges, and thus, by \cref{clm:fastcrossing} and the union bound,
the total time of traversing meta-edges in the Euler tour is at most $2\alpha \cdot \frac{11m}{\alpha} = 22m$ with probability 
at least 
\[
    1- 2\alpha \cdot 3n^6\cdot e^{-\alpha} \geq 1- 6n^7\cdot e^{-\alpha}
\]

\textbf{Traverse time of the backbone edges.}
By \eqref{eq:assumpt}, $F$ has weight at most $4m$. Furthermore, recall that each  edge of $F$ has weight at most $m/\alpha$.
These imply that the traverse time of the edges of $F$ within the Euler tour is stochastically dominated by a sum $X$ of independent geometric random variables with expectation $\mu\leq 8m$ and minimum probability $p_*$ bounded from below by $\tfrac{\alpha}{m}$.
Thus, by applying \cref{lem:jansontail}, as before 
\begin{equation*} 
    \Pr{X \geq 32 m   } = \Pr{X \geq \frac{32 m }{\mu} \cdot  \mu  } \leq e^{ - p_* \cdot \mu \cdot \left(\frac{32 m }{\mu}  -1- \log \frac{32 m }{\mu} \right)} \leq e^{ - \frac{\alpha}{m} \cdot \mu \cdot  \frac{m}{\mu} } = e^{-\alpha}.
\end{equation*} 

It follows that the time to explore $\mathcal{G}$ following the Euler tour of $H$ is at most $22m + 32m = 54m$ with probability at least $ 1- 6n^7\cdot e^{-\alpha} - e^{-\alpha} = 1-  e^{-\Omega(m^{1/6})}$.
\end{proof}

As we noted earlier, the construction of the $O(m)$ time exploration in the proof of \cref{thm:orderm} exploits the knowledge of the future snapshots.
We now show that one cannot hope to show a general $O(m)$ exploration time bound \whp~without knowing the schedule of the edges that appear in the future.  
 
 \begin{proposition}\label{prop:online-lower-bound}    
    For any even $n\geq 4$ and  $q:=q(n)\in (0,1)$ there exists a graph $G$ with $3n-2$ edges and a distribution $\mu$ on the spanning trees $\mathcal{T}$ of $G$ such that for any randomised online exploration algorithm $A$ of $\mathcal{G}\sim (\cT, \mu)$,  
 \[
    \mathbb{P}\left[ \Dexp_A(\cG) \geq \frac{n\ln(1/q)}{4}  \right] \geq  q.
 \]  
\end{proposition} 
 
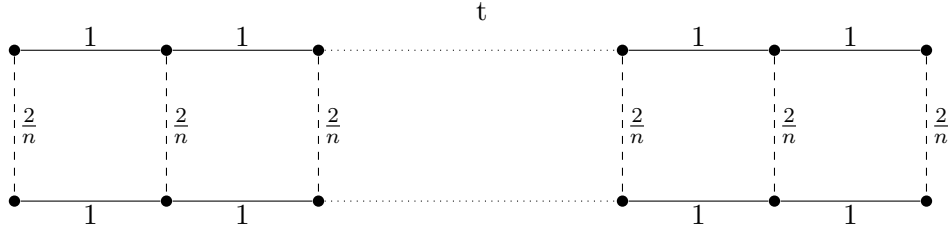
\begin{figure}{t}
    \centering

    	\begin{tikzpicture}[every node/.style={circle,fill=black,inner sep=1.5pt}, 
	edge label/.style={draw=none,fill=none,shape=rectangle}]

	\node (v1t) at (0,2) {};
	\node (v2t) at (2,2) {};
	\node (v3t) at (4,2) {};
	
	\node (v1) at (0,0) {};
	\node (v2) at (2,0) {};
	\node (v3) at (4,0) {};

	\node (vn2t) at (8,2) {};
	\node (vn1t) at (10,2) {};
	\node (vnt)  at (12,2) {};

	\node (vn2) at (8,0) {};
	\node (vn1) at (10,0) {};
	\node (vn)  at (12,0) {};

	\draw (v1t) -- node[midway, above, edge label] {$1$} (v2t);
	\draw (v2t) -- node[midway, above, edge label] {$1$} (v3t);
	\draw[dotted] (v3t) -- (vn2t);
	\draw (vn2t) -- node[midway, above, edge label] {$1$} (vn1t);
	\draw (vn1t) -- node[midway, above, edge label] {$1$} (vnt);
	
	\draw (v1) -- node[midway, below, edge label] {$1$} (v2);
	\draw (v2) -- node[midway, below, edge label] {$1$} (v3);
	\draw[dotted] (v3) -- (vn2);
	\draw (vn2) -- node[midway, below, edge label] {$1$} (vn1);
	\draw (vn1) -- node[midway, below, edge label] {$1$} (vn);

	\draw[dashed] (v1t) -- node[midway, right, edge label] {$ \frac{2}{n}$} (v1); 
	\draw[dashed] (v2t) -- node[midway, right, edge label] {$\frac{2}{n}$} (v2); 
	\draw[dashed] (v3t) -- node[midway, right, edge label] {$\frac{2}{n}$} (v3); 
	\draw[dashed] (vnt) -- node[midway, right, edge label] {$\frac{2}{n}$} (vn); 
	\draw[dashed] (vn1t) -- node[midway, right, edge label] {$\frac{2}{n}$} (vn1); 
	\draw[dashed] (vn2t) -- node[midway, right, edge label] {$\frac{2}{n}$} (vn2); 
 
\end{tikzpicture}
    \caption{Ladder example}
    \label{fig:ladder}
\end{figure}

 \begin{proof} Let $G$ be the ladder on $n=2k$ vertices consisting of two $k$-vertex paths, the first has vertex set $[k]$, and the second $[k+1, 2k]$. For each $i\in [k]$, let $T_i$ be the spanning tree of $G$ consisting of both paths and a single edge, which we call a bridge, between vertices $i$ and $k+i$.
 Let $\mu(T)=1/k$ if $T=T_i$ for some $i\in [k]$ and $\mu(T)=0$ otherwise, see \cref{fig:ladder}. Let $\mathcal{G}\sim (\cT, \mu)$, where $\cT$ is the set of spanning trees of $G$.
 	

    We interpret a randomised online exploration algorithm $A$ as an explorer on the vertices of $\mathcal{G}$.
    Observe that in order to visit all vertices of $\mathcal{G}$ the explorer must cross from one path to the other at some point. Crossing requires a bridge at the vertex where the explorer currently resides.
%
%
    Thus, the exploration time $\Dexp_A(\cG)$ is lower bounded by the time it takes for the explorer to catch a rabbit moving along the path with vertex set $[k]$, where each time step the rabbit jumps to a vertex chosen independently and uniformly at random. In the remainder of the proof, we establish a lower bound on the time required to catch the rabbit.



    We associate with the rabbit a sequence $(b_t)_{t \geq 1}$, where each $b_t \in [k]$ is the rabbit's position at time step $t$.
%
    We represent the explorer's movement on the path
    by a sequence $(a_t)_{t \geq 1}$, where $a_1\in [k]$ is the explorer's initial position (chosen by the algorithm without seeing the rabbit), 
    and, for each $t \geq 2$, the value $a_t \in [n]$ denotes the vertex to which the explorer moves at the end of time step $t-1$.   
    For every $t \geq 1$, the decision $a_t$ at time step $t-1$ is based on the explorer’s movement history $(a_1, a_2, \ldots, a_{t-1})$, the rabbit's positions up to time $t-1$, represented by $(b_1, b_2, \ldots, b_{t-1})$, and the random bits used by the algorithm. 
    The explorer catches the rabbit once $a_t=b_t$.
 	
 	By \cref{Prop:JimmysPrinciple}, the number of steps $t$ such that $a_t=b_t$ within an interval of length $T$ is distributed as $\operatorname{Bin}(T,1/k)$. 
    So the time it takes the explorer to catch the rabbit is geometrically distributed with parameter $1/k$. 
    However, for any integer $\lambda>0$ 
 	\[
 	\mathbb{P}[\geo{1/k} > \lambda ] =  (1 - 1/k)^{\lambda } = \exp\Big(\lambda \ln(1-1/k)\Big) \geq \exp\Big(\lambda \cdot \frac{- 1/k}{1 -1/k}\Big)\geq \exp(-2\lambda/k), 
 	\]
 	where we have used the inequality $\ln(1 + x)\geq \frac{x}{1+x}$, which holds for all $x>-1$, and the assumption that $k\geq 2$. If we let $\lambda = \frac{k\ln(1/q)}{2} = \frac{n\ln(1/q)}{4}$, then the time to catch the rabbit, and therefore $\Dexp_A(\cG)$, is greater than $\lambda$ with probability at least $q$, as claimed.
\end{proof} 
 
By plugging in $q=1/\operatorname{poly}(n)$ in \cref{prop:online-lower-bound}, we conclude the existence of sparse graphs, i.e., graphs with $m=O(n)$, for which any online algorithm needs time $\Omega(n\log n) = \Omega(m \log m)$ to explore the graphs with certainty $1-n^{-\Omega(1)}$. This shows that \cite[Theorem 6.1]{ErlebachOnTemporal2021} is best possible for an online explorer. 

Furthermore, taking $q=e^{-\Omega(n^{1/6})}=e^{-\Omega(m^{1/6})}$ in \cref{prop:online-lower-bound}, shows that, for sparse graphs, for a probability guarantee matching the one for the $O(m)$ exploration time of an offline explorer given in \cref{thm:orderm}, an online explorer needs time at least $\Omega(m^{7/6})$.

\section{Random Temporal Graphs on Stars}\label{sec:uniform}

This section contains several results derived from examples of the RST model $(\mathcal{S},\unif)$ where the set $\mathcal{S}$ consists of stars, or `star-like' graphs, and $\unif$ is the uniform measure. In \cref{sec:arbstars} we prove \cref{thm:urst-some-stars} which determines the time to explore $\mathcal{G}\sim(\mathcal{S},\unif)$, where $\mathcal{S}$ is a set of $k$ distinct $n$-vertex stars. In \cref{sec:lowerdeg} we consider  a `star-like' RST with a tunable maximum degree that is difficult to explore. In \cref{sec:onlinestars} we show that w.h.p.\ any \textit{online} explorer needs time $\Omega(n^2)$ to explore $\mathcal{G}\sim(\mathcal{S},\unif)$ when $\mathcal{S}$ is a set of $\lfloor n/2\rfloor$ distinct $n$-vertex stars, which is larger than the $O(n^{3/2})$ bound on the exploration time guaranteed by \cref{thm:general-upper}. We begin in \cref{sec:bottleneck} with some results and notation that play a key role in several of the proofs. 

\subsection{Auxillary Results}\label{sec:bottleneck}
Given a star, we call the unique vertex with degree greater than one its \emph{centre}. Given a set $\mathcal{S}$ of stars on vertex set $[n]$, let $K$ denote the set of centres of the stars in $\mathcal{S}$, and $L = [n] \setminus K$ denote the set of non-centres. Of course, depending on which star is sampled in a particular snapshot, it could be that at that time step a vertex $v\in K$ is currently a leaf (thus itself a `non-centre'), however it is never the case that a vertex $u \in L$ can be the centre of a star.

We will establish two lemmas and introduce some notation that is of use throughout this section.  
The first one is at the heart of our arguments and is used implicitly in \cite{ErlebachOnTemporal2021}. 
In this lemma, we use a slightly more general notion of ``stars'' that allows for isolated vertices. Let $\cS^*$ be a set of graphs on a fixed vertex set $[n]$, where each graph contains exactly one vertex of degree greater than one; that is, each graph in $\cS^*$ consists of a star together with a set of isolated vertices. We refer to each such graph as a \emph{generalised star}, and the unique vertex of degree greater than one as the \emph{centre} of the generalised star. Moreover, we assume that for each vertex $v \in [n]$, there is at most one graph in $\cS^*$ with centre $v$. Let $K \subseteq [n]$ denote the set of vertices $v$ for which $\cS^*$ contains a generalised star centred at $v$, and define $L := [n] \setminus K$.

\begin{lemma}\label{lem:pathsthroughC} 
    Let $\mathcal{G}$ be a temporal graph where each snapshot is from $\cS^*$, and assume $|L|\geq 2$. Let $P$ be any temporal walk between two distinct vertices in $L$.
    Then, there exists a vertex $v \in K$ that is the generalised star centre in snapshots at two distinct time steps between the times of the first and the last time-edges of $P$.
\end{lemma}
\begin{proof}
    Denote by $s$ the number of time-edges in $P$.
    An edge between two vertices exists in a snapshot only when one of them is the centre of the generalised star in that snapshot. Hence, for each time-edge in $P$, one of its endpoints has to be the centre of the generalised star in the corresponding snapshot. As both start and end vertices of $P$ are non-centres, there are $s$ distinct time steps corresponding to the time-edges of $P$ in which one of the $s-1$ internal vertices of $P$ is the generalised star centre. Thus, by the pigeonhole principle, one of these $s-1$ vertices is the generalised star centre in at least two snapshots at distinct time steps between the times of the first and the last time-edges of $P$. 
\end{proof}    

Let $\cG$ be a temporal graph of infinite lifetime with snapshots from $\cS$, and assume without loss of generality that the set $K$ of centres is $[k]$. Associate with $\cG$ a sequence $\cX(\cG) = (X_i)_{i \geq 1}$, where $X_i \in [k]$ is the centre of the star in the $i$-th snapshot of $\cG$.

Given a sequence $\cX(\cG)$, let $t_0 = 0$, and for every $i \geq 1$, define 
\begin{equation}\label{eq:t_i}
    t_i = \min\{t  : \textup{ there exists }t_{i-1}< a<b\leq t \textup{ with }X_a=X_b \}. 
\end{equation}
In words, $t_i$ is the minimum index $t$ such that the subsequence $(X_{1}, X_{2}, \ldots, X_{t})$ can be partitioned into $i$ intervals each of which contains an element from $[k]$ twice. These intervals are $(X_{t_s+1}, X_{t_s+2}, \ldots, X_{t_{s+1}})$, $s \in \{0, 1, \ldots, i-1 \}$.

Given a set of vertices $R \subseteq [n]$, denote by $T_{\cG}(R)$ the shortest time needed to visit all vertices in $R$ in temporal graph $\cG$ starting from a worst case vertex in $[n]$.
The next lemma shows that $T_{\cG}(L)$ is sandwiched between $t_{\ell-1}$ and $t_{\ell}$, where $\ell = |L|$ and $t_{\ell-1},t_{\ell}$ are defined in \eqref{eq:t_i}.

\begin{lemma}\label{lem:leaves-exp}
    Let $\mathcal{G}$ be a temporal graph where each snapshot is from $\mathcal{S}$, and assume $\ell := |L|\geq 2$.
    Then 
    \[
        t_{\ell-1} \leq T_{\cG}(L) \leq t_{\ell}
    \]
\end{lemma}
\begin{proof}
    We begin with the lower bound. Let $v_1, v_2, \ldots, v_{\ell}$ be the vertices in $L$, listed in the order they are visited by a temporal walk $P$ with earliest arrival time. For each $i \in [\ell-1]$, let $P_i$ denote the subwalk of $P$ whose first time-edge is the one immediately following the time-edge that visits $v_i$ for the first time, and whose last time-edge is the one that visits $v_{i+1}$ for the first time.
    
    By \cref{lem:pathsthroughC}, between the times of the first and the last time-edges of each subwalk $P_i$ there is a vertex that becomes a centre of a snapshot in two distinct time steps.
    This implies that in the prefix of the sequence $\cX(\cG)$ consisting of the first $T_{\cG}(L)$ elements, there are $\ell - 1$ disjoint intervals, each of which contains some element appearing at least twice. Consequently, $t_{\ell-1} \leq T_{\cG}(L)$.

    We now prove the upper bound.
    For $i \in [\ell]$, consider the temporal graph $\cH_i := \cG_{[t_{i-1}+1, t_{i}]}$. We claim that for any $u \in [n]$ and any $w \in L$, the temporal graph $\cH_i$ contains a temporal path from $u$ to $w$. Indeed, let $v \in K$ be a vertex that becomes the star centre in two distinct snapshots of $\cH_{i}$, and let $s_1,s_2$, where $t_{i-1}+1 \leq s_1 < s_2 \leq t_{i}$, be the corresponding time steps. Then $((u,v),s_1), ((v,w),s_2)$ is a temporal path from $u$ to $w$ in $\cH_i$, as claimed. 
    We can thus use $\cH_1, \cH_2, \ldots, \cH_{\ell}$ to visit the verities in $L$ one by one by using the first $t_{\ell}$ snapshots of $\cG$, which implies $T_{\cG}(L) \leq t_{\ell}$.
\end{proof}

\subsection{Arbitrary Set of Stars}\label{sec:arbstars}

\renewcommand{\xshift}{0}
\renewcommand{\scale}{1.5}
\renewcommand{\snapshotscale}{0.75}
\tikzset{
  knoten/.style={
    draw=black,
    circle,
    thick,
    minimum size=0.45cm,
    inner sep=0pt,
    fill=white,
    text centered
  },
  visistedknoten/.style={
    knoten,
    fill=green!30
  },
  edge/.style={
    gray
  },
  visitedge/.style={
    line width=2pt,
    green!70!black
  }
}

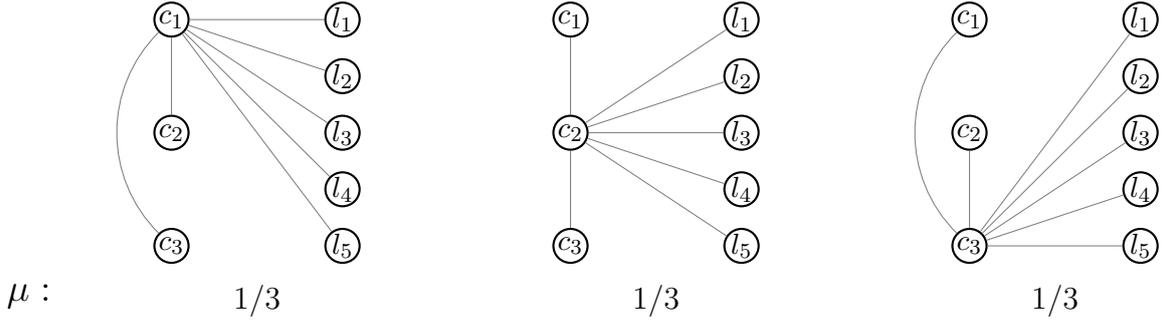
\begin{figure}
				\centering
				\begin{tikzpicture}[scale=\snapshotscale]
					
					\foreach \vertex/\name/\x/\y in 
					{c1/$c_1$/0/2, c2/$c_2$/0/0, c3/$c_3$/0/-2, l1/$l_1$/3/2, l2/$l_2$/3/1, l3/$l_3$/3/0, l4/$l_4$/3/-1, l5/$l_5$/3/-2} 
					{
						\node[knoten] (\vertex) at (\x, \y) {\name};
					}
					
					\foreach \from/\to in 
					{c1/c2, c1/l1, c1/l2, c1/l3, c1/l4, c1/l5} 
					{
						\draw[edge] (\from) -- (\to);
					}
                    \draw[edge] (c1) to[out=-135,in=135] (c3);
                    
					\node[below] at (-2.5+\xshift,-2.5) {{\Large $\mu:$}};
						\node[below] at (1.5+\xshift,-2.5) {{\large $1/3$}};
			\renewcommand{\xshift}{7}
					
				\foreach \vertex/\name/\x/\y in 
					{c1/$c_1$/0/2, c2/$c_2$/0/0, c3/$c_3$/0/-2, l1/$l_1$/3/2, l2/$l_2$/3/1, l3/$l_3$/3/0, l4/$l_4$/3/-1, l5/$l_5$/3/-2} 
				{
					\path (\x,\y) ++(\xshift,0) node[knoten] (\vertex) {\name};
				}
					
					\foreach \from/\to in 
					{c2/c1, c2/c3, c2/l1, c2/l2, c2/l3, c2/l4, c2/l5} 
					{
						\draw[edge] (\from) -- (\to);
					}
					\node[below] at (1.5+\xshift,-2.5) {{\large $1/3$}};
			
			\renewcommand{\xshift}{14}
					
				\foreach \vertex/\name/\x/\y in 
					{c1/$c_1$/0/2, c2/$c_2$/0/0, c3/$c_3$/0/-2, l1/$l_1$/3/2, l2/$l_2$/3/1, l3/$l_3$/3/0, l4/$l_4$/3/-1, l5/$l_5$/3/-2} 
				{
					\path (\x,\y) ++(\xshift,0) node[knoten] (\vertex) {\name};
				}
					
					\foreach \from/\to in 
					{c3/c2, c3/l1, c3/l2, c3/l3, c3/l4, c3/l5} 
					{
						\draw[edge] (\from) -- (\to);
					}
                    \draw[edge] (c1) to[out=-135,in=135] (c3);
					\node[below] at (1.5+\xshift,-2.5) {{\large $1/3$}};
				\end{tikzpicture}
				
				\caption{Distribution $\mu$ for \Cref{thm:urst-some-stars} with $k=3$ and $n=8$. The construction is a randomised version of \cite[Figure 1]{ErlebachOnTemporal2021}.}
                \label{fig:stars}
			\end{figure}
The aim of this subsection is to establish fairly tight bounds on the exploration time of arbitrary sets of $k$ distinct $n$-vertex stars. The RST $(\cS, \unif)$ we consider  is a randomised version of the lower-bound construction by Erlebach, Hoffmann, and Kammer~\cite[Lemma 3.1]{ErlebachOnTemporal2021}, which was used to show that the $O(n^2)$ upper bound on the exploration time of deterministic always-connected temporal graphs is tight.  
An adversarial star-based construction was also used by Avin, Kouck{\'y}, and Lotker  \cite{AvinKL18} for temporal graphs which are hard to cover by random walks. 

The main result (\cref{thm:urst-some-stars}) is somewhat technical, so we will defer its statement to later in the section, however we give the following weaker, yet cleaner, corollary here.

\begin{restatable}{theorem}{weakstars}\label{thm:weak-stars}
   Let $\cS$ be a set of $k$ distinct $n$-vertex stars, $\cG \sim (\cS,\unif)$,  and $T=(n -k)\sqrt{\tfrac{\pi  }{2}k }$. Then, for any $\varepsilon>0$ there exists some $N_\varepsilon $ such that if  $N_\varepsilon \leq  k  \leq n - N_\varepsilon \sqrt{n}\ln n$, then    
   \[
    \Pr{  \left| \texp(\cG)  -T  \right|\leq \varepsilon \cdot T} \geq 1- \varepsilon.
    \]
\end{restatable}


This statement is far weaker than \cref{thm:urst-some-stars} as the latter holds for all non-trivial $k$ and gives explicit error bounds in terms of $n$ and $k$. \cref{thm:urst-some-stars} also shows that for large $k$, above the range considered in \cref{thm:weak-stars}, the behaviour changes and the time to explore the temporal graph is w.h.p.\ $\Theta(n\log n)$. The motivation for proving such a precise result is that, due to the importance of stars in the adversarial constructions \cite{ErlebachOnTemporal2021},  it is interesting to determine their precise behaviour in the random setting. 
The choice of $k$ maximising the bound is around  $k=\lceil n/3 \rceil$, 
and a more careful application of \cref{thm:urst-some-stars} gives the following corollary. \begin{theorem}\label{cor:urst-half-stars}
    Let $\cS$ be a set of $\lceil n/3 \rceil$ distinct  $n$-vertex  stars, and let $\cG \sim (\cS,\unif)$. Then,  \[\Pr{  \left|\texp(\cG) - \sqrt{\tfrac{ 2\pi}{27} }    \cdot n^{3/2} \right|\leq  n\ln^2 n } \geq 1- n^{-\omega(1)} .\] 
\end{theorem}

\cref{cor:urst-half-stars} shows that the $O(n^{3/2})$ bound on the exploration time of random temporal graphs given in \cref{thm:general-upper} is tight. Interestingly this does not match the worst case choice of $k$ in the star-based adversarial construction from \cite{ErlebachOnTemporal2021}, which is around $k=n/2$. We conjecture that \cref{cor:urst-half-stars} is essentially the worst case for exploration of RST's. 

\begin{conjecture}\label{conj:half-stars}
  Let $\cT$ be a set of trees on vertex set $[n]$, $\mu$ be a probability distribution on $\cT$, and $\cG \sim (\cT, \mu)$. Then, for any $\varepsilon >0$ there exists some $N_\varepsilon$ such that for any $n\geq N_\varepsilon$ 
  \[\Pr{ \texp(\cG) \leq  (1+\varepsilon) \cdot \sqrt{\tfrac{2 \pi}{27} }   \cdot n^{3/2}   } \geq 1- \varepsilon .\] 
\end{conjecture}

In order to prove \cref{thm:urst-some-stars}, and thus  \cref{thm:weak-stars} and \cref{cor:urst-half-stars},  we estimate the time required to explore the centres and non-centres separately. In particular, for the upper bound, we first explore the centres and then the non-centres. For the lower bound, it follows that any exploration takes at least as long as exploring either the centres or the non-centres alone.

\paragraph{Time to visit non-centres.} 
According to \cref{lem:leaves-exp}, the time $T_{\cG}(L)$ needed to explore the non-centres is essentially determined by $t_\ell$, as defined in \eqref{eq:t_i}. Letting $\tau_i = t_i - t_{i-1}$, where the $t_i$'s are given by \eqref{eq:t_i}, we observe that the $\tau_i$ are i.i.d.~random variables and that $t_\ell = \sum_{i=1}^\ell \tau_i$.

Roughly speaking, $\tau_1$ corresponds to the number of people needed in a room until two share the same birthday, assuming there are $k$ days in a year. An asymptotic expansion for $\Ex{\tau_1}$ is known (see, for example, \cite[Section 1.2.11.3]{KnuthVol1}), but for our purposes, the following explicit bounds from \cite[Theorem 2]{Wie05} are both sufficient and more convenient.

\begin{equation}\label{eq:Etau}
    \sqrt{\frac{\pi k}{2}} - \frac{2}{5} 
    \leq 
    \Ex{\tau_1} 
    \leq \sqrt{\frac{\pi k}{2}} + \frac{8}{5}.
\end{equation}

The next lemma shows that for a fixed $\ell \in \bN$, the value $t_\ell$ concentrates around $\ell \cdot \Ex{\tau_1}$.

\begin{lemma}\label{lem:birthday}
    There exists $c>0$ such that, for any $\ell \in \bN$ and $t\ge 0$, 
        \[ 
            \displaystyle{\Pr{\big|t_\ell - \ell\cdot \Ex{\tau_1}\big| \geq t   } \leq 2 \exp\left( \frac{-ct^2}{\ell k } \right) }.
        \]
\end{lemma}
\begin{proof}
    For $t\geq 1$ the probability that no element of $[k]$ appears twice in $(X_i)_{i=1}^t$ is  
    \begin{equation}\label{eq:birthdayupper}
        \Pr{\tau_1 >  t} =	\prod _{i=1}^{t-1}  \left(1-{\frac {i}{k}}\right)\leq  \prod_{i=1}^{t-1} \exp\Big(-\frac {i}{k}\Big) = \exp\Big(- \sum_{i=1}^{t-1} \frac {i}{k}\Big) =\exp\Big(-{\frac {t(t-1)}{2k}}\Big),  
    \end{equation}  
    where we have used the inequality $x \leq  1+x \leq \exp(x) $ that holds for all $x$.
    
    Observe, if $t \geq 3$, then $(t-1)(t-2) \leq t^2/5 $; and if $ t \in [0,2]$, then $2 e^{-t^2/10k} \geq 2 e^{-4/10}  \geq 1 $. Thus, for all $t \geq 0$,   
    \[ 
        \Pr{\tau_1 \geq t } \leq  2 \exp\Big(-\frac{t^2}{10k}\Big).
    \]     
    From this and \cite[Proposition 2.6.1]{Vershynin} it follows that the random variable $\tau_1$ is sub-Gaussian, and there exists some universal constant $C$ such that the sub-Gaussian norm (see \eqref{eq:subgaussnorm} in \cref{sec:furtherprelim} for the definition) satisfies 
    \[
        ||\tau_1||_{\psi^2} \leq C\sqrt{k}.
    \]  
    For each $i\geq 1$, we let $\tau'_i=\tau_i-\Ex{\tau_i}$ be the centred version of $\tau_i$, thus all $\tau_i'$ are i.i.d.~with $\Ex{\tau_i'}=0$. Note that  centering a random variable cannot increase its sub-Gaussian norm by more than a constant \cite[Lemma 2.7.8]{Vershynin}. Thus, there exists a universal constant $C'$ such that 
    \[ 
        ||\tau_1'||_{\psi^2}  \leq C' \sqrt{k}.
    \]
    It follows from the sub-Gaussian Hoeffding inequality \cref{thm:subGHoff} that there exists some universal constant $c>0$ such that, for any $t\geq 0$, 
    \[
        \Pr{\Bigg|\sum_{i=1}^{\ell}\tau_i - \ell\cdot \Ex{\tau_1}\Bigg| \geq t   } = \Pr{\Bigg|\sum_{i=1}^{\ell}\tau_i'\Bigg| \geq t   } \leq 2 \exp\left( \frac{-c t^2}{\ell k} \right),  
    \] 
    and thus the result follows since $ t_\ell=\sum_{i=1}^{\ell}\tau_i$. 
\end{proof}

\begin{corollary}\label{clm:covingcentres}
    Let $\ell := |L|$. There exists a universal constant $c>0$ such that for any~$t\ge 0$, 
    \[ 
		\Pr{\Big|T_{\cG}(L) - \ell \cdot \Ex{\tau_1}\Big| \geq t + k +1  } \leq 2 \exp\left( \frac{-ct^2}{\ell k } \right). 
    \]
\end{corollary}
\begin{proof} 
    By \cref{lem:leaves-exp}, we have $t_{\ell-1} \leq T_{\cG}(L) \leq t_{\ell}$. Note, since any interval of the sequence $\cX(\cG)$ of length at least $k+1$ must contain some vertex from $[k]$  twice, we have that $t_{\ell} - (k+1) \leq t_{\ell-1}$,
	which implies that $|T_{\cG}(L) - t_{\ell}| \leq k+1$. 
	By the triangle inequality, we have 
	\[
		\big| T_{\cG}(L) - \ell \cdot \Ex{\tau_1 } \big| \leq |T_{\cG}(L) - t_{\ell}| + |t_{\ell} - \ell \cdot \Ex{\tau_1}| \leq k+1 + |t_{\ell} - \ell \cdot \Ex{\tau_1}|.
	\]
	Thus, if $t + k +1 \leq \big| T_{\cG}(L) - \ell \cdot \Ex{\tau_1 } \big|$, then $|t_{\ell} - \ell \cdot \Ex{\tau_1}| \geq t$.
	The result then follows directly from \cref{lem:birthday}.
\end{proof}

\paragraph{Time to visit centres.}
The time $T_{\cG}(K)$ required to visit the centres is essentially determined by the time needed to collect all coupons in the classical coupon collector problem. The lower bound on this time (\cref{clm:coveringleaves}) is not immediate. In order to establish it, we will need the following technical lemma.

 \begin{lemma}\label{lem:coupon} 
    Let $k \geq 2$, and let $(X_i)_{i=1}^{\infty}$ be a sequence of independent random variables each sampled uniformly from~$[k]$.
    For $t\geq 0$, define 
    \[
        U(t) := \Bigg| [k]\setminus \bigcup_{i=1}^t\{X_i\} \Bigg|
    \]
    to be the number of elements of~$[k]$ which do not occur in $(X_i)_{i=1}^{t}$. Then, for any $a \geq 2$, $b\in(0,1/2)$, and $\varepsilon\in(0,1)$,  
\begin{enumerate}
    \item\label{itm:couponup} $\displaystyle{\Pr{ \big| U(ak\ln k) \big|> 0} \leq  k^{- a/2}}$; and     
    \item\label{itm:couponlow} if we let $\vartheta = b k\ln k \geq 1$, then  
    \[
        k^{1-b}/2 \leq \Ex{U(\vartheta)} = k(1-1/k)^{\vartheta}\leq k^{1-b}, 
    \]
    and 
    \[
        \Pr{ \big|U(\vartheta) - \Ex{U(\vartheta)} \big| > \varepsilon \Ex{U(\vartheta)} }\leq    2\exp\bigg( -  \frac{ \varepsilon^2k^{1-2b} }{ \ln k} \bigg).
    \]
\end{enumerate}   
\end{lemma}
\begin{proof}
For $s = x k \log k$, and a fixed $i\in [k]$, the probability that $i$ does not occur in $(X_i)_{i=1}^{s}$ is 
\begin{equation}\label{eq:i-not-occur}
    \Pr{i \in U(s)} = (1-1/k)^{s}= (1-1/k)^{xk\ln k} \leq k^{-x}.
\end{equation}

Now, since $k^{-a}\leq k^{-a/2 -1}$, Item \ref{itm:couponup} follows from \cref{eq:i-not-occur} and a union bound over all $i\in [k]$.    

For Item \ref{itm:couponlow}, it follows from \cref{eq:i-not-occur} that $\Ex{U(\vartheta)} = k(1-1/k)^{\vartheta}$.
Furthermore, 
\[ 
    k^{-b} \geq \Big(1-\frac{1}{k}\Big)^{\vartheta} = \Big(1-\frac{b\ln k}{bk\ln k}\Big)^{bk\ln  k} \geq e^{-b\ln k }\Big(1- \frac{b\ln k}{k} \Big)\geq k^{-b}/2,
\] 
where the first inequality is by \cref{eq:i-not-occur}, and in the second inequality we have used the fact that $(1 + x/n)^n \geq e^x(1 - x^2/n)$ holds for all $n \geq  1, |x| \leq  n$ (see, e.g., \cite[Proposition B.3]{MotwaniR95}). Consequently, 
\[
    k^{1-b}/2 \leq \Ex{U(\vartheta)} \leq k^{1-b}.
\] 
Now, note that $U(\vartheta)=f(X_1, \dots, X_\vartheta) $ for some function $f:[k]^\vartheta \rightarrow \mathbb{R}$. Observe also that if two sequences of inputs to this function differ only in one place, then the outputs differ by at most one. Thus, $f$ satisfies Condition (L) of \cref{lem:McDiarmid}, which we apply with $\lambda = \varepsilon \Ex{U(\vartheta)} $ to derive
\[
    \Pr{ \big|U(\vartheta) - \Ex{U(\vartheta)} \big| > \varepsilon \Ex{U(\vartheta)} }  \leq 2\exp\left( - \frac{2 (\varepsilon k^{1-b}/2)^2 }{bk\ln k} \right) \leq 2\exp\left( -  \frac{ \varepsilon^2k^{1-2b} }{ \ln k} \right),
\] 
as claimed. 
\end{proof}

\begin{lemma}\label{clm:coveringleaves} 
    There exists a universal constant $c'>0$ such that, for any $C:=C(k) \geq 2$, and $k$ large enough, we have
    \[  
        \Pr{ \tfrac{1}{10}k\ln k \leq  T_{\cG}(K) \leq Ck\ln k } \geq 1- k^{-C/2} -4\exp\Big( -  \frac{c'  k^{4/5} }{ \ln k} \Big). 
    \]
\end{lemma}
\begin{proof}
For the upper bound, we follow the `coupon collector' strategy of simply visiting the current centre in each snapshot (and staying put if the centre has not changed).
Thus, applying \cref{lem:coupon} \ref{itm:couponup} with $a = C$, we obtain
\begin{equation}\label{eq:TK-upper}
    \Pr{ T_{\cG}(K) > Ck\ln k } \leq k^{-C/2}.
\end{equation}

For the lower bound, we fix $b=1/10$ and let $\vartheta = bk \ln k$. Our idea is to show that within the first $\vartheta$ many steps many vertices in $K$ have not been picked as centres, an event we call $\mathcal{E}$, and then within this time it is unlikely for an explorer to have visited all of these vertices by considering the visited and unvisited certres in $K$ as an RST model on a smaller set of stars. 

 Let $L'\subseteq K$ be the random set of vertices in $K$ that are leaves in every snapshot of $\cH:=\cG_{[1,\vartheta]}$. 
 Denote $\ell' := |L'|$, and
 let $\Lambda=k(1-1/k)^\vartheta$ and $\mathcal{E}:=\{ \big|\ell' -\Lambda\big| \leq  \Lambda/10\}$. Then, by \cref{lem:coupon} \ref{itm:couponlow}, we have
\begin{equation}\label{eq:proE}
	\Pr{\lnot \mathcal{E}} \leq   2\exp\bigg( -  \frac{ k^{1-2b} }{100 \ln k} \bigg).
\end{equation}
    Conditional on the vertices of $L'$ never being centres, $\cH$ picks a random centre in each step uniformly from $K \setminus L'$. Thus, conditional on this event, the first $\vartheta$ snapshots of $\cG$, restricted to the vertex set $[k]$, have the distribution of an RST model $\cG'\sim(\cS',\unif)$ on vertex set $[k]$, where $\cS'$ consists of $k-\ell' $ distinct $k$-vertex stars. Let $T_{\cG'}(L')$ be the time needed to visit the non-centres of $\mathcal{G}'$. Observe that $T_{\cG}(K)\geq T_{\cG}(L') $, since  $L' \subseteq K $. Thus, by the coupling above  
    \begin{equation}\label{eq:GandG'couple}
    	\Pr{T_{\cG}(K) <  bk\ln k} \leq  \Pr{T_{\cG}(L') <  bk\ln k , \;\mathcal{E}} + \Pr{\lnot \mathcal{E}} \leq  \Pr{T_{\cG'}(L') <  bk\ln k}  + \Pr{\lnot \mathcal{E}}  . 
    \end{equation} 
      We seek to apply \cref{clm:covingcentres} to $\cG'$ to bound $\Pr{T_{\cG'}(L') <  bk\ln k}$. In this setting,  the total number of vertices is $n'=k$ and the number of centre vertices is $k'=k-\ell'$. We denote the corresponding `birthday times' by $ \tau_i'$, which are defined in the same way as $\tau_i$, i.e., $\tau_i' = t_i' - t_{i-1}'$, where $(t_i')_{i \geq 0}$ is defined for $\cG'$ in the same way as $(t_i)_{i \geq 0}$ is defined for $\cG$. 
    Thus, by applying \cref{clm:covingcentres} to $\cG'$ and $L'$, for any $t \geq 0$, we have 
    \begin{equation}\label{eq:TL'}
        \Pr{T_{\cG'}(L') \leq \ell' \cdot \Ex{\tau_1'}  - t - k' -  1 } 
        \leq 
        2 \exp\left( \frac{-ct^2}{\ell' \cdot k'} \right)
        \leq
        2 \exp\left( \frac{-ct^2}{\ell' \cdot k} \right). 
    \end{equation}	
    Recall, by \cref{lem:coupon} \ref{itm:couponlow}, $k^{1-b}/2 \leq  \Lambda \leq k^{1-b} $, and, by \eqref{eq:Etau}, $\Ex{\tau_1'} \geq
    \sqrt{\frac{\pi }{2}k'} -\frac{2}{5} = \sqrt{\frac{\pi }{2}(k - \ell')}-\frac{2}{5}$. Thus,  if $\mathcal{E}$ holds,  then $\tfrac{9}{20} k^{1-b}\leq \ell' \leq  \tfrac{11}{10} k^{1-b} $, and so, for large $k$, we have 
   
    \[ 
        	\ell' \cdot \Ex{\tau_1'} \geq    \frac{9}{20} k^{1-b} \cdot \left(\sqrt{\frac{\pi}{2} \left(k -  \tfrac{11}{10} k^{1-b} \right)} - \frac{2}{5} \right) \geq  \frac{1}{2} k^{3/2-b}.   
    \]  
    Now, set $t=\frac{1}{10} k^{3/2-b}$.
    Recalling that $b=1/10$, we have $k^{3/2-b} \gg  k\ln k\gg k'$, and thus, for any sufficiently large $k$, we have 
    $bk\ln k < \frac{1}{2} k^{3/2-b} - t - k'-1 \leq \ell' \cdot \Ex{\tau_1'} - t - k' - 1$. Consequently, using \cref{eq:TL'}, we obtain
       \begin{align}\label{eq:T'G'}
            \Pr{T_{\cG'}(L') < bk\ln k  }
            &\leq \Pr{T_{\cG'}(L') < \frac{1}{2} k^{3/2-b} - t - k'-1  }\notag 
            \\
            &\leq
            \Pr{T_{\cG'}(L') \leq \ell' \cdot \Ex{\tau_1'}  - t - k' -  1 }\notag
            \\ 
            &\leq 2 \exp\left( -\frac{c (k^{3/2-b}/10)^2}{\tfrac{11}{10}k^{1-b}\cdot  k } \right) = 
            2 \exp\left( -\frac{c k^{1-b}}{110} \right).  
     \end{align}
    Inserting \eqref{eq:proE} and \eqref{eq:T'G'} into \eqref{eq:GandG'couple} gives 
    \begin{equation}\label{eq:TK-lower}     		
        \Pr{T_{\cG}(K) < \tfrac{1}{10} k\ln k  }
        \leq 
        2 \exp\left( -\frac{c k^{1-b}}{110} \right)
        +
        2\exp\left( -  \frac{k^{1-2b} }{100 \ln k} \right)  
        \leq 
        4\exp\left( -  \frac{c'k^{4/5} }{ \ln k} \right),
    \end{equation} 
    where $c'>0$ is some sufficiently small universal constant.
    Finally, the lemma follows by combining \cref{eq:TK-upper} and \cref{eq:TK-lower}.   
\end{proof}

We now have all the ingredients we need to prove \cref{thm:urst-some-stars}, the main result of this subsection.
We note that for fixed $k$, \cref{thm:urst-some-stars} does not guarantee a probability tending to one. However, for any fixed constant $k$, \cref{thm:orderm} implies that the exploration time of the corresponding RST is $O(n)$ with probability $1 - e^{-\Omega(n)}$.

\begin{restatable}{theorem}{arbstars}\label{thm:urst-some-stars}
   Let $\cS$ be a set of $k:=k(n)\geq 2$ distinct $n$-vertex stars, and let $\cG \sim (\cS,\unif)$.  Then, there exists a universal constant $d>0$ such that for any $C:=C(k) \geq 2$ we have   
    \begin{align*}
        &\Pr{ \max\Big\{ f_{-}(k),\; \tfrac{1}{10}k\ln k  \Big\}   \leq \texp(\cG) \leq f_+(k) +  Ck\ln k }   \geq 1- 3k^{-  dC} -  4\exp\Big( -  \frac{d k^{4/5} }{\ln k} \Big).
    \end{align*} 
    where 
    \[ 
        f_{\pm}(k) := 
        \Ex{\tau_1}(n-k)
        \left(
            1 \pm \sqrt{\tfrac{C\ln k }{n-k}} \pm 2\sqrt{\tfrac{k}{(n-k)^2}}
        \right).  
    \]
\end{restatable}

\begin{proof} Without loss of generality, assume that the set $K$ of centres of the stars in $\cS$ is $[k]$.

We will split the exploration into two tasks: exploring the non-centres in $L$, and exploring the centres in $K$. For these tasks the explorer may visit any vertices in the graph, but must visit all vertices in the target set. 
Since any exploration schedule must explore both $L$ and $K$, and the distributions of $T_{\cG}(L)$ and $T_{\cG}(K)$ do not depend on specific $\cG \sim (\cS,\unif)$, we have the  stochastic domination 
\begin{equation}\label{eq:texpbdds}
    \max\{ T_{\cG}(L), T_{\cG}(K)\} 
   \preceq
    \texp(\mathcal{G}) \preceq T_{\cG}(L) + T_{\cG}(K). \
\end{equation}

Thus, it remains to bound $T_{\cG}(L)$ and $T_{\cG}(K)$. 
Before doing so, observe that due to the presence of the non-explicit universal constant $d > 0$ in the statement of the result, we may assume that $k$ is at least some sufficiently large constant throughout.

For centres, by \cref{clm:coveringleaves}, there exists a constant $c' > 0$ such that
\begin{equation}\label{eq:Texp-TK}
    \Pr{ \tfrac{1}{10}k\ln k \leq  T_{\cG}(K) \leq Ck\ln k } \geq 1- k^{-C/2} -4\exp\Big( -  \frac{c'  k^{4/5} }{ \ln k} \Big). 
\end{equation}  

For non-centres, by \cref{clm:covingcentres}, there exists a constant $c > 0$ such that 
\begin{equation}\label{eq:Texp-TL-1}
    \Pr{
    \big| T_{\cG}(L) - (n-k) \cdot \Ex{\tau_1} \big| \leq t + k +1} \geq 1 - 2 \exp\left( \frac{-c t^2}{(n-k) k } \right).
\end{equation}
From \cref{eq:Etau}, for any sufficiently large $k$, we have 
$\Ex{\tau_1} \geq \sqrt{k/2}$.
Thus, for $t= \sqrt{ C(n-k)k \ln k}$,  
\begin{align*}
    (n-k) \cdot \Ex{\tau_1} - t - k - 1 
    &= 
    \Ex{\tau_1}(n-k) \left(1 - \sqrt{\tfrac{Ck\ln k }{\Ex{\tau_1}^2 (n-k)}} - \frac{k+1}{\Ex{\tau_1} (n-k)}\right) \\
    &\geq \Ex{\tau_1}(n-k)  \left(1 - \sqrt{\tfrac{C\ln k }{n-k}} - 2\sqrt{\tfrac{k}{(n-k)^2}} \right) = f_{-}(k), 
\end{align*}
and, similarly,
\begin{align*}
    (n-k) \cdot \Ex{\tau_1} + t + k + 1 
    &\leq \Ex{\tau_1}(n-k)  \left(1 + \sqrt{\tfrac{C\ln k }{n-k}} + 2\sqrt{\tfrac{k}{(n-k)^2}} \right) = f_{+}(k).
\end{align*}
Consequently, from \cref{eq:Texp-TL-1}, we derive
\begin{equation}\label{eq:Texp-TL}
    \Pr{f_{-}(k) \leq T_{\cG}(L) \leq f_{+}(k)} \geq 1 - 2 \exp\left( \frac{-c t^2}{(n-k) k } \right) = 1 - 2k^{-c\cdot C}. 
\end{equation}

Now, by setting $d := \min\{1/2, c, c'\}$, from \cref{eq:texpbdds}, \cref{eq:Texp-TK}, and \cref{eq:Texp-TL}, we obtain
\begin{align*}
    &\Pr{ \max\Big\{ f_{-}(k),\; \tfrac{1}{10}k\ln k  \Big\}   \leq \texp(\cG) \leq f_+(k) +  Ck\ln k }   \geq 1- 3k^{-dC} -  4\exp\Big( -  \frac{d k^{4/5} }{ \ln k} \Big),
\end{align*} 
as desired.
\end{proof}

\cref{thm:weak-stars} now follows readily from \cref{thm:urst-some-stars}. 

\begin{proof}[Proof of \cref{thm:weak-stars}]
Let  $w \leq  k  \leq n - w \sqrt{n}\ln n$, where we will be taking $w$ to be a suitably large constant. Then,  by \eqref{eq:Etau} we have 
\begin{equation}\label{eq:cheap1}
  \big(1 - 2w^{-1} \big)\cdot\sqrt{\frac{\pi k }{2}    }\leq   \Ex{\tau_1} \leq  \big(1 + 2w^{-1} \big)\cdot\sqrt{\frac{\pi k }{2}   }.  \end{equation}  

We will apply \cref{thm:urst-some-stars} with $C(k) = 2 $, note that in this case, for large $w$,   
\begin{equation}\label{eq:cheap2}
 \sqrt{\frac{C\ln  k}{n-k}}  \leq  w^{-1}  , \quad   2\sqrt{\frac{k}{(n-k)^2}} \leq w^{-1/2} \quad \text{and} \quad 3k^{-  dC}  + 4\exp\Big( -  \frac{d k^{4/5} }{\ln k} \Big) \leq w^{-d} .\end{equation} Recalling that $T= (n-k)\sqrt{\tfrac{\pi  }{2}k}$, it follows from \eqref{eq:cheap1} and \eqref{eq:cheap2}   that for large $w$ 
 \begin{equation*}   f_{+}(k)  \leq  \big(1 +2w^{-1} \big)(1+w^{-1} + w^{-1/2})\cdot T \leq \big(1 +8w^{-1/2} \big)\cdot T,      \end{equation*} and similarly $f_-(k) \geq (1-8w^{-1/2})\cdot T $. It follows that, 
  \begin{equation}\label{eq:convergebounds}    \left| f_{\pm}(k)  -T  \right| \leq 8w^{-1/2} \cdot T, \qquad \tfrac{1}{10}k\ln k  <  f_{-}(k), \qquad \text{and} \qquad  Ck\ln k \leq 2w^{-1} \cdot T.   \end{equation}Thus, using  \eqref{eq:convergebounds} and \eqref{eq:cheap2} in  \cref{thm:urst-some-stars}   gives 
 \[\Pr{  \left| \texp(\cG)  -T  \right|\leq    10w^{-1/2}  \cdot  T} \geq 1 - w^{-d}.  \] Thus if we choose $w$ suitably large satisfying $w\geq \max\{100/\varepsilon^2 , \varepsilon^{-1/d} \}  $ then the result follows. 
\end{proof}

\cref{cor:urst-half-stars} follows from \cref{thm:urst-some-stars} by taking $k = k(n)=\lceil \frac{n}{3} \rceil$ and $C(k) = \log 2k \approx \log n$.  Another corollary is obtained by taking $k = k(n) = n$ and $C=2$.

\begin{corollary}\label{cor:urst-all-stars}
    There exists a constant $c > 0$ such that, for any sufficiently large $n \in \bN$, if
    $\cS$ is a set of $n$ distinct  $n$-vertex  stars, and $\cG \sim (\cS,\unif)$, then
    \[
        \Pr{ \tfrac{1}{10}k\ln k  \leq \texp(\cG) \leq  3k\ln k } \geq 1- n^{-c}.
    \] 
\end{corollary}

\subsection{Exploration lower bound based on degree}\label{sec:lowerdeg}

The following theorem is a lower bound on the worst case exploration time of temporal graphs which have snapshots (and even underlying graphs) of bounded degree. Similarly to the previous subsection, this is a randomisation of a deterministic example studied in \cite[Proposition 3.3]{ErlebachOnTemporal2021}. 

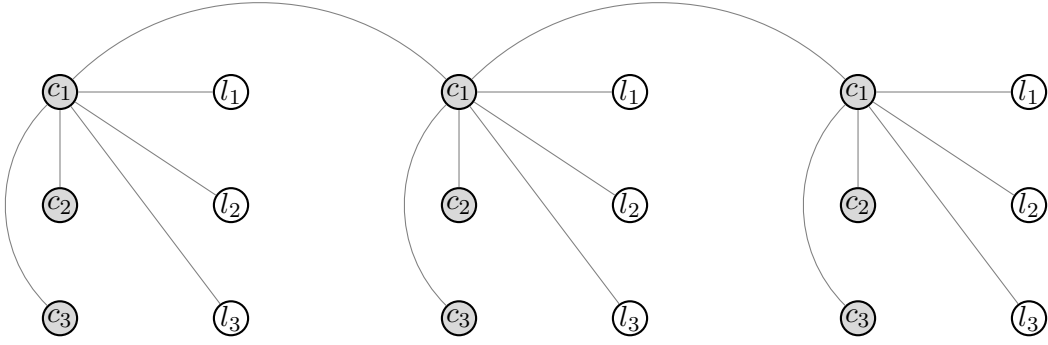
\begin{figure}
				\centering
				\begin{tikzpicture}[scale=\snapshotscale]
					
					\foreach \vertex/\name/\x/\y in 
					{c11/$c_1$/0/2, c2/$c_2$/0/0, c3/$c_3$/0/-2} 
					{
						\node[knoten, fill=gray!30] (\vertex) at (\x, \y) {\name};
					}

                    \foreach \vertex/\name/\x/\y in 
					{l1/$l_1$/3/2, l2/$l_2$/3/0, l3/$l_3$/3/-2} 
					{
						\node[knoten] (\vertex) at (\x, \y) {\name};
					}
					
					\foreach \from/\to in 
					{c11/c2, c11/l1, c11/l2, c11/l3} 
					{
						\draw[edge] (\from) -- (\to);
					}
                    \draw[edge] (c11) to[out=-135,in=135] (c3);
                    
			\renewcommand{\xshift}{7}
					
				\foreach \vertex/\name/\x/\y in 
					{c12/$c_1$/0/2, c2/$c_2$/0/0, c3/$c_3$/0/-2 } 
				{
					\path (\x,\y) ++(\xshift,0) node[knoten, fill=gray!30] (\vertex) {\name};
				}

                \foreach \vertex/\name/\x/\y in 
					{ l1/$l_1$/3/2, l2/$l_2$/3/0, l3/$l_3$/3/-2} 
				{
					\path (\x,\y) ++(\xshift,0) node[knoten] (\vertex) {\name};
				}
					
					\foreach \from/\to in 
					{c12/c2, c12/l1, c12/l2, c12/l3} 
					{
						\draw[edge] (\from) -- (\to);
					}
                    \draw[edge] (c12) to[out=-135,in=135] (c3);
			
			\renewcommand{\xshift}{14}
					
				\foreach \vertex/\name/\x/\y in 
					{c13/$c_1$/0/2, c2/$c_2$/0/0, c3/$c_3$/0/-2} 
				{
					\path (\x,\y) ++(\xshift,0) node[knoten, fill=gray!30] (\vertex) {\name};
				}
                \foreach \vertex/\name/\x/\y in 
					{ l1/$l_1$/3/2, l2/$l_2$/3/0, l3/$l_3$/3/-2} 
				{
					\path (\x,\y) ++(\xshift,0) node[knoten] (\vertex) {\name};
				}
					
					\foreach \from/\to in 
					{c13/c2, c13/l1, c13/l2, c13/l3} 
					{
						\draw[edge] (\from) -- (\to);
					}
                    \draw[edge] (c13) to[out=-135,in=135] (c3);
                    \draw[edge] (c11) to[out=45,in=135] (c12);
                    \draw[edge] (c12) to[out=45,in=135] (c13);
				\end{tikzpicture}
				
				\caption{An example of a tree from $\cT$ constructed in \Cref{thm:bounded-degree-lower} with $d=7$ and $n=18$. 
                Centres are indicated with gray, and non-centres with white.}
                \label{fig:degree-lower}
			\end{figure}

\begin{restatable}{theorem}{bdddeglower}
    \label{thm:bounded-degree-lower}
    For each odd $d \geq 5$ and $n = k(d - 1)$, where $k \in \mathbb{N}$, there exists a set $\mathcal{T}$ of $n$-vertex trees with maximum degree at most $d$ such that, for $\mathcal{G} \sim (\mathcal{T}, \unif)$, we have
    \[
        \Pr{\texp(\mathcal{G}) > \frac{\sqrt{d} \cdot  n}{50} } \geq 1- e^{-\Omega(n)}. 
    \]  
\end{restatable}
\begin{proof}
    Let $d$ and $n$ be as in the statement of the lemma.
    We define the following set $\cT$ of trees. 
    Let $\cS$ be a set of $\frac{d-1}{2}$ many stars over vertex set $[d-1]$ with centres in $K:=\left[ \frac{d-1}{2} \right]$. For $i \in K$, denote by $R_i$ the star in $\mathcal{S}$ with centre $i$. For each $R_i \in \cS$, we include in $\cT$ exactly one tree on vertex set $[n]$, denoted $T_i$, which is obtained by creating $n/(d-1)$ vertex-disjoint copies of $R_i$ and connecting their centres in a path such that the centre of the $j$-th copy is connected to the centre of the $(j+1)$-th copy for $1 \leq j < n/(d-1)$. See \Cref{fig:degree-lower} for an example of a tree in $\cT$.
    
    We assume that the vertex set of the $j$-th copy of any star in $\mathcal{S}$ is $[(j-1)(d-1)+1,\, j(d-1)]$, where the first half of this interval, denoted $K_j$, corresponds to the centre vertices, and the second half, denoted $L_j$, corresponds to the non-centre vertices.

    Let $\mathcal{G} \sim (\mathcal{T}, \unif)$, and let $W$ be a temporal walk that visits all vertices in $\mathcal{G}$.  
    Define the following sets
    \[
        M_j := 
        \begin{cases}
            L_j \cup K_{j+1}, & \text{if } j = 1, \\
            K_{j-1} \cup L_j, & \text{if } j = n / (d - 1), \\
            K_{j-1} \cup L_j \cup K_{j+1}, & \text{if } 1 < j < n/(d-1).
        \end{cases}
    \]
    We interpret each set $M_j$ as the set of non-centres of generalised stars on vertex set $M_j \cup K_j$ whose centres lie in $K_j$.
    Note that any temporal walk between a vertex in $K_j \cup L_j$ and a vertex outside of $K_j \cup L_j$ must pass through a vertex in $K_{j-1}$ or $K_{j+1}$, that is, a vertex in $M_j \setminus L_j$.
    This implies that we can partition $W$ into temporal subwalks, each of which, for some $j$, either connects two vertices within $L_j$, or connects one vertex in $L_j$ to another in $M_j \setminus L_j$, with all internal vertices lying in~$K_j$.

    Each such subwalk can be viewed as a temporal walk between two non-centres in a temporal graph whose snapshots are generalised stars on vertex set $M_j \cup K_j$ with centres in $K_j$ and non-centres in $M_j$. 
    Therefore, by \cref{lem:pathsthroughC}, there exists a snapshot that appears at two distinct time steps between the times of the first and last time-edges of the  subwalk.
 
    Associate with $\mathcal{G}$ a sequence $\mathcal{X}(\mathcal{G}) = (X_i)_{i \geq 1}$, where $X_i = j \in \left[ \frac{d-1}{2} \right]$ indicates that the $i$-th snapshot of $\mathcal{G}$ is $T_j$.  
    Define the sequence $(t_i)_{i \geq 0}$ as in \cref{eq:t_i}, and let $\tau_i := t_i - t_{i-1}$ for $i \geq 1$.
    The above discussion (cf. \cref{lem:leaves-exp}) implies that all $n/2$ vertices in $L = \bigcup_{i=1}^{n/(d-1)} L_i$ cannot be visited earlier than time $t_{\ell-1}$, where $\ell = |L| = n/2$.

    Recalling that $d\geq 5$ and $n\geq 4$, and using \eqref{eq:Etau}, we have 
    \[
        (\ell-1) \cdot \Ex{\tau_1}\geq \Big(\frac{n}{2}-1\Big) \cdot \left( \sqrt{\frac{\pi}{2} \cdot \frac{d-1}{2}    }   - \frac{2}{5}\right) \geq     \frac{n}{4}\cdot \frac{\sqrt{d}}{10},
    \] 
    and thus, by \cref{lem:birthday} with $t=n\sqrt{d}/200$,
    \[
        \Pr{t_{\ell-1} \leq \frac{\sqrt{d} \cdot n }{50} } 
        \leq  
        \Pr{ t_{\ell-1} \leq    (\ell-1) \cdot \Ex{\tau_1} - t   } \leq  2 \exp\left( \frac{-c d n^2/200^2}{(\tfrac{n}{2}-1)\cdot \frac{d-1}{2} } \right) = e^{- \Omega(n) },
    \]
    which implies the desired result.
\end{proof}

 \subsection{Lower Bound for Online Exploration}\label{sec:onlinestars}

In this section, we show that any online exploration algorithm requires $\Omega(n^2)$ time steps to explore the RST $(\cS, \unif)$ with $\lceil n/2 \rceil$ stars. This demonstrates that knowledge of future snapshots is essential for the upper bound in \cref{cor:urst-half-stars}. For convenience, we assume $n$ is even, but the result also holds for odd values of $n$.

\begin{proposition}\label{prop:online-stars} 
    Let $n \in \bN$ be even and $n \geq 4$.
    Let $\cS$ be a set of $n/2$ distinct $n$-vertex stars, and $\cG \sim (\cS,\unif)$.
    Then, any randomised online exploration algorithm $A$ of $\mathcal{G} $ satisfies 
    \[ 
        \Pr{\Dexp_A(\cG) > \frac{n^2}{15} } \geq 1 - \exp\Big(- \frac{n}{100}\Big).  
    \] 
\end{proposition}
\begin{proof}
Let $k=n/2$.
Let $K := [k]$ be the centres of the stars in $\cS$,
and $L := [n] \setminus K$ to be the non-centres.

We represent the graph $\mathcal{G}$ by a sequence $(c_t)_{t \geq 1}$, where each $c_t \in K$ denotes the centre of the star in the $t$-th snapshot of $\mathcal{G}$.  
Furthermore, we interpret a randomised online exploration algorithm $A$ as an explorer moving on the vertices of $\mathcal{G}$, and represent the explorer's walk by a sequence $(a_t)_{t \geq 1}$, where $a_1$ is the initial position (chosen by the algorithm without seeing any snapshots of $\cG$), and, for each $t \geq 2$, the value $a_t \in [n]$ denotes the vertex to which the explorer moves at the end of time step $t-1$.  
For every $t \geq 1$, the decision $a_t$ at time step $t-1$ is based on the explorer’s movement history $(a_1, a_2, \ldots, a_{t-1})$, the revealed snapshots of $\mathcal{G}$ up to time $t-1$, represented by $(c_1, c_2, \ldots, c_{t-1})$, and the random bits used by the algorithm.

Since no two vertices in $L$ are connected by an edge in any snapshot of $\mathcal{G}$, the only way the explorer can visit a vertex in $L$ at time step $t$ is by being located at the star center of the current snapshot at the start of time step $t$. That is, the explorer must have moved to a vertex $v \in K$ at the end of time step $t-1$, and the $t$-th sampled snapshot of $\mathcal{G}$ must be the star centered at $v$. This condition is equivalent to the equality $a_t = c_t$.

Thus, if the explorer visits all vertices of $\cG$ by time $T$, then
$c(A,T) = \sum_{t=1}^{T} \mathbf{1}(a_t=b_t)$ is at least $|L|-1 = n/2-1$ (where the "$-1$" accounts for the possibility that the explorer's initial position lies in $L$).
This and the fact that $c(A,T)$ is stochastically dominated by $Z \sim \operatorname{Bin}(T,1/k) = \operatorname{Bin}(T,2/n)$ (by \cref{Prop:JimmysPrinciple}) imply
\[
    \Pr{\Dexp_A(\cG) \leq T} 
    \leq 
    \Pr{c(A,T) \geq n/2-1}
    \leq 
    \Pr{Z \geq n/2-1}.
\]

To complete the proof, we choose $T = n(n - 2)/6$ and $\delta = 1/2$ so that $(1 + \delta)\cdot \frac{2T}{n} = n/2 - 1$. Then, since $n^2/15 \leq T$ for $n \geq 4$, applying \cref{lem:simplechb} yields 
\begin{align*}
    \Pr{\Dexp_A(\cG) \leq n^2/15} 
    &\leq
    \Pr{\Dexp_A(\cG) \leq T}
    \leq 
    \Pr{Z \geq (1+\delta)2T/n} \\ &\leq \exp\left(\frac{-(1/2)^2}{3}\cdot \frac{2T}{n}   \right) = \exp\left( -  \frac{n-2}{36}   \right) \leq 
    \exp\left( -  \frac{n}{100}   \right), 
\end{align*}
which implies the claim. 
\end{proof}

\iftoggle{anonymous}{
}{%
	\bigskip
	\textbf{Acknowledgments.}
We thank Thomas Erlebach for some interesting and enlightening discussions, in particular at the Algorithmic Aspects of Temporal Graphs VII workshop. We also thank Patrick Totzke for raising a nice question on the power of online exploration algorithms, which we answered in \cref{prop:online-stars}. This work was supported
by a Royal Society International Exchanges grant IES$\backslash$R1$\backslash$231140.
}

\bibliographystyle{abbrv}
 \bibliography{bibliography}

\appendix

\section{Concentration and Auxillary Results}\label{sec:furtherprelim}

\paragraph{Concentration inequalities.} A random variable $X$ is \emph{binomially distributed} with parameters $n\geq 1$ and $p\in [0,1]$, denoted $X\sim \operatorname{Bin}(n,p)$, if $\mathbb{P}( X = k ) = \binom{n}{k}p^k(1-p)^{n-k} $. We will make use of the following classic Chernoff bound for binomial random variables.  

\begin{lemma}[{\cite[Theorems 4.4 \& 4.5]{mitzenmacher2017probability}}]\label{lem:simplechb}
    Let $X\sim\operatorname{Bin}(n,p)$, where $n\in \mathbb{N}_{+} $ and $p\in (0,1)$. Let $\mu := \bE[X]$. Then, for any $0<\delta <1$, \[\mathbb{P}[X \leq (1-\delta) \mu ]\leq e^{-\delta^{2}\mu /2}, \qquad \text{and}\qquad \mathbb{P}[X \geq (1+\delta) \mu ]\leq e^{-\delta^{2}\mu /3}.\]
\end{lemma}

A random variable $X$ is \textit{geometrically distributed} with parameter $p\in [0,1]$, denoted $X\sim \geo{p}$, if $\Pr{X=k} = (1-p)^{k-1}p $ for any integer $k\geq 1$.

We will use the following Chernoff bound for  sums of geometric random variables.

\begin{lemma}[{\cite[Theorem 2.1]{JansonTail}}] \label{lem:jansontail}
    For any $n\geq 1$ and $p_1, \dots, p_n \in (0,1]$, let $X=\sum_{i=1}^nX_i$, where $X_i \sim \geo{p_i}$. Let $p_*=\min_{i\in [n]}p_i$ and $\mu=\Ex{X}=\sum_{i=1}^n\frac{1}{p_i}$. Then, for any $\lambda\geq 1$, \[\Pr{X \geq \lambda \mu  } \leq \exp( - p_* \cdot \mu \cdot \left(\lambda -1- \log \lambda\right)).\]
\end{lemma}

We will also use McDiarmid's inequality. 
\begin{lemma}[{\cite[Corollary 2.27 \& Remark 2.28]{JansonBook}}]\label{lem:McDiarmid} Let $X_1, \dots, X_\ell$ be independent random variables, with $X_i $ taking values in a set $\Lambda_i$. Assume that a function $f:\Lambda_1 \times \cdots \Lambda_\ell \rightarrow \mathbb{R}$
satisfies the following Lipschitz condition for some numbers $c_i$, $i \in [\ell]$: 
\begin{description}
    \item[(L)] If two vectors $x,x' \in \prod_{i=1}^\ell \Lambda_{i=1}^\ell$ differ only in the $i$-th coordinate, then $|f(x)-f(x')|\leq c_i$.   
\end{description} 
Then, for any $\lambda > 0$, the random variable $Y = f(X_1,\dots,X_\ell)$   
satisfies 
\[\Pr{|Y - \Ex{Y}|> \lambda} \leq 2\exp\Bigg(- \frac{2\lambda^2}{\sum_{i=1}^\ell c_i^2}  \Bigg).   \]  

\end{lemma}

We will also use a concentration result for sums of sub-Gaussian  random variables \cite{Vershynin}. The sub-Gaussian norm $||X||_{\psi^2}$ of a random variable $X$ is  defined as follows,     
\begin{equation} \label{eq:subgaussnorm}
    ||X||_{\psi^2}:=\inf\{K >0 : \Ex{\exp(X^2/K^2) } \leq 2\}.
\end{equation}
We say that a random variable $X$ is sub-Gaussian if $||X||_{\psi^2}$ is bounded. 
We will use a variant of Hoeffding's inequality for sub-Gaussian random variables.

\begin{theorem}[{\cite[Theorem 2.7.3]{Vershynin}}]\label{thm:subGHoff}Let $X_1,\dots , X_N$ be independent sub-Gaussian random variables with zero mean. Then, there exists a universal constant $c>0$ such that for every $t \geq  0$
\[ 
    \Pr{\left|\sum_{i=1}^{N}X_i \right| \geq t   } \leq 2 \exp\left( \frac{-c t^2}{\sum_{i=1}^{N}||X_i||^2_{\psi^2} } \right). 
\]  
\end{theorem}

\paragraph{Lower Bound for Randomised Online Decisions.} The following proposition says that if are playing a game where you get to open one of $k$ boxes in each round, where a single uniformly sampled box contains a gold bar, then you cannot do better than just opening the first box each time. 
We need this for our lower bounds for online explorers. The proof of this is very basic, but we could not find it anywhere to cite, so we include it for completeness.

\begin{proposition}\label{Prop:JimmysPrinciple} Let $(b)_{t\geq 1}$ be i.i.d.~with uniform distribution over $[k]$. Let $S\subseteq \mathbb{N}$ and $A$ be any randomised algorithm which in step $t \geq 1$ has access to $(a_i,b_i)_{i=1}^{t-1}$ and outputs $a_t\in S$. Then $c(A,T) = \sum_{t=1}^{T} \mathbf{1}(a_t=b_t)$ is stochastically dominated by a random variable with distribution   $\operatorname{Bin}(T,1/k)$. Furthermore, if $S=[k]$ then $c(A,T)$ has distribution $\operatorname{Bin}(T,1/k)$.       
\end{proposition}

\begin{proof} We begin by showing that for any algorithm $A$ taking values in a set $S $, there exists an algorithm $A'$ taking values in $[k]$ such that $c(A,T) \preceq c(A',T)$. We can consider each $a_t$ as a function  $a_t:((a_i,b_i)_{i=1}^{t-1},\omega_t)\mapsto S$, where $\omega_t\in\{0,1\}^{\mathbb{N}}$ is an independent $0/1$-random string. We define a new algorithm $A'$ by setting each output $a_t$ as follows 
\[
    a'_t =\begin{cases}
		a_t &\text{ if } a_t\in [k],\\
		1 &\text{ if } a_t\not\in [k].\\
	\end{cases}
\]    
Observe that if $a_t \in S\backslash [k]$ then, regardless of the value of $b_t$, we have $\mathbf{1}(a_t=b_t) =0$. It then follows that for every sequence $\big((a_i,b_i)_{i=1}^{t-1},\omega_t\big)_{t\in [T]} $  we have 
\[c(A,T)    =\sum_{t=1}^{T} \mathbf{1}(a_t=b_t, a_t \in [k])  = \sum_{t=1}^{T} \mathbf{1}(a'_t=b_t, a_t \in [k])  \leq c(A',T),    \]   
giving $c(A,T) \preceq c(A',T)$ as claimed. Thus, we can assume from now on that $S=[k]$. 

Now, for any such algorithm $A$ taking values in $[k]$, since $b_t$ is independent of $(a_i)_{i\in [t]}$ we have \[\Pr{a_t=b_t} = \sum_{c \in [k]} \Pr{a_t= b_t = c}= \sum_{c \in [k]} \Pr{b_t = c \mid a_t=c}\Pr{ a_t=c} = \sum_{c \in [k]} \frac{1}{k}\Pr{ a_t=c} = \frac{1}{k}.  \] Thus, for each $t\geq 0$, the random variable $X_t= \mathbf{1}(a_t=b_t) \sim \ber{1/k}$ is a Bernoulli random variable with probability $1/k$. 
    
    We now show that, for any $T\geq 1$, the random variables $\{X_t\}_{t \in [T]}$ are mutually independent. For $t\in [T]$,  $\mathbf{\alpha}   \in \{0,1\}^{t-1}$,  and $Q \subseteq [t-1]$, we define the event   $\mathcal{E}(Q,\mathbf{\alpha}) = \bigcap_{q\in Q}\{X_q=\alpha_q\}$.  Then,   
	\begin{align*}
		  \Pr{X_t =1 \mid \mathcal{E}(Q,\mathbf{\alpha})} &= \sum_{c \in [k]} \Pr{a_t= b_t = c  \mid \mathcal{E}(Q,\mathbf{\alpha})}\\
		&= \sum_{c \in [k]} \Pr{b_t = c   \mid a_t=c, \mathcal{E}(Q,\mathbf{\alpha})}\Pr{a_t= c  \mid     \mathcal{E}(Q,\mathbf{\alpha})}.
        \end{align*} 
        Now, since $b_t$ is independent of the sigma algebra generated by $\{ (a_i)_{i\in[t]}, (b_i)_{i\in [t-1]}\}$, which contains all events of the form $\{a_t=c\} \cap \mathcal{E}(Q,\alpha)$, and  $\Pr{a_t= \cdot   \mid \mathcal{E}(Q,\mathbf{\alpha})}$ is a probability distribution on~$[k]$, 
	\begin{equation*}
		  \Pr{X_t =1 \mid \mathcal{E}(Q,\mathbf{\alpha})}	= \sum_{c  \in [k]} \frac{1}{k} \cdot \Pr{a_t= c  \mid \mathcal{E}(Q,\mathbf{\alpha})} 	=  \frac{1}{k} = \Pr{X_t =1} , 
	\end{equation*}
    and similarly $\Pr{X_t =0 \mid \mathcal{E}(Q,\mathbf{\alpha})} = \Pr{X_t =0} $.  Thus, for any  $\alpha \in \{0,1\}^T$ and subset $I\subseteq [T]$ with largest element $i^*$ we have    
    \[ 
\begin{aligned}
\Pr{\cap_{i \in I}\{X_i=\alpha_i\} } &=  \Pr{X_{i^*}=\alpha_{i^*}\mid \cap_{i \in I\setminus \{i^*\}}\{X_i=\alpha_i\} } \cdot  \Pr{\cap_{i \in I\setminus \{i^*\}}\{X_i=\alpha_i\} } \\  
&  =  \Pr{X_{i^*}=\alpha_{i^*}  } \cdot  \Pr{\cap_{i \in I\setminus \{i^*\}}\{X_i=\alpha_i\} } =  \prod_{i \in I}\Pr{ X_i=\alpha_i  },  \end{aligned}\]by induction. Thus $c(A,T) = \sum_{t=1}^{T} X_t$ follows the Binomial distribution $\operatorname{Bin}(T,1/k)$.  
\end{proof}

\end{document}